\documentclass[a4paper,onecolumn,amsmath,amssymb,accepted=2025-06-10]{quantumarticle}
\pdfoutput=1

\usepackage{graphicx} 
\usepackage{amsmath}
\usepackage{amsthm}
\usepackage{amssymb}
\usepackage{mathrsfs}
\usepackage{xcolor}
\usepackage{soul}
\usepackage{comment}
\usepackage[colorlinks]{hyperref}
\usepackage{theoremref}

\makeatletter
\newtheorem{lem}{Lemma}

\newtheorem{defn}{Definition}
\newtheorem{cor}{Corollary}
\newtheorem{thm}{Theorem}

\DeclareMathOperator*{\argmin}{arg\,min}

\usepackage{algorithm}
\usepackage{algpseudocode}

\begin{document}
\title{Stabilizer ground states for simulating quantum many-body physics: theory, algorithms, and applications}
\author{Jiace Sun}
\affiliation{Division of Chemistry and Chemical Engineering, California Institute of Technology, Pasadena, CA 91125, USA}
\email{jsun3@caltech.edu.}
\author{Lixue Cheng}
\affiliation{Division of Chemistry and Chemical Engineering, California Institute of Technology, Pasadena, CA 91125, USA}
\affiliation{Department of Chemistry, The Hong Kong University of Science and Technology, Clear Water Bay, Kowloon, Hong Kong (SAR) 999077, China}
\author{Shi-Xin Zhang}
\affiliation{Institute of Physics, Chinese Academy of Sciences, Beijing 100190, China}

\begin{abstract}
Stabilizer states, which are also known as the Clifford states, have been commonly utilized in quantum information, quantum error correction, and quantum circuit simulation due to their simple mathematical structure.
In this work, we apply stabilizer states to tackle quantum many-body ground state problems and introduce the concept of stabilizer ground states.
We establish an equivalence formalism for identifying stabilizer ground states of general Pauli Hamiltonians.
Moreover, we develop an exact and linear-scaled algorithm to obtain stabilizer ground states of 1D local Hamiltonians and thus free from discrete optimization.
This proposed equivalence formalism and linear-scaled algorithm are not only applicable to finite-size systems, but also adaptable to infinite periodic systems.
The scalability and efficiency of the algorithms are numerically benchmarked on different Hamiltonians.
Finally, we demonstrate that stabilizer ground states are promising tools for not only qualitative understanding of quantum systems, but also cornerstones of more advanced classical or quantum algorithms.
\end{abstract}

\maketitle

\section{Introduction}

Discovering the ground state of quantum Hamiltonians remains a central challenge in quantum many-body physics. The difficulty stems from the exponential scaling of the Hilbert space dimension, rendering exact calculations impractical in most cases \cite{Kitaev2002a_z, Aharonov2002_z, Kempe2006_z, Schuch2009a_z, Huang2020c_z, lee2023evaluating}. Consequently, a variety of classical and quantum algorithms have been developed to address this problem. Two of the most commonly employed strategies for ground state determination are variational methods and approaches based on imaginary time evolution. The variational approach, represented by techniques such as variational quantum Monte Carlo (VQMC) \cite{Carleo2017_z, Deng2017c_z, Glasser2018b_z, Zhang2019b_z}, density matrix renormalization group (DMRG) \cite{Verstraete2008a_z, Schollwock2011_z, schollwock2005density}, and variational quantum eigensolver (VQE) \cite{kandala2017hardware, cerezo2022variational, tilly2022variational}, constructs a parameterized quantum state ansatz $|\psi(\boldsymbol{\theta})\rangle$ and minimizes the energy with respect to the parameters $\boldsymbol{\theta}$. In contrast, imaginary time evolution methods, including auxiliary-field quantum Monte Carlo (AFQMC) \cite{motta2018ab, lee2022twenty, carlson2011auxiliary}, time-evolving block decimation (TEBD) \cite{schollwock2013matrix, orus2008infinite, haegeman2016unifying}, and quantum imaginary-time evolution (QITE) \cite{motta2020determining, mcardle2019variational}, iteratively apply $e^{-\Delta\beta H}$ to an initial state $|\psi_{0}\rangle$, converging toward the ground state. In both categories, the choice of initial state is crucial. A well-chosen initial state can mitigate non-convex optimization challenges in variational methods and improve the overlap with the true ground state $|\langle\psi_{0}|\psi_{\text{gs}}\rangle|$, which directly impacts the efficiency of imaginary time evolution. The importance of the initial state is equally apparent in other algorithms like Krylov subspace methods \cite{dargel2011adaptive, cortes2022quantum} and quantum phase estimation (QPE) \cite{kitaev1995quantum, kitaev2002classical}.

In most scenarios, mean-field states \cite{chaikin1995principles, kadanoff2009more} are commonly chosen as initial guesses. These states not only simplify computations but also often capture essential aspects of a system's physics, making them useful for qualitative insights \cite{negele1982mean, lykos1963discussion}. However, mean-field approximations can be overly simplistic for general many-body systems, and they frequently fail to describe more complex phenomena, such as those found in topological systems \cite{kitaev2006anyons, chen2012symmetry} and multireference systems \cite{lyakh2012multireference, lischka2018multireference}.  Therefore, alternative choices of quantum states with both mathematical simplicity and physical expressivity are highly demanded.

During the development of quantum computing \cite{nielsen2001quantum,gottesman1998theory} in recent decades, a new type of states called stabilizer states has gained significant attention. From one perspective, it is the set of quantum states ``stabilized'' by a maximum number of Pauli operators, which implies a polynomial-sized classical description \cite{gottesman1998heisenberg,aaronson2004improved}. From another perspective, it is the set of quantum states reachable solely by Clifford operations, i.e., the combination of CNOT, Hadamard, and phase gates, which are the ``easy-to-implement'' gates in the fault-tolerant quantum computation (FTQC) era \cite{campbell2017roads,howard2017application,chamberland2019fault}. Compared with product states, stabilizer states are able to not only capture long-range area-law entanglement that is significant for understanding topological order and symmetry-protected topological states \cite{Zeng2019_z}, but also support volume-law entanglement \cite{perez2006matrix}, which is a feature lacking in other ansatzes, such as matrix product states (MPS) \cite{fattal2004entanglement}. Thanks to these features, Clifford operations and stabilizer states are utilized as important tools in the explorations of quantum information \cite{Webb2016_z,Huang2020b_z}, quantum dynamics \cite{Nahum2017_z,VonKeyserlingk2018a_z,Nahum2018a_z}, quantum error correction \cite{gottesman1997stabilizer,fowler2012surface}, topological quantum computing \cite{Nayak2008_z}, quantum circuit simulation \cite{bravyi2019simulation,bravyi2016improved,beguvsic2023simulating}, and quantum-classical hybrid algorithms \cite{cheng2022clifford,Zhang2021d_z,sun2024toward,mishmash2023hierarchical,schleich2023partitioning}. 

With the polynomial-sized classical description and efficient implementation on quantum computers, stabilizer states naturally become a promising candidate for quantum initial states, facilitating the construction of advanced algorithms or quantum state ansatzes.
Furthermore, one can even take advantage of both stabilizer states and mean-field states by moving to the Heisenberg picture and treating one of them as a Hamiltonian transformation (i.e. Clifford transformation or basis rotation).
In fact, stabilizer states have been used to enhance the power of VQE \cite{cheng2022clifford,CAFQA,ising_stab_gs}, DMRG \cite{qian2024augmenting}, tensor network \cite{masot2024stabilizer}, and quantum Monte Carlo \cite{jeevanesan2024quantum}.
However, the challenge lies in the absence of scalable and general algorithms for identifying the appropriate stabilizer initial state, typically the one with the minimum energy for a given Hamiltonian, which is defined as the stabilizer ground state in this work. Although several optimization-based methods have been proposed \cite{CAFQA,mitarai2022quadratic}, they are not guaranteed to reach the real minimum and are not scalable to large systems due to the intrinsic difficulty of discrete optimization and the $O(2^{(n+1)(n+2)/2})$ scaling of the number of $n$-qubit stabilizer states \cite{number_stab_groups}.

In this work, we provide a series of algorithms to identify the stabilizer ground state of different types of Hamiltonians. We start by theoretically establishing the equivalence between the stabilizer ground states and the closed maximum-commuting Pauli subsets (CMCS) for general Pauli Hamiltonians. We further present an exact and linear-scaled algorithm to find the stabilizer ground states of 1D local Hamiltonians. For properly defined sparse Hamiltonians, this exact 1D local algorithm is proved to be computationally efficient with a scaling of $O(n\exp(Ck\log k))$ with some constant $C$, where $n$ is the number of qubits and $k$ represents the locality. Additionally, we prove that stabilizer ground states of 1D local Hamiltonians can be prepared on quantum computers with fewer operations than general stabilizer states. Furthermore, we present that both the equivalence formalism for general Hamiltonians and the linear-scaled algorithm for 1D local Hamiltonians can be extended to infinite periodic systems. By numerical benchmarking on different systems with classical and quantum algorithms, we reveal that stabilizer ground states are promising tools for both qualitative understanding of quantum systems and serving as the initial states of advanced classical or quantum algorithms for quantum many-body ground state problems. We envision stabilizer ground states and the corresponding algorithms as a pivotal foundation for a wide range of interesting applications, and highlight the collective power of the theories, algorithms, and applications to advance the field of quantum physics. 

This paper is organized as follows: We first introduce the notations and mathematical backgrounds of stabilizer states in Sec~\ref{sec:background}. The equivalence between the stabilizer ground state and CMCS for general Hamiltonians is derived in Sec.~\ref{sec:general}, and the exact linear-scaled algorithm for the stabilizer ground states of 1D local Hamiltonians is further presented in Sec.~\ref{sec:1D}. Sec.~\ref{sec:periodic} extends both the equivalence formalism for general Hamiltonians and the algorithm for 1D local Hamiltonians to infinite periodic systems. In Sec.~\ref{sec:applications}, we discuss the advantages and potential applications of stabilizer ground states to both classical and quantum algorithms, along with a comparison with other approximated ground states. 
Secs.~\ref{sec:cost} and \ref{sec:comparison} benchmark the exact 1D local algorithm on example Hamiltonians by numerically verifying the computational scaling and comparing the performances with numerically optimized stabilizer ground states, respectively.
In Sec.~\ref{sec:phase}, we demonstrate the ability of stabilizer ground states to qualitatively describe topological systems.
The power of stabilizer ground states as the cornerstone of advanced classical or quantum algorithms is shown in Sec.~\ref{sec:extended} and Sec.~\ref{sec:vqe}, respectively.
Finally, we draw conclusions and outline future directions for the development and applications of stabilizer ground states in Sec.~\ref{sec:conclusion}.

\section{Theory} \label{sec:theory}
\subsection{Notations and mathematical background of stabilizer groups} \label{sec:background}

We first revisit the definitions and a few frequently used properties of Pauli operators and stabilizer groups \cite{nielsen2001quantum,gottesman1997stabilizer}.

Let $\mathcal{P}_{n}=\pm\{I,X,Y,Z\}^{\otimes n}$ represent the set of Hermitian $n$-qubit Pauli operators.
It is important to clarify that $\mathcal{P}_{n}$ itself is not a group since $\mathcal{P}_{n}$ does not include anti-Hermitian operators.
For any two elements $P_{i},P_{j}\in\mathcal{\mathcal{P}}_{n}$, $P_{i}$ and $P_{j}$ either commute or anticommute. 
A Pauli operator $Q$ commutes with a set of Pauli operators $\boldsymbol{P}=\{P_{i}\}$, denoted as $[Q,\boldsymbol{P}]=0$, if $[Q,P]=0$ for each $P\in\boldsymbol{P}$. 
We denote it as $[Q,\boldsymbol{P}]\neq 0$ if $Q$ anticommutes with any $P\in\boldsymbol{P}$.

A stabilizer group $\boldsymbol{S}$ is a subset of $\mathcal{P}_{n}$ that forms a group and satisfies $-I\notin\boldsymbol{S}$.
Any two elements $P,Q$ in $\boldsymbol{S}$ commute with each other, otherwise $PQPQ=-I$ violates the definition.
$\langle\boldsymbol{P}\rangle=\langle P_{1},...,P_{l}\rangle=\{\prod_{Q\in\boldsymbol{Q}}Q |\boldsymbol{Q}\subseteq\boldsymbol{P}\}$ 
is the stabilizer group generated by a set of Pauli operators $\boldsymbol{P}=\{P_{i=1}^{l}\}$, if $\langle\boldsymbol{P}\rangle$ satisfies the definition of stabilizer group. We say $\boldsymbol{P}$ is a set of generators of the stabilizer group \textbf{$\boldsymbol{S}$} when $\boldsymbol{S}=\langle\boldsymbol{P}\rangle$. 
An $n$-qubit stabilizer group $\boldsymbol{S}$ has at most $n$ independent generators $\boldsymbol{g}=\{g_i\}$ and $2^{n}$ elements. If $|\boldsymbol{g}|=n$, \textbf{$\boldsymbol{S}$ }is a full stabilizer group and we have either $P\in\boldsymbol{S}$ or $-P\in\boldsymbol{S}$ (denoted as $P\in\pm\boldsymbol{S}$) for any $[P,\boldsymbol{S}]=0$, $P\in\mathcal{P}_{n}$. 

We say state $|\psi\rangle$ is stabilized by $P\in\mathcal{P}_{n}$ if $P|\psi\rangle=|\psi\rangle$.
We define the stabilizer group of a given $|\psi\rangle$ as
\begin{equation}
\text{Stab}(|\psi\rangle)=\{P \in \mathcal{P}_{n} |\, P|\psi\rangle=|\psi\rangle\},
\end{equation}
and $|\psi\rangle$ is a stabilizer state when $\text{Stab}(|\psi\rangle)$ is a full stabilizer group. The mapping $|\psi\rangle\rightarrow\text{Stab}(|\psi\rangle)$ from stabilizer states to full stabilizer groups is a one-to-one correspondence.

\subsection{Stabilizer ground states of general Hamiltonians} \label{sec:general}

In this section, we establish the equivalence between the stabilizer ground state and the closed maximally-commuting Pauli subset.
Furthermore, we show that, for properly defined sparse Hamiltonian (which includes almost all common Hamiltonians), such equivalence implies a much cheaper algorithm to get the stabilizer state compared with the brute-force approach.
We first define the stabilizer ground state of Hamiltonian $H$:

\begin{defn} (Stabilizer ground state)
The \textbf{stabilizer ground state} of a given Hamiltonian $H$ is the stabilizer
state $\vert \psi\rangle$ with the lowest energy expectation $\langle\psi|H|\psi\rangle$.
\end{defn}
The number of $n$-qubit stabilizer states is $\mathcal{S}(n)=2^{n}\prod_{i=1}^{n}(2^{i}+1)\sim2^{\frac{1}{2}(n+1)(n+2)}$, thus looping over all the stabilizer states is infeasible for large $n$ \cite{number_stab_groups}. 
To find the stabilizer ground state, ~\thref{expec} is first presented to determine the expectation value of a Pauli operator in a stabilizer state:
\begin{lem}
\thlabel{expec} (Expectation values of stabilizer states)
For any $n$-qubit stabilizer state $|\psi\rangle$ and Pauli operator $P\in\mathcal{P}_n$, if $P\notin\pm\text{Stab}(|\psi\rangle)$, then $\langle\psi|P|\psi\rangle=0$
\end{lem}
\begin{proof}
Since $\text{Stab}(|\psi\rangle)$ is a full stabilizer group, there exists $Q\in\text{Stab}(|\psi\rangle)$ such that $\{P,Q\}=0$. 
Thus
\begin{equation}
\langle\psi|P|\psi\rangle=\langle\psi|PQ|\psi\rangle=-\langle\psi|QP|\psi\rangle=-\langle\psi|P|\psi\rangle.
\end{equation}
Therefore, $\langle\psi|P|\psi\rangle=0$.
\end{proof}
We further extend the concept of energy expectation associated with a stabilizer state to a (not necessarily full) stabilizer group:
\begin{defn} (Stabilizer group energies)
The energy of a stabilizer group $\boldsymbol{S}$ for a given Pauli operator $P$ or a Pauli Hamiltonian $H=\sum_{P\in\boldsymbol{P}}w_{P}P$ is defined by
\begin{equation}
\begin{aligned}
E_{\text{stab}}(P,\boldsymbol{S}) & =\begin{cases}
\pm1 & P\in\pm\boldsymbol{S}\\
0 & \text{otherwise}
\end{cases},\\
E_{\text{stab}}(H,\boldsymbol{S}) & =\sum_{P\in\boldsymbol{P}}w_{P}E_{\text{stab}}(P,\boldsymbol{S}).
\end{aligned}
\end{equation}
\end{defn}
The relationship of stabilizer state energies and stabilizer group energies can be given by:
\begin{cor}
\thlabel{stab_energy} (Stabilizer state energies are stabilizer group energies)
We denote $\tilde{\boldsymbol{P}}=\pm \boldsymbol{P}=\{\pm P|P\in\boldsymbol{P}\}$ for a set of Pauli operators $\boldsymbol{P}$. 
For a Hamiltonian $H=\sum_{P\in\boldsymbol{P}}w_{P}P$ and a stabilizer state $|\psi\rangle$, let stabilizer group $\boldsymbol{S}=\langle\text{Stab}(|\psi\rangle)\cap\tilde{\boldsymbol{P}}\rangle$,
we have $\langle\psi|H|\psi\rangle=E_{\text{stab}}(H,\boldsymbol{S})$.
\end{cor}

\thref{stab_energy} implies that, the energy of a stabilizer state $|\psi\rangle$ depends only on $\text{Stab}(|\psi\rangle)\cap\tilde{\boldsymbol{P}}\subseteq\tilde{\boldsymbol{P}}$. Since the number of subsets of $\tilde{\boldsymbol{P}}$ is much less than $\mathcal{S}(n)\sim2^{\frac{1}{2}(n+1)(n+2)}$ when $\tilde{\boldsymbol{P}}$ is sparse (see \thref{def:sparse} for rigorous definition), it gives a better way to determine the stabilizer ground state. Rigorously, we introduce the concept of closed commuting subsets as follows:
\begin{defn}
\thlabel{CCS} (CCS) We define the \textbf{closed commuting subsets} (CCS) induced by $\boldsymbol{P}$ (or $H$) as 
\begin{equation}
\mathscr{S}(\boldsymbol{P})=\{\boldsymbol{Q}\subseteq\tilde{\boldsymbol{P}}|\boldsymbol{Q}=\langle\boldsymbol{Q}\rangle\cap\tilde{\boldsymbol{P}},-I\notin\langle\boldsymbol{Q}\rangle\}\label{eq:CCS}
\end{equation}
We note that $-I\notin\langle\boldsymbol{Q}\rangle$ implicitly indicates that $\langle\boldsymbol{Q}\rangle$ is a stabilizer group.
\end{defn}
Physically $\boldsymbol{Q}=\langle\boldsymbol{Q}\rangle\cap\tilde{\boldsymbol{P}}$ is saying that, $\langle\boldsymbol{Q}\rangle$ is generated by elements in $\tilde{\boldsymbol{P}}$. The name ``closed'' means that $\boldsymbol{Q}$ is closed under group mutiplication operations within the range of $\tilde{\boldsymbol{P}}$. Thus $\{\langle\boldsymbol{Q}\rangle|\boldsymbol{Q}\in\mathscr{S}(\boldsymbol{P})\}$ is the full set of stabilizer groups generated by elements in $\tilde{\boldsymbol{P}}$, and $\mathscr{S}(\boldsymbol{P})$ is an equivalent approach to represent with the generators in $\tilde{\boldsymbol{P}}$ instead. (knowing either $\boldsymbol{Q}$ or $\langle\boldsymbol{Q}\rangle$ immediately gives the other) As an example, we consider $\boldsymbol{P}=\{+Z_{1},+Z_{2},-X_{1}X_{2}\}$, then $\tilde{\boldsymbol{P}}=\{\pm Z_{1},\pm Z_{2},\pm X_{1}X_{2}\}$. Then $\mathscr{S}(\boldsymbol{P})$ includes $\emptyset,\{sZ_{1}\},\{sZ_{2}\},\{s_{1}Z_{1},s_{2}Z_{2}\},\{sX_{1}X_{2}\}$, where each $s$ can be $\pm1$ independently. We additionally note that each CCS corresponds to a stabilizer subspace and can be viewed as a discrete analog of a \textit{noncontextual subspace} in the sense used in Contextual Subspace VQE \cite{kirby2021contextual}, where all operators in the subset mutually commute and can be jointly diagonalized.

We now present \thref{sparse_gs}, which states that the stabilizer ground state can be obtained by searching for $\boldsymbol{Q}\in\mathscr{S}(\boldsymbol{P})$ with the lowest $E_{\text{stab}}(H,\langle\boldsymbol{Q}\rangle)$. The proof is given in the Appendix \ref{appendix:proof_sparse}. 

\begin{thm}
\thlabel{sparse_gs} (CMCS gives the stabilizer ground state) Given a Hamiltonian $H=\sum_{P\in\boldsymbol{P}}w_{P}P$ , then
\begin{equation}
E_{\text{gs}}=\min_{\boldsymbol{Q}\in\mathscr{S}(\boldsymbol{P})}E_{\text{stab}}(H,\langle\boldsymbol{Q}\rangle)\label{eq:CMCS}
\end{equation}
 is the stabilizer ground state energy. Such $\boldsymbol{Q}$ minimizes $E_{\text{stab}}(H,\langle\boldsymbol{Q}\rangle)$ is named as the \textbf{closed maximally-commuting subset} (CMCS) of $\tilde{\boldsymbol{P}}$ (or $H$). Additionally, for any $\boldsymbol{S}=\langle\boldsymbol{Q}\rangle$, $\boldsymbol{Q}\in\mathscr{S}(\boldsymbol{P})$ such that $E_{\text{stab}}(H,\boldsymbol{S})=E_{\text{gs}}$, each stabilizer state $|\psi\rangle$ stabilized by $\boldsymbol{S}$ is a (degenerate) stabilizer ground state.
\end{thm}
\thref{sparse_gs} suggests that the stabilizer ground state of a
Hamiltonian $H=\sum_{P\in\boldsymbol{P}}w_{P}P$ can be found by listing
all elements of $\mathscr{S}(\boldsymbol{P})$, and thus the computational cost is scaled with $|\mathscr{S}(\boldsymbol{P})|$. However, the
exact value of $|\mathscr{S}(\boldsymbol{P})|$ heavily depends
on the form of the Hamiltonian, e.g. the commutation/anticommutation
relations between the Pauli terms. We give a loose upper bound of $|\mathscr{S}(\boldsymbol{P})|$ as follows:
\begin{lem}
\thlabel{nS} (Upper bound of CCS) If $\boldsymbol{P}$ is defined on at most $n$ qubits,
then $|\mathscr{S}(\boldsymbol{P})|\leq(n+1)|\tilde{\boldsymbol{P}}|^{n}$.
\end{lem}
\begin{proof}
For any $\boldsymbol{Q}\in\mathscr{S}(\boldsymbol{P})$, $\langle\boldsymbol{Q}\rangle$ is constructed by at most $n$ independent generators in $\tilde{\boldsymbol{P}}$.
We simply select each generator one by one, each with at most $|\tilde{\boldsymbol{P}}|$ choices. 
Thus, we have $|\mathscr{S}(\boldsymbol{P})|\leq\sum_{l=0}^{n}|\tilde{\boldsymbol{P}}|^{l}\leq(n+1)|\tilde{\boldsymbol{P}}|^{n}$.
\end{proof}
In fact, we will see that if $|\boldsymbol{P}|$ scales polynomial with $n$, $|\mathscr{S}(\boldsymbol{P})|$ will have a slower growing rate than the number of $n$-qubit stabilizer states $\mathcal{S}(n)$.
Therefore we define the sparse Hamiltonians as follows:
\begin{defn} \thlabel{def:sparse} (Sparse Hamiltonians) 
A $n$-qubit Pauli Hamiltonian $H=\sum_{P\in\boldsymbol{P}}w_{P}P$
is \textbf{sparse} if $|\boldsymbol{P}|\sim O(\text{poly}(n))$.
\end{defn}
We note that almost all the common Hamiltonians are sparse Hamiltonians.
Now we estimate $|\mathscr{S}(\boldsymbol{P})|$ of sparse Hamiltonians:
\begin{cor} \thlabel{upper_bound_CCS} (Upper bound of CCS for sparse Hamiltonians) 
For sparse Hamiltonians $H=\sum_{P\in\boldsymbol{P}}w_{P}P$, $|\mathscr{S}(\boldsymbol{P})|\sim\exp(Cn\log n)$
for some constant $C$.
\end{cor}

For sparse Hamiltonians, although $|\mathscr{S}(\boldsymbol{P})|$ still increases exponentially with $n$, it is much smaller than the number of $n$-qubit stabilizer states $\mathcal{S}(n)\sim2^{\frac{1}{2}(n+1)(n+2)}$. 

\subsection{Stabilizer ground states of 1D local Hamiltonians \label{sec:1D}}

In this section, we present an algorithm to determine the stabilizer ground state of a given $n$-qubit 1D $k$-local Hamiltonian with a $O(n\exp (Ck\log k))$ computational scaling. Note that such scaling reduces to the $\exp(Cn\log n)$ scaling in \thref{upper_bound_CCS} in the limit of $k\rightarrow n$ up to a (possibly) different factor $C$. We refer to this algorithm as the \textbf{exact 1D local algorithm}. To illustrate the motivation of this algorithm, we first consider the constrained 1D classical Hamiltonians in Sec. \ref{sec:classical_gs} and introduce a state machine algorithm with a linear scaling to obtain the exact ground state. Then we build an exact mapping from the 1D local stabilizer ground state problem to the constrained 1D classical ground state problem in Sec. \ref{sec:mapping}. We formally present the exact 1D local algorithm for stabilizer ground states in Secs. \ref{sec:SM} and \ref{sec:F}, and its computational complexity is discussed in Sec. \ref{sec:complexity}.

\subsubsection{A classical analogy: State machine algorithm for the constrained 1D classical ground state problem \label{sec:classical_gs}}

We first consider a general 1D $n$-site, $k$-local classical Hamiltonian without constraints on a chain $\{\sigma_1,...,\sigma_n\}$ as:
\begin{equation}
H=\sum_{m=1}^{n}h_{m}(\sigma_{m},\sigma_{m+1},...,\sigma_{m+k-1}),\label{eq:classical_H}
\end{equation}
where each $\sigma_{i}$ is a discrete variable which can take $d$ distinct values, and $h_{i}$ is an arbitrary discrete function. Here $k$-local means that each term acts on at most $k$ continuous sites. The ground state of such Hamiltonian can be exactly determined by the following state machine algorithm with a computational complexity of $O(nd^{k+1})$. We introduce $s_{m}=(\sigma_{m},\sigma_{m+1},...,\sigma_{m+k-1})$, which combines $k$ degree of freedoms (DOFs) as a single DOF with $d^{k}$ possible values. The Hamiltonian can be written as the following $1$-local form
\begin{equation}
H=\sum_{m=1}^{n}h_{m}(s_{m}),\label{eq:1local_H}
\end{equation}
with the price that different $s_{i}$ are dependent. In other words, the coupling is transferred from the Hamiltonian to the valid state space, and we will later see that such change of perspective is necessary in the quantum case, as the definition of independent local state is impossible. A key step in this algorithm is that, for a fixed $s_{m}$, the possible values of $\boldsymbol{s}_{\leq m}=\{s_{1},...,s_{m}\}$ are now decoupled from $\boldsymbol{s}_{>m}=\{s_{m+1},...,s_{n}\}$, i.e., 
\begin{equation}
\{\boldsymbol{s}|s_{m}\}=\{\boldsymbol{s}_{\le m}|s_{m}\}\otimes\{\boldsymbol{s}_{>m}|s_{m}\}\label{eq:decoupling}
\end{equation}
where $\boldsymbol{s}=\{s_{1},...,s_{n}\}$. This independence is because there is no overlap between $\boldsymbol{s}_{\leq m}$ and $\boldsymbol{s}_{>m}$ over $\{\sigma_{i}\}$ beyond $s_{m}=(\sigma_{m},\sigma_{m+1},...,\sigma_{m+k-1})$. We define the energy on $\boldsymbol{s}_{\leq m}$ as
\begin{equation}
E(\boldsymbol{s}_{\leq m})=\sum_{i=1}^{m}h_{i}(s_{i})\label{eq:H_1local}
\end{equation}
and the ``conditioned ground state energy''  $E_{\text{gs}}^{\leq m}(s_{m})$ as
\begin{align}
E_{\text{gs}}^{\leq m}(s_{m}) & =\min_{\boldsymbol{s}_{\leq m}^{\prime}\in\{\boldsymbol{s}_{\leq m}|s_{m}\}}E(\boldsymbol{s}_{\leq m}^\prime).\label{eq:def_Egs_s}
\end{align}

The sequence $\{s_{1}, \ldots, s_{n}\}$ in Eq.~\eqref{eq:decoupling} admits a natural physical interpretation as a non-probabilistic analogue of a Markov chain, where the past ($\boldsymbol{s}_{\leq m}$) and future ($\boldsymbol{s}_{>m}$) are conditionally independent given the present state ($s_{m}$).
In addition, $E_{\text{gs}}^{\leq m}(s_{m})$, as a function defined on the ``non-probabilistic Markov chain'' $s_{m}$, has the property that, for given values of $E_{\text{gs}}^{\leq m}(s_{m})$ for all $s_{m}$, $E_{\text{gs}}^{\leq m+1}(s_{m+1})$ for each $s_{m+1}$ can be determined within $O(d^{k+1})$ time.
This conclusion is proved by first noticing that
\begin{equation}
\begin{aligned}\{\boldsymbol{s}_{\leq m}|s_{m+1}\} & =\cup_{s_{m}}\{\boldsymbol{s}_{\leq m}|s_{m+1},s_{m}\}\\
 & =\cup_{s_{m}\in\{s_{m}^{\prime}|s_{m+1}\}}\{\boldsymbol{s}_{\leq m}|s_{m}\}\\
 & :=\cup_{s_{m}\in F_{\text{back}}(s_{m+1})}\{\boldsymbol{s}_{\leq m}|s_{m}\}.
\end{aligned}
\label{eq:next_index}
\end{equation}
In the first line, all possible values of $s_{m}$ are looped over to determine $\boldsymbol{s}_{\leq m}$. The shared sites $\sigma_{m+1},...,\sigma_{m+k-1}$ in both $s_{m+1}$ and $s_{m}$ must have the same values, and this constraint is denoted as $s_{m}\in\{s_{m}^{\prime}|s_{m+1}\}$ in the second line. If this constraint holds, we have $\{\boldsymbol{s}_{\leq m}|s_{m+1},s_{m}\}=\{\boldsymbol{s}_{\leq m}|s_{m}\}$ due to the decoupling between $\boldsymbol{s}_{\leq m}$ and $s_{m+1}$ for given $s_{m}$ (see Eq. \eqref{eq:decoupling}). 
We also introduce a shorthand $F_{\text{back}}(s_{m+1})$ for $\{s_{m}^{\prime}|s_{m+1}\}$ in the last line. For instance, the $2$-local spin Hamiltonian, i.e., $k=2$, $s_{m}=(\sigma_{m},\sigma_{m+1})$, and $\sigma_{i}\in\{\uparrow,\downarrow\}$, has $F_{\text{back}}(\sigma_{m+1},\sigma_{m+2})=\{(\sigma_{m},\sigma_{m+1})|\sigma_{m}\in\{\uparrow,\downarrow\}\}$. From Eq. \eqref{eq:next_index}, the recurrence relation between $E_{\text{gs}}^{\leq m}(s_{m})$ and $E_{\text{gs}}^{\leq m+1}(s_{m+1})$ can be derived as
\begin{equation}
\begin{aligned}E_{\text{gs}}^{\leq m+1}(s_{m+1}) & =\min_{\boldsymbol{s}_{\leq m+1}^{\prime}\in\{\boldsymbol{s}_{\leq m+1}|s_{m+1}\}}E(\boldsymbol{s}_{\leq m+1}^{\prime})\\
 & =h_{m+1}(s_{m+1})+\min_{\boldsymbol{s}_{\leq m}^{\prime}\in\{\boldsymbol{s}_{\leq m}|s_{m+1}\}}E(\boldsymbol{s}_{\leq m}^{\prime})\\
 & =h_{m+1}(s_{m+1})+\min_{s_{m}\in F_{\text{back}}(s_{m+1})}E_{\text{gs}}^{\leq m}(s_{m}).
\end{aligned}
\label{eq:classical_min}
\end{equation}
For the above 2-local spin Hamiltonian example, Eq. \eqref{eq:classical_min} is equivalent to
\begin{equation}
E_{\text{gs}}^{\leq m+1}(\sigma_{m+1},\sigma_{m+2})=h_{m+1}(\sigma_{m+1},\sigma_{m+2})+\min_{\sigma_{m}\in\{\uparrow,\downarrow\}}E_{\text{gs}}^{\leq m}(\sigma_{m},\sigma_{m+1}).
\end{equation}
Thus, to obtain the ground state energy, we can sequentially build tables $T_{m}$ for $m=1,2,...,n$ by Eq. \eqref{eq:classical_min} to record the values of $E_{\text{gs}}^{\leq m}(s_{m})$ over all $s_m$, with
the initial table $T_1$ as  $E_{\text{gs}}^{\leq 1}(s_{1})=h_{1}(s_{1})$. The final ground state energy is expressed as 
\begin{equation}
E_{\text{gs}}=\min_{s_{n}}E_{\text{gs}}^{\leq n}(s_{n}).\label{eq:classical_final_min}
\end{equation}
Since $F_{\text{back}}(s_{m+1})$ only has $d$ elements (i.e. only $\sigma_{m}$ can choose different values freely), and $s_{m}$ has $d^{k}$ choices, one can build $T_{m+1}$ from $T_{m}$ within $O(d^{k+1})$ time, so the final computational complexity is $O(nd^{k+1})$. The ground state $\boldsymbol{s}^{\star}$ (or $\boldsymbol{\sigma}^{\star}$) is achieved when each table $T_m$ is built with $s_{m}=s_{m}^{\star}$, where $s_{m}^{\star}$ minimizes the right hand side (RHS) of Eq. \eqref{eq:classical_min} ($m<n$) or Eq. \eqref{eq:classical_final_min} ($m=n$). A schematic workflow diagram of this algorithm using 2-local spin Hamiltonian example is shown in Fig. \ref{fig:classical}.

\begin{figure}

\centering

\includegraphics[width=0.9\linewidth]{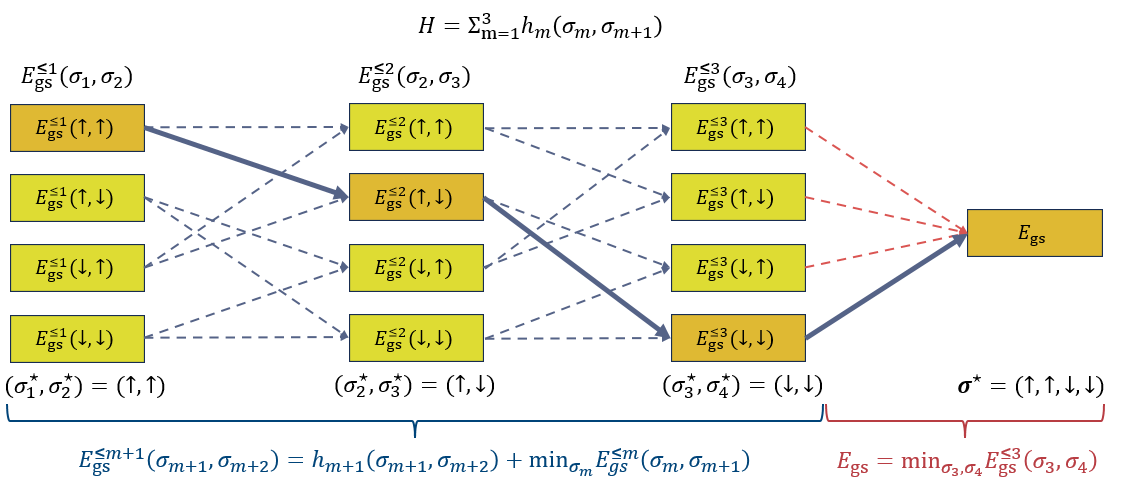}
\caption{State machine algorithm of the 1D local classical ground state problem illustrated by a $2$-local spin Hamiltonian $H=\sum_{m=1}^{3}h_{m}(\sigma_{m},\sigma_{m+1})$. For illustration, we use $\sigma$ instead of $s$ everywhere. For each $m$ we define the ``conditioned ground state energy'' $E_{\text{gs}}^{\leq m}(\sigma_{m},\sigma_{m+1})$ (see Eq. \eqref{eq:def_Egs_s}). $E_{\text{gs}}^{\leq 1}$ is simply $h_{1}$, and $E_{\text{gs}}^{\leq m+1}$ can be determined from $E_{\text{gs}}^{\leq m}$ via Eq. \ref{eq:classical_min}. The final ground state is the minimum value of $E_{\text{gs}}^{\leq n}$ (see Eq. \eqref{eq:classical_final_min}). The ground state can be obtained by finding each $s_{m}=(\sigma_{m},\sigma_{m+1})$ (orange blocks) that minimizes the RHS of Eq. \eqref{eq:classical_min} and Eq. \eqref{eq:classical_final_min}.}

\label{fig:classical}
\end{figure}

Although the constraints between $s_{m}$ are local ($s_{m}$ and $s_{m^{\prime}}$ are independent if $|m-m^{\prime}|\geq k$) in the above case, this approach also works for some special non-local constraints. An illustrative example is $H=\sum_{m=1}^{n}h_{m}(s_{m})$, $s_{m}\in\{0,1,...,d-1\}$ with constraints $f_{m}(\sum_{i=1}^{m}s_{i})\geq0$ and $g_{m}(\sum_{i=m+1}^{n}s_{i})\geq0$ for each $m$, where the summation $\sum_{i=1}^{m}s_{i}$ and $\sum_{i=m+1}^{n}s_{i}$ are modulo $d$. We introduce a state machine
\begin{equation}
A_{m}=(s_{m},\alpha_{m}=\sum_{i=1}^{m}s_{i},\beta_{m}=\sum_{i=m+1}^{n}s_{i}),
\end{equation}
which is a function of $\boldsymbol{s}=\{s_{1},...,s_{n}\}$. With a fixed $A_{m}=(s_{m},\alpha_{m},\beta_{m})$, the constraints $f_{l}$ and $g_{l}$ with $l\leq m$ can be rewritten as
\begin{equation}
\begin{aligned}f_{l}(\sum_{i=1}^{l}s_{i}) & =f_{l}(\alpha_{m}-\tau_{l}(\boldsymbol{s}_{\leq m}))\geq0\\
g_{l}(\sum_{i=l+1}^{n}s_{i}) & =g_{l}(\beta_{m}+\tau_{l}(\boldsymbol{s}_{\leq m}))\geq0,
\end{aligned}
\end{equation}
where $\tau_{l}(\boldsymbol{s}_{\leq m})=\sum_{i=l+1}^{m}s_{i}$. This implies that these two constraints act only on $\boldsymbol{s}_{\leq m}$. Similarly, the constraints of $l>m$ act only on $\boldsymbol{s}_{>m}$. Therefore, $\boldsymbol{s}_{\leq m}$ and $\boldsymbol{s}_{>m}$ are again decoupled given $A_m$, i.e.,
\begin{equation}
\{\boldsymbol{s}|A_{m}\}=\{\boldsymbol{s}_{\leq m}|A_{m}\}\otimes\{\boldsymbol{s}_{>m}|A_{m}\}.\label{eq:Markov_independence_s}
\end{equation}
For a given $A_m$, $\boldsymbol{s}_{\leq m}$ fully determines each $A_{l\leq m}$ as $A_{l}=(s_{l},\alpha_{m}-\tau_{l}(\boldsymbol{s}_{\leq m}),\beta_{m}+\tau_{l}(\boldsymbol{s}_{\leq m}))$, and $\boldsymbol{s}_{>m}$ similarly determines each $A_{l>m}$. Thus we conclude
\begin{equation}
\{A_{1},...,A_{n}|A_{m}\}=\{A_{1},...,A_{m}|A_{m}\}\otimes\{A_{m+1},...,A_{n}|A_{m}\},\label{eq:decoupling_SM}
\end{equation}
which is a generalized form of Eq. \eqref{eq:decoupling}. In fact, Eq. \eqref{eq:decoupling} can be viewed as a special case of Eq. \eqref{eq:decoupling_SM} when $A_{m}=s_{m}$. Since $A_{m}$ contains the information of $s_{m}$, we can rewrite $h_{m}(s_{m})$ as $h_{m}^{\prime}(A_{m})$. To use the previous algorithm, we introduce a transition function
\begin{equation}
\begin{aligned}F(A_{m}) & =\{A_{m+1}|A_{m}\}\\
 & =\{(s_{m+1},\alpha_{m+1}=\alpha_{m}+s_{m+1},\beta_{m+1}=\beta_{m}-s_{m+1})|f_{m+1}(\alpha_{m+1})\geq0,g_{m+1}(\beta_{m+1})\ge0\},
\end{aligned}
\end{equation}
and thus $F_{\text{back}}(A_{m})=\{A_{m-1}|A_{m}\in F(A_{m-1})\}$. We can then similarly define the ``conditioned ground state energy'' $E_{\text{gs}}^{\leq m}(A_{m})$ and recursively construct the values for $m=1,2,...,n$ to determine the ground state. For this specific problem, $A_m$ has $d^3$ choices, and $F(A_m)$ has $d$ elements, so the computational complexity to obtain the ground state is $O(nd^{4})$.

We further generalize the above solution of constrained 1D classical ground state problems. Without loss of generality, the input Hamiltonian can be assumed as $1$-local since a $k$-local problem can always be rewritten as a 1-local problem. The key is to construct a state machine $A_{m}$ that satisfies Eq. \eqref{eq:decoupling_SM}. Once we constructed such $A_{m}$, the ground state could be obtained by recursively building tables $T_{m}$ that record each $E_{\text{gs}}^{\leq m}(A_{m})$. Now we rigorously define a state machine as follows:
\begin{defn}
\thlabel{SM} (State machine) Consider a classical system with single-site states $s_{1},...,s_{n}$, where each of them can take some discrete values $s_{i}\in d_{i}$. The valid space $\boldsymbol{D}$ is a subset of the full space $d_{1}\otimes d_{2}\otimes...\otimes d_{n}$ due to some constraints. $A_{m}$ defined on $\boldsymbol{D}$ is a \textbf{state machine} if (1) $A_{m}$ contains information of $s_{m}$, i.e. there exists a \textbf{state function} $G(A_{m}(\boldsymbol{s}))=s_{m}$ for any $\boldsymbol{s}\in\boldsymbol{D}$, and (2) given $A_{m}$, the possible values of the left state machines $A_{1},...,A_{m}$ and the possible values of right state machines $A_{m+1},...,A_{n}$ are independent, i.e.
\begin{equation}
\{A_{1},...,A_{n}|A_{m}\}=\{A_{1},...,A_{m}|A_{m}\}\otimes\{A_{m+1},...,A_{n}|A_{m}\}.\label{eq:Markov_independence}
\end{equation}
The corresponding \textbf{transition function} is defined as
\begin{align}
F(A_{m}) & =\{A_{m+1}|A_{m}\}\label{eq:def_F}\\
 & :=\{A_{m+1}(\boldsymbol{s})|A_{m}(\boldsymbol{s})=A_{m},\boldsymbol{s}\in\boldsymbol{D}\},
\end{align}
where the second line is a rigorous definition of $\{A_{m+1}|A_{m}\}$. The definitions of other expressions like $\{...|...\}$ are similar.
\end{defn}
As mentioned previously, it can be understood as a non-probabilistic version of the Markov chain. We additionally require the existence of the state function $G$ so that we could write $h_{m}(s_{m})$ as $h_{m}(G(A_{m}))$ in the Hamiltonian. Now the above state machine algorithm of the constrained classical ground state problem can be generalized in the following theorem, and a schematic workflow is shown in Fig. \ref{fig:SM}.
\begin{thm}
\thlabel{classical_gs} (State machine algorithm of the constrained classical ground state problem) If $A_{m}$ is a state machine on $\boldsymbol{s}\in\boldsymbol{D}$ with the state function $G$ and transition function $F$, then the ground state of Hamiltonian $H=\sum_{m=1}^{n}h_{m}(s_{m})$ can be determined by
\end{thm}
\begin{enumerate}
\item Construct $\mathcal{A}_{1}=\{A_{1}(\boldsymbol{s})|\boldsymbol{s}\in\boldsymbol{D}\}$ and $E_{\text{gs}}^{\leq 1}(A_{1})=h_{1}(G(A_{1}))$ for all $A_{1}\in\mathcal{A}_{1}$.
\item From $m=1$ to $n-1$, create $\mathcal{A}_{m+1}=\cup_{A_m\in\mathcal{A}_{m}}F(A_m)$ and
\begin{equation}
E_{\text{gs}}^{\leq m+1}(A_{m+1}\in\mathcal{A}_{m+1})=h_{m+1}(G(A_{m+1}))+\min_{A_{m}\in F_{\text{back}}(A_{m+1})}E_{\text{gs}}^{\leq m}(A_{m}),\label{eq:E_next}
\end{equation}
where $F_{\text{back}}(A_{m+1})=\{A_{m}|A_{m+1}\in F(A_{m})\}$.
\item The ground state energy is 
\begin{equation}
E_{\text{gs}}=\min_{A_{n}\in\mathcal{A}_{n}}E_{\text{gs}}^{\leq n}(A_{n}).\label{eq:E_final}
\end{equation}
\item To obtain the ground state $\boldsymbol{s}^{\star}$, we find $\{A_{1}^{\star},...,A_{n}^{\star}\}$ by 
\begin{equation}
\begin{aligned}A_{n}^{\star} & =\argmin_{A_{n}\in\mathcal{A}_{n}}E_{\text{gs}}^{\leq n}(A_{n}),\\
A_{m}^{\star} & =\argmin_{A_{m}\in F_{\text{back}}(A_{m+1}^{\star})}E_{\text{gs}}^{\leq m}(A_{m})
\end{aligned}
\label{eq:gs_state}
\end{equation}
for $m=1,2,...,n-1$. Then the components of $\boldsymbol{s}^{\star}$ is $s_{m}^{\star}=G(A_{m}^{\star})$ for $m=1,2,...,n$.
\end{enumerate}
\begin{proof}
The Hamiltonian can be rewritten as $H=\sum_{m=1}^{n}h_{m}(G(A_{m}))$, thus the all derivations in the first $k$-local classical ground state problem still hold with all $s_{m}$ changed to $A_{m}$.
\end{proof}
We additionally note that, the algorithm in \thref{classical_gs} is not automatically linearly-scaled, unless we can ensure that the evaluations of transition function $F$ and the initial $\mathcal{A}_1$ can be done in a constant time.
\begin{figure}[tbh]
\centering

\includegraphics[width=0.95\linewidth]{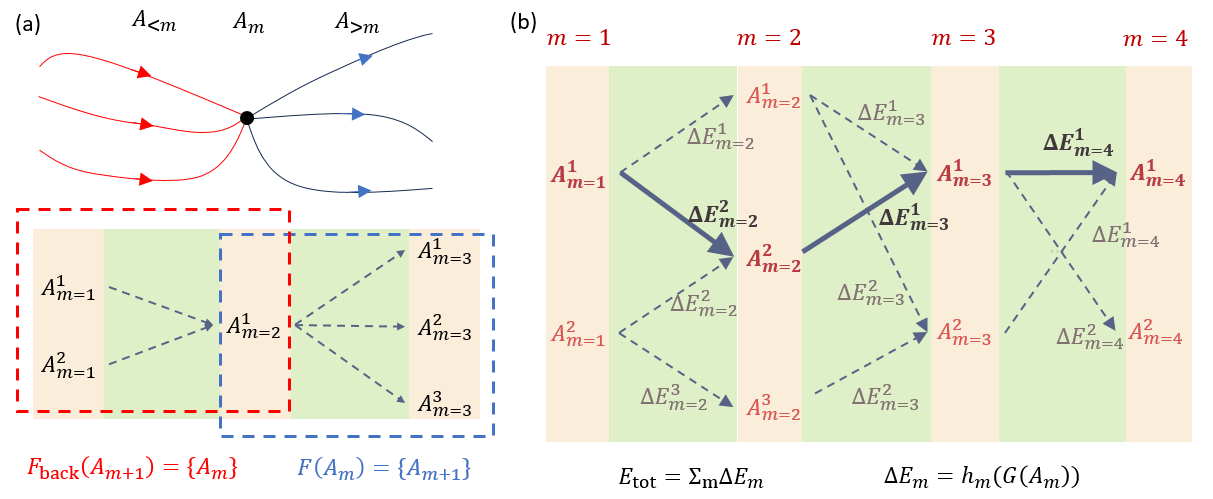}

\caption{(a) Schematic diagram of the state machine and transition function defined in \thref{SM}. Given $A_{m}$, the possible paths of $A_{<m}=\{A_{1},...,A_{m-1}\}$ should be decoupled from the possible paths of $A_{>m}=\{A_{m+1},...,A_{n}\}$. The transition function $F(A_{m})$ gives the possible values of $A_{m+1}$ given $A_{m}$, and $F_{\text{back}}(A_{m+1})$ gives the possible values of $A_{m}$ given $A_{m+1}$. (b) Schematic diagram of the state machine algorithm in \thref{classical_gs} to obtain the ground state of constrained 1D local classical Hamiltonians. The Hamiltonian is given by $H=\sum_{m=1}^{n}h_{m}(s_{m})$ in some constrained space $\boldsymbol{D}$. Given a state machine $A_{m}$ and transition function $F$ satisfying \thref{SM}, we start from $A_{1}\in\mathcal{A}_{1}$, and sequentially determine the possible values of $A_{m+1}\in F(A_{m})$ for each $A_{m}$. (lines connecting $A_{m}$ and $A_{m+1}$ indicates that $A_{m+1}\in F(A_{m})$) In this process, each $A_{m}$ can generate multiple $A_{m+1}$, and each $A_{m+1}$ can be generated from multiple $A_{m}$. The energy difference $\Delta E_{m}=h_{m}(G(A_{m}))$ is a function of $A_{m}$, where $G$ is the state functions satisfying $G(A_{m})=s_{m}$. By such construction, each path $\{A_{1},...,A_{n}\}$ such that $A_{m+1}\in F(A_{m})$, $m=1,...,n-1$ uniquely maps a valid state $\boldsymbol{s}\in\boldsymbol{D}$ via $s_{m}=G(A_{m})$. To determine the ground state, we record the value of $E_{\text{gs}}^{\leq m}(A_{m})$ for each $A_{m}$ in a table starting from $E_{\text{gs}}^{\leq 1}(A_{1})=h_{1}(A_{1})$. Once we record all values of $E_{\text{gs}}^{\leq m}(A_{m})$, the values of $E_{\text{gs}}^{\leq m+1}(A_{m+1})$ can be determined from Eq. \eqref{eq:E_next}. The final ground state is simply $E_{\text{gs}}=\min_{A_{n}}(E_{\text{gs}}^{\leq n}(A_{n}))$, and the ground state is the corresponding path $\{s_{1}^{\text{\ensuremath{\star}}},....,s_{n}^{\star}\}$ (bold lines) such that each minimization in Eq. \eqref{eq:E_next} and Eq. \ref{eq:E_final} is taken at $s_{m}=s_{m}^{\star}$.}

\label{fig:SM}
\end{figure}

\subsubsection{Mapping the stabilizer ground state problem to the constrained classical ground state problem \label{sec:mapping}}

In this subsection, we show that the 1D local stabilizer ground state problem can be rigorously mapped to the constrained classical ground state problem considered in Sec. \ref{sec:classical_gs}. we first define the 1D $k$-local Hamiltonian as follows:
\begin{defn} ($k$-local Hamiltonians)
A Hamiltonian $H=\sum_{P\in\boldsymbol{P}}w_{P}P$ is $k$-local if each $P\in\boldsymbol{P}$ is between qubit $q_{P}^{\text{first}}$ to $q_{P}^{\text{last}}$ that satisfies $q_{P}^{\text{last}}-q_{P}^{\text{first}}\leq k-1$.
\end{defn}
In the context of stabilizer ground states, we need to find suitable states with equivalent identities as $\boldsymbol{s}$. According to \thref{SM} and \thref{classical_gs}, the stabilizer ground state energy can be found via $E_{\text{gs}}=\min_{\boldsymbol{Q}\in\mathscr{S}(\boldsymbol{P})}E_{\text{stab}}(H,\langle\boldsymbol{Q}\rangle)$. In other words, $\mathscr{S}(\boldsymbol{P})=\{\boldsymbol{Q}\subseteq\tilde{\boldsymbol{P}}|\boldsymbol{Q}=\langle\boldsymbol{Q}\rangle\cap\tilde{\boldsymbol{P}},-I\notin\langle\boldsymbol{Q}\rangle\}$ serves as the state space $\boldsymbol{D}$, and each $\boldsymbol{Q}\in\mathscr{S}(\boldsymbol{P})$ is a state $\boldsymbol{s}$. We find the equivalent identity of single-site states $s_{m}$ as follows.   We divide $\tilde{\boldsymbol{P}}$ onto sites $m$ by the last non-identity qubit $q_{P}^{\text{last}}$ for any $P\in\tilde{\boldsymbol{P}}$, and define $\tilde{\boldsymbol{P}}_{I}=\{P\in\boldsymbol{P}|q_{P}^{\text{last}}\in I\}$ for a given index set $I$. For index sets $\{m\}$, $\{i|a\leq i\leq b\}$, $\{i|i>m\}$, $\{i|i\geq m\}$, $\{i|i<m\}$, and $\{i|i\leq m\}$, we use shorthands $\tilde{\boldsymbol{P}}_{m}$, $\tilde{\boldsymbol{P}}_{a,b}$, $\tilde{\boldsymbol{P}}_{>m}$, $\tilde{\boldsymbol{P}}_{\geq m}$, $\tilde{\boldsymbol{P}}_{<m}$, and $\tilde{\boldsymbol{P}}_{\leq m}$, respectively. For example, $\tilde{\boldsymbol{P}}_{\leq m}=\{P\in\tilde{\boldsymbol{P}}|q_{P}^{\text{last}}\leq m\}$. According to the definition, we have $\tilde{\boldsymbol{P}}=\cup_{m=1}^{n}\tilde{\boldsymbol{P}}_{m}$. We also introduce similar shorthands for $\boldsymbol{P}_{I}$. Since $\boldsymbol{Q}\in\mathscr{S}(\boldsymbol{P})$ is a subset of $\tilde{\boldsymbol{P}}$, we have $\boldsymbol{Q}=\cup_{m=1}^{n}\boldsymbol{Q}_{m}$ with $\boldsymbol{Q}_{m}=\boldsymbol{Q}\cap\tilde{\boldsymbol{P}}_{m}$, and $\boldsymbol{Q}_{m}$ is the equivalent identity of single-site states $s_{m}$. The constraint $\boldsymbol{Q}\in\mathscr{S}(\boldsymbol{P})$ is the equivalent relation of $\boldsymbol{s}\in\boldsymbol{D}$. To give an example of the constraint, we consider $\boldsymbol{P}=\{Z_{1},Z_{2},Z_{1}Z_{2}\}$, and correspondingly $\boldsymbol{P}_{1}=\{Z_{1}\}$, $\boldsymbol{P}_{2}=\{Z_{2},Z_{1}Z_{2}\}$. If $Z_{1}\in\boldsymbol{Q}$ and $Z_{2}\in\boldsymbol{Q}$ we must have $Z_{1}Z_{2}\in\boldsymbol{Q}$ as well, so $\boldsymbol{Q}_{1}=\{Z_{1}\}$ and $\boldsymbol{Q}_{2}=\{Z_{2}\}$ cannot hold simultaneously, although either of them can be taken if the other one is empty. Finally, the stabilizer energy $E_{\text{stab}}(H,\langle\boldsymbol{Q}\rangle)$ for any $\boldsymbol{Q}\in\mathscr{S}(\boldsymbol{P})$ can also be written as the sum of local $\boldsymbol{Q}_{m}$:
\begin{equation}
\begin{aligned}E_{\text{stab}}(H,\langle\boldsymbol{Q}\rangle) & =\sum_{P\in\langle\boldsymbol{Q}\rangle\cap\boldsymbol{P}}w_{P}-\sum_{P\in\langle\boldsymbol{Q}\rangle\cap(-\boldsymbol{P})}w_{-P}\\
 & =\sum_{m}\big(\sum_{Q\in\boldsymbol{Q}_{m}\cap\boldsymbol{P}_{m}}w_{Q}-\sum_{Q\in\boldsymbol{Q}_{m}\cap(-\boldsymbol{P}_{m})}w_{-Q}\big)\\
 & :=\sum_{m}h_{m}(\boldsymbol{Q}_{m}) 
\end{aligned}
\label{eq:h_Qm}
\end{equation}

In summary, the mapping from the 1D local stabilizer ground state problem to the 1D local constrained classical ground state problem is:
\begin{enumerate}
\item $\boldsymbol{Q}_{m}\rightarrow s_{m}$
\item $\{\boldsymbol{Q}_{m}|\boldsymbol{Q}_{m}\subseteq\tilde{\boldsymbol{P}}_{m}\}\rightarrow\boldsymbol{d}_{m}$
\item $\mathscr{S}(\boldsymbol{P})\rightarrow\boldsymbol{D}$
\item $\boldsymbol{Q}\in\mathscr{S}(\boldsymbol{P})\rightarrow\boldsymbol{s}\in\boldsymbol{D}$
\item $h_{m}(\boldsymbol{Q}_{m})\rightarrow h_{m}(s_{m})$
\end{enumerate}

We additionally note that the single-site state $\boldsymbol{Q}_{m}$ becomes the CCS $\boldsymbol{Q}$ in the limit of $k\rightarrow n$, which suggests that the 1D local algorithm reduces to the general solution when $k\rightarrow n$.

\subsubsection{Construction of state machine $A_{m}$ \label{sec:SM}}

In the previous section, we have already mapped the 1D local stabilizer ground state problem to a constrained classical ground state problem. Once we construct the state machine satisfying \thref{SM}, we can obtain the ground state by \thref{classical_gs}.
We first introduce the projection operation as follow:
\begin{defn}
\thlabel{def:projection} (Projection) We denote $\mathcal{P}_{I}=\pm\{I_{i},X_{i},Y_{i},Z_{i}\}^{\otimes_{i\in I}}$, where $i$ indicates the qubit index, and $I$ stands for some index set. The \textbf{projection} of a set of Pauli operators $\boldsymbol{P}$ to qubits $I$ is $\mathbb{P}_{I}(\boldsymbol{P})=\boldsymbol{P}\cap\mathcal{P}_{I}$. The notations of the index set $I$ are the same with $\tilde{\boldsymbol{P}}_{I}$ in Sec. \ref{sec:mapping}.
\end{defn}
It is important to note that, if $\boldsymbol{S}$ is a stabilizer group,  $\mathcal{P}_{I}(\boldsymbol{S})$ is also a stabilizer group. Following the definition of state machine in \thref{SM} and the mapping in Sec. \ref{sec:mapping}, we hope to construct $A_{m}$ on $\boldsymbol{Q}\in\mathscr{S}(\boldsymbol{P})$ satisfying Eq. \eqref{eq:Markov_independence}. In the second example of Sec. \ref{sec:classical_gs}, we first constructed $A_{m}$ satisfying Eq. \eqref{eq:Markov_independence_s}. In Appendix \ref{appendix:proof_Markov_independence_s}, we prove that Eq. \eqref{eq:Markov_independence_s} is a necessary condition of Eq. \eqref{eq:Markov_independence}. According to the mapping \textbf{$\boldsymbol{s}\rightarrow\boldsymbol{Q}$}, Eq. \eqref{eq:Markov_independence_s} is
\begin{equation}
\{\boldsymbol{Q}|A_{m}\}=\{\boldsymbol{Q}_{\leq m}|A_{m}\}\otimes\{\boldsymbol{Q}_{>m}|A_{m}\},\label{eq:Markov_independence_Q}
\end{equation}
i.e., given $A_{m}$, \textbf{$\boldsymbol{Q}_{\leq m}$ }and $\boldsymbol{Q}_{>m}$ are decoupled. In Appendix \ref{appendix:SM}, we derive 
\begin{equation}
A_{m}(\boldsymbol{Q})=(\boldsymbol{S}_{\text{proj}}^{m}(\boldsymbol{Q}_{\leq m}),\tilde{\boldsymbol{P}}_{\text{invalid}}^{m}(\boldsymbol{Q}_{\leq m}),\boldsymbol{S}_{\text{right}}^{m}(\boldsymbol{Q}_{\geq m}))\label{eq:expr_Am}
\end{equation}
that satisfies Eq. \eqref{eq:Markov_independence_s} and further Eq. \eqref{eq:Markov_independence}, where
\begin{align}
\tilde{\boldsymbol{P}}_{\text{invalid}}^{m}(\boldsymbol{Q}_{\leq m}) & =\{P\in\tilde{\boldsymbol{P}}_{>m}|[P,\boldsymbol{Q}_{\leq m}]\neq0\},\label{eq:Pinvalid}\\
\boldsymbol{S}_{\text{proj}}^{m}(\boldsymbol{Q}_{\leq m}) & =\mathbb{P}_{>m-k}(\langle\boldsymbol{Q}_{\leq m}\rangle),\label{eq:Sproj}\\
\boldsymbol{S}_{\text{right}}^{m}(\boldsymbol{Q}_{\geq m}) & =\mathbb{P}_{\leq m}(\langle\boldsymbol{Q}_{\geq m}\rangle).\label{eq:Sright}
\end{align}
A physical illustration is that, the coupling between $\boldsymbol{Q}_{\leq m}$ and $\boldsymbol{Q}_{>m}$ comes from (1) $[\boldsymbol{Q}_{>m},\boldsymbol{Q}_{\leq m}]=0$ since elements of a stabilizer group commute, and (2) group multiplication operations between $\boldsymbol{Q}_{\leq m}$ and $\boldsymbol{Q}_{>m}$ does not generated new elements in either $\tilde{\boldsymbol{P}}_{>m}$ and $\tilde{\boldsymbol{P}}_{\leq m}$. Given $A_{m}=(\boldsymbol{S}_{\text{proj}}^{m},\tilde{\boldsymbol{P}}_{\text{invalid}}^{m},\boldsymbol{S}_{\text{right}}^{m})$, (1) is equivalent to $\tilde{\boldsymbol{P}}_{\text{invalid}}^{m}(\boldsymbol{Q}_{\leq m})=\tilde{\boldsymbol{P}}_{\text{invalid}}^{m}$, and (2) is equivalent to $\boldsymbol{S}_{\text{proj}}^{m}(\boldsymbol{Q}_{\leq m})=\boldsymbol{S}_{\text{proj}}^{m}$ (for the $\tilde{\boldsymbol{P}}_{>m}$ part) and $\boldsymbol{S}_{\text{right}}^{m}(\boldsymbol{Q}_{\geq m})=\boldsymbol{S}_{\text{right}}^{m}$ (for the $\tilde{\boldsymbol{P}}_{\leq m}$ part), respectively. Thus, $\boldsymbol{Q}_{\leq m}$ and $\boldsymbol{Q}_{>m}$ are now decoupled.

Finally, by noticing
\begin{align}
\boldsymbol{S}_{\text{right}}^{m}(\boldsymbol{Q}_{\geq m})\cap\tilde{\boldsymbol{P}}_{m} & =\mathbb{P}_{\leq m}(\langle\boldsymbol{Q}_{\geq m}\rangle)\cap\tilde{\boldsymbol{P}}_{m}=\boldsymbol{Q}_{m},
\end{align}
the state function $G$ is
\begin{equation}
G(A_{m}=(\boldsymbol{S}_{\text{proj}}^{m},\tilde{\boldsymbol{P}}_{\text{invalid}}^{m},\boldsymbol{S}_{\text{right}}^{m}))=\boldsymbol{S}_{\text{right}}^{m}\cap\tilde{\boldsymbol{P}}_{m}.\label{eq:G}
\end{equation}
Therefore, the state machine for the 1D local stabilizer ground state problem is
\begin{cor}
\thlabel{SM_Q} (State machine of the 1D local stabilizer ground state problem) $A_{m}$ defined in Eq. \eqref{eq:expr_Am} is a state machine of the state space $\mathscr{S}(\boldsymbol{P})$. The state function $G$ is given in Eq. \eqref{eq:G}.
\end{cor}

\subsubsection{Construction of transition function $F$ \label{sec:F}}

After defining the state machine in \thref{SM_Q}, we need to further derive the transition function $F(A_{m})=\{A_{m+1}(\boldsymbol{Q})|A_{m}=A_{m}(\boldsymbol{Q}),\boldsymbol{Q}\in\mathscr{S}(\boldsymbol{P})\}$ in order to use \thref{classical_gs}. A naive strategy is looping over all $\boldsymbol{Q}\in\mathscr{S}(\boldsymbol{P})$ and calculating each $A_{m}(\boldsymbol{Q})$ and $A_{m+1}(\boldsymbol{Q})$, however its computational complexity is still exponential. Our strategy is introduced as follows with a linear complexity. 
We first notice that both $\boldsymbol{S}_{\text{proj}}^{m+1}$ and $\tilde{\boldsymbol{P}}_{\text{invalid}}^{m+1}$ can be obtained from $A_{m}$ and $\boldsymbol{S}_{\text{right}}^{m+1}$ as:
\begin{align}
\boldsymbol{S}_{\text{proj}}^{m+1}(\boldsymbol{S}_{\text{proj}}^{m},\boldsymbol{Q}_{m+1}) & =\langle\mathbb{P}_{>m-k+1}(\boldsymbol{S}_{\text{proj}}^{m}(\boldsymbol{Q}_{\leq m})),\boldsymbol{Q}_{m+1}\rangle,\label{eq:Sproj_next}\\
\tilde{\boldsymbol{P}}_{\text{invalid}}^{m+1}(\tilde{\boldsymbol{P}}_{\text{invalid}}^{m},\boldsymbol{Q}_{m+1})= & \{P\in\tilde{\boldsymbol{P}}_{\geq m+1}|P\in\tilde{\boldsymbol{P}}_{\text{invalid}}^{m}(\boldsymbol{Q}_{\leq m})\text{ or }[P,\boldsymbol{Q}_{m+1}]\neq0\},\label{eq:Pinvalid_next}
\end{align}
where $\boldsymbol{Q}_{m+1}=\boldsymbol{S}_{\text{right}}^{m+1}\cap\tilde{\boldsymbol{P}}_{m+1}$ is determined by $\boldsymbol{S}_{\text{right}}^{m+1}$. The derivations can be found in Appendix \ref{appendix:SM}. For given $A_{m}$ and $\boldsymbol{S}_{\text{right}}^{m+1}$, $A_{m+1}$ is determined via Eq. \eqref{eq:Sproj_next} and Eq. \eqref{eq:Pinvalid_next}, and we denote it as
\begin{equation}
A_{m+1}=A_{m+1}(A_{m},\boldsymbol{S}_{\text{right}}^{m+1}).\label{eq:Am+1_next_Sright}
\end{equation}
This suggests that, to obtain the transition function $F(A_{m})$, we need to determine all possible $\boldsymbol{S}_{\text{right}}^{m+1}$ given $A_{m}$:
\begin{equation}
\mathcal{S}_{\text{right}}^{m+1}(A_{m})=\{\boldsymbol{S}_{\text{right}}^{m+1}|A_{m}\}\equiv\{\boldsymbol{S}_{\text{right}}^{m+1}(\boldsymbol{Q})|A_{m}=A_{m}(\boldsymbol{Q}),\boldsymbol{Q}\in\mathscr{S}(\boldsymbol{P})\}.
\end{equation}
Thus, the transition function $F(A_m)$ can be written as 
\begin{equation}
F(A_{m})=\{A_{m+1}(A_{m},\boldsymbol{S}_{\text{right}}^{m+1})|\boldsymbol{S}_{\text{right}}^{m+1}\in\mathcal{S}_{\text{right}}^{m+1}(A_{m})\}.\label{eq:F_Am}
\end{equation}
However, the exact computation of $\mathcal{S}_{\text{right}}^{m+1}(A_{m})$ is still exponentially scaled. A solution to this issue is proposed by relaxing the definition of state machines and transition functions as follows:
\begin{defn}
\thlabel{relaxed_F} (Relaxed state machine and transition function) Let $A_{m}$ be a state machine defined on \textbf{$\boldsymbol{D}$} with transition function $F$, and $\mathcal{A}_{m}=\{A_{m}(\boldsymbol{s})|\boldsymbol{s}\in\boldsymbol{D}\}$. For an enlarged space $\tilde{\mathcal{A}}_{m}\supseteq\mathcal{A}_{m}$, we say $A_{m}\in\tilde{\mathcal{A}}_{m}$ is a \textbf{relaxed state machine}, and $\tilde{F}(A_{m}\in\tilde{\mathcal{A}}_{m})$ is a \textbf{relaxed transition function} if
\end{defn}
\begin{enumerate}
\item $\tilde{F}(A_{m}\in\mathcal{A}_{m})\cap\mathcal{A}_{m+1}=F(A_{m})$, i.e. $\tilde{F}$ behaves the same with $F$ in the ``valid region'' $\mathcal{A}_{m}$.
\item $\tilde{F}(A_{m}\in(\tilde{\mathcal{A}}_{m}-\mathcal{A}_{m}))\subseteq\tilde{\mathcal{A}}_{m+1}-\mathcal{A}_{m+1}$, i.e. if $A_{m+1}\in\tilde{\mathcal{A}}_{m+1}$, then $A_{m}\in\tilde{\mathcal{A}}_{m}$
\end{enumerate}
According to this definition, once some path $\{A_{1},A_{2},...\}$ enters into the ``invalid region'' at some $A_{m}$, i.e., $A_{m}\in\tilde{\mathcal{A}}_{m}-\mathcal{A}_{m}$, all the following $\{A_{m+1},...\}$  will stay in the invalid region and the path could even be early terminated at step $t < n$ when $\tilde{F}(A_{t})=\emptyset$. To assure each full path $\{A_{1},...,A_{n}\}$ always staying in the valid region, the condition of $\tilde{\mathcal{A}}_{n}-\mathcal{A}_{n}=\emptyset$ is additionally required. With this additional condition, the replacement of $F$ to the relaxed definitions, $\tilde{F}$, in \thref{classical_gs} should still result in the ground state. We refer to this replaced algorithm as the relaxed state machine algorithm, whose schematic diagram is shown in Fig. \ref{fig:SM_relaxed}.
\begin{cor}
\thlabel{relaxed_gs} (Relaxed state machine algorithm of contrained 1D local classical ground state) Following the notations in \thref{relaxed_F}, we replace $F$ by $\tilde{F}$ and replace the initial $\mathcal{A}_{1}$ by some $\tilde{\mathcal{A}}_{1}\supseteq\mathcal{A}_{1}$ in \thref{classical_gs}, and generate $\tilde{\mathcal{A}}_{m}$, $m=1,2,....,n$ with the same procedure. If we additionally have $\tilde{\mathcal{A}}_{n}=\mathcal{A}_{n}$, then the modified algorithm still gives the correct ground state.
\end{cor}
\begin{figure}[tbh]
\centering

\includegraphics[width=0.7\linewidth]{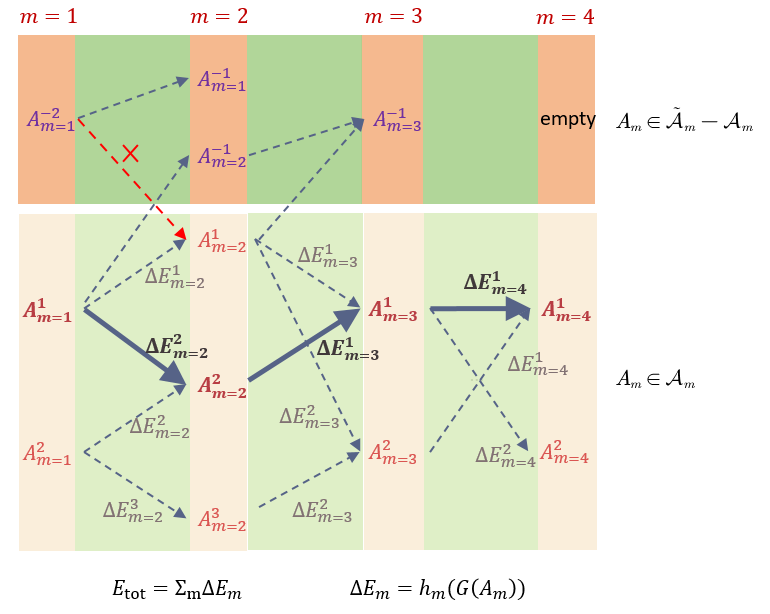}

\caption{Schematic diagram of the relaxed state machine algorithm to obtain the ground state of constrained 1D local classical Hamiltonians in \thref{relaxed_gs}. Given a relaxed state machine $A_{m}$ and a relaxed transition function $\tilde{F}$ defined in \thref{relaxed_F}, the full space of the relaxed state machine $\tilde{\mathcal{A}}_{m}$ is divided to a valid part $\mathcal{A}_{m}$ and a non-valid part $\tilde{\mathcal{A}}_{m}-\mathcal{A}_{m}$. $\tilde{F}$ can connect $A_{m}\in\tilde{\mathcal{A}}_{m}$ with $A_{m+1}\in\tilde{\mathcal{A}}_{m+1}$, $A_{m}\in\tilde{\mathcal{A}}_{m}$ with $A_{m+1}\protect\notin\tilde{\mathcal{A}}_{m+1}$, $A_{m}\protect\notin\tilde{\mathcal{A}}_{m}$ with $A_{m+1}\protect\notin\tilde{\mathcal{A}}_{m+1}$, but not $A_{m}\protect\notin\tilde{\mathcal{A}}_{m}$ with $A_{m+1}\in\tilde{\mathcal{A}}_{m+1}$ (red dashed line with a ``$\times$'' sign). In the valid part $\mathcal{A}_{m}$, the behavior of $\tilde{F}$ should be the same as the original transition function $F$. (the bottom part is the same as Fig. \ref{fig:SM}. If there's no invalid state machine at $m=n$, i.e. $\tilde{\mathcal{A}}_{m}=\mathcal{A}_{m}$, then the ground state can be obtained by the same approach in \thref{classical_gs} with the relaxed transition function $\tilde{F}$. (once the path $A_{1},A_{2},...$ enters into the top part is never reaches $m=n$ so does not affect the result) With the mapping of the stabilizer ground problem to the constrained classical ground state problem in Sec. \ref{sec:mapping}, and the constructed relaxed state machine $A_{m}$, relaxed transition function $\tilde{F}$, and state function $G$, one can obtain the stabilizer ground state of 1D local Hamiltonians with a linear scaling. (See \thref{relaxed_gs_Q})}

\label{fig:SM_relaxed}
\end{figure}

We now revisit the stabilizer ground state problem, where an enlarged space $\mathcal{\tilde{S}}_{\text{right}}^{m+1}(A_{m})\supseteq\mathcal{S}_{\text{right}}^{m+1}(A_{m})$ needs to be created. By implying $\tilde{\mathcal{A}}_m\supseteq \mathcal{A}_m$ for each $m$, the corresponding construction (Eq. \eqref{eq:F_Am_relaxed}) automatically satisfies first condition of \thref{relaxed_F}.
\begin{equation}
\tilde{F}(A_{m})=\{A_{m+1}(A_{m},\boldsymbol{S}_{\text{right}}^{m+1})|\boldsymbol{S}_{\text{right}}^{m+1}\in\tilde{\mathcal{S}}_{\text{right}}^{m+1}(A_{m})\}.\label{eq:F_Am_relaxed}
\end{equation}
In Appendix \ref{appendix:proof_F}, we show that the second condition of \thref{relaxed_F} can be achieved if each $\boldsymbol{S}_{\text{right}}^{m+1}\in\tilde{\mathcal{S}}_{\text{right}}^{m+1}(A_{m})$ satisfies
\begin{equation}
\begin{aligned}-I & \notin\langle\boldsymbol{Q}_{m+1}\rangle,\\
\boldsymbol{Q}_{m+1}\cap\tilde{\boldsymbol{P}}_{\text{invalid}}^{m}(\boldsymbol{Q}_{\leq m}) & =\emptyset,\\
\langle\mathbb{P}_{\leq m}(\boldsymbol{S}_{\text{right}}^{m+1}),\boldsymbol{Q}_{m}\rangle & =\boldsymbol{S}_{\text{right}}^{m},\\
\boldsymbol{S}_{\text{proj}}^{m+1}\cap\boldsymbol{S}_{\text{right}}^{m+1} & =\boldsymbol{Q}_{m+1},
\end{aligned}
\label{eq:Sright_next_conditions}
\end{equation}
where we used $\boldsymbol{Q}_{m}=\boldsymbol{S}_{\text{right}}^{m}\cap\boldsymbol{P}_{m}$, $\boldsymbol{Q}_{m+1}=\boldsymbol{S}_{\text{right}}^{m+1}\cap\boldsymbol{P}_{m+1}$, and $\boldsymbol{S}_{\text{proj}}^{m+1}=\boldsymbol{S}_{\text{proj}}^{m+1}(\boldsymbol{S}_{\text{proj}}^{m},\boldsymbol{Q}_{m+1})$ defined in Eq. \eqref{eq:Sproj_next}. Thus, $\tilde{\mathcal{S}}_{\text{right}}^{m+1}(A_{m})$ could be constructed by looping over all stabilizer groups between qubit $m-k+2$ and $m+1$, and keeping only the terms satisfying Eq. \eqref{eq:Sright_next_conditions}. However, this apporach has a non-ideal upper bound $|\tilde{\mathcal{S}}_{\text{right}}^{m}|\leq\mathcal{S}(k)\sim2^{\frac{1}{2}(k+1)(k+2)}$. We note that if $\tilde{F}$ on $\tilde{\mathcal{A}}_{m}\supseteq\mathcal{A}_{m}$ is a relaxed transition function, $\tilde{F}$ on any $\tilde{\mathcal{A}}_{m}^{\prime}$ with $\tilde{\mathcal{A}}_{m}\supseteq\tilde{\mathcal{A}}_{m}^{\prime}\supseteq\mathcal{A}_{m}$ is automatically a relaxed transition function (Fig. \ref{fig:SM_relaxed}).  Thanks to this fact, a better solution can be obtained with the truncation operation defined as follows:
\begin{defn} \thlabel{def:truncation} (Truncation)
Let $P\in\mathcal{P}_{n}$ be a Pauli operator, and $P=\pm p_{1}\otimes p_{2}\otimes...\otimes p_{n}$, where $p_{i}\in\{I_{i},X_{i},Y_{i},Z_{i}\}$. The \textbf{truncation} of $P$ to qubits $I$ is $\mathbb{T}_{I}(P)=\otimes_{i\in I}p_{i}$. Similarly, the truncation of a set of Pauli operators $\boldsymbol{P}$ is $\mathbb{T}_{I}(\boldsymbol{P})=\{\mathbb{T}_{I}(P)|P\in\boldsymbol{P}\}$.
\end{defn}
A tighter constraint is $\boldsymbol{S}_{\text{right}}^{m}(\boldsymbol{Q}\in\mathscr{S}(\boldsymbol{P}))\in\mathcal{T}_{\text{right}}^{m}$, where
\begin{equation}
\mathcal{T}_{\text{right}}^{m}=\{\langle\boldsymbol{Q}\rangle|\boldsymbol{Q}\in\mathscr{S}(\mathbb{T}_{\leq m}(\boldsymbol{P}_{\geq m}))\},
\end{equation}
which are stabilizer groups generated by elements in $\pm \mathbb{T}_{\leq m}(\boldsymbol{P}_{\geq m})$. For example, if a stabilizer group $\boldsymbol{S}$ is generated by elements in $\boldsymbol{P}=\{X_{1}X_{2},Z_{1}Z_{2},Z_{2}\}$, $\boldsymbol{S}_{1}=\mathbb{P}_{1}(\boldsymbol{S})$ is generated by elements in $\mathbb{T}_{1}(\boldsymbol{P})=\{X_{1},Z_{1}\}$. Note that the converse statement is not true, i.e. $\boldsymbol{S}_{1}$ generated by $\mathbb{T}_{1}(\boldsymbol{P})$ is not necessarily some $\mathbb{P}_{1}(\boldsymbol{S})$. Clearly the number of elements in $\mathbb{T}_{\leq m}(\boldsymbol{P}_{\geq m})$ is up to the number of elements in $\boldsymbol{P}_{m,m+k-1}$.  This leads to a $O(\exp(Ck\log k))$ scaling for some constant $C$ given the Hamiltonian is sparse, which will be shown in \thref{local_nS}.

Finally, we define
\begin{equation}
\tilde{\mathcal{S}}_{\text{right}}^{m+1}(A_{m})=\{\boldsymbol{S}_{\text{right}}^{m}\in\mathcal{T}_{\text{right}}^{m}|\boldsymbol{S}_{\text{right}}^{m}\text{ satisfies Eq. \eqref{eq:Sright_next_conditions}}\}.\label{eq:relaxed_Sright_next}
\end{equation}
The corresponding relaxed transtion function $\tilde{F}(A_{m})$ is now given by 
\begin{cor}
\thlabel{relaxed_F_Q} (Relaxed transition function of the stabilizer ground state problem) $\tilde{F}(A_{m})$ defined in Eq. \eqref{eq:F_Am_relaxed} with $\tilde{\mathcal{S}}_{\text{right}}^{m+1}(A_{m})$ defined in Eq. \eqref{eq:relaxed_Sright_next} is a relaxed transition function of the state machine $A_{m}$ defined in \thref{SM_Q}.
\end{cor}
The last component we missed is an approach to determine the initial values $\tilde{\mathcal{A}}_{1}$. A clever construction is $\tilde{A}_{1}=\cup_{A_{0}\in\tilde{A}_{0}}F(A_{0})$ from $\tilde{A}_{0}$ by adding an ancillary site $0$ with $\boldsymbol{Q}_{0}=\boldsymbol{P}_{0}=\emptyset$. With this ancillary site, we can derive $\boldsymbol{S}_{\text{proj}}^{0}=\mathbb{P}_{>-k}(\langle\boldsymbol{Q}_{\leq0}\rangle)=\langle\emptyset\rangle$, $\tilde{\boldsymbol{P}}_{\text{invalid}}^{0}=\{P\in\tilde{\boldsymbol{P}}_{>0}|[P,\boldsymbol{Q}_{\leq0}]\neq0\}=\emptyset$, $\boldsymbol{S}_{\text{right}}^{0}=\mathbb{P}_{\leq0}(\langle\boldsymbol{Q}_{\geq0}\rangle)=\langle\emptyset\rangle$, and thus $\tilde{\mathcal{A}}_{0}= \{A_0\}$, where $A_{0}=(\langle\emptyset\rangle,\emptyset,\langle\emptyset\rangle)$. This naturally gives $\tilde{\mathcal{A}}_{1}=\tilde{F}(A_{0}=(\langle\emptyset\rangle,\emptyset,\langle\emptyset\rangle))$. According to \thref{relaxed_gs}, the final exact 1D local algorithm is given by:
\begin{cor}
\thlabel{relaxed_gs_Q} (Exact 1D local algorithm of the stabilizer ground state problem) With the state machine $A_{m}$, relaxed transition function $\tilde{F}$, and state function $G$ defined in \thref{SM_Q}, \thref{relaxed_F_Q} and Eq. \eqref{eq:G}, starting from $\tilde{\mathcal{A}}_{1}=\tilde{F}(A_{0}=(\langle\emptyset\rangle,\emptyset,\langle\emptyset\rangle))$, the stabilizer ground state of 1D local Hamiltonians can be determined by \thref{relaxed_gs}.
\end{cor}

\subsubsection{Computational complexity \label{sec:complexity}}

We estimate the computational scaling of the exact 1D local algorithm, which is controlled by $|\tilde{\mathcal{A}}_{m}|$. For the similar reason mentioned in Sec. \ref{sec:general}, it does not have a simple expression. A loose upper bound is provided in the following and its detailed proof is given in Appendix \ref{appendix:proof_local_nS}.
\begin{thm}
\thlabel{local_nS} (Upper bound of $|\tilde{\mathcal{A}}_{m}|$) Let $H=\sum_{P\in\boldsymbol{P}}w_{P}P$ be a $k$-local Hamiltonian, and there exists $M$ such that $|\boldsymbol{P}_{m}|\leq M$ for each $m$. Then for any $m$, we can construct candidate values of $A_{m}$ as $\tilde{\mathcal{A}}_{m}^{\prime}\supseteq\tilde{\mathcal{A}}_{m}$ solely from $\mathbb{T}_{m-2k+1,m}(\tilde{\boldsymbol{P}})$, such that $|\tilde{\mathcal{A}}_{m}^{\prime}|<N_{A}=(4kM)^{3k}=\exp(3k\log4kM)$.
\end{thm}
We rewrite $|\tilde{\mathcal{A}}_{m}|\sim O(\exp(Ck\log kM))$ for simplicity and further analysis. The total cost spent on each site m is bounded by the product of (1) $|\tilde{\mathcal{A}}_{m}|$, (2) $|\tilde{\mathcal{A}}_{m+1}|$, and (3) time to compute a single $A_{m+1}\in\tilde{\mathcal{A}}_{m+1}$ from a single $A_{m}\in\tilde{\mathcal{A}}_{m}$. Since all stabilizer state operations used in the algorithm can be realized in polynomial scaling of $k$, the total cost scales as
\[
\begin{aligned}T & \sim n\times O(\exp(Ck\log kM))\times O(\exp(Ck\log kM))\times O(\text{poly}(k))\\
 & \sim O(n\exp(C^{\prime}k\log kM))
\end{aligned}
\]
for some constant $C^{\prime}$. Similar to \thref{def:sparse}, we consider the 1D local and sparse Hamiltonians defined as follows:
\begin{defn}
\thlabel{def:local_sparse} (local sparse Hamiltonians) For a 1D $k$-local Pauli Hamiltonian $H=\sum_{P\in\boldsymbol{P}}w_{P}P$, we say that it is \textbf{sparse} if $|\boldsymbol{P}_{m}|\sim O(\text{poly}(k))$.
\end{defn}
The conclusions $T\sim O(n\exp(C^{\prime}k\log kM))$ and $M\sim O(\text{poly}(k))$ lead to the final computation complexity as:
\begin{cor}
\thlabel{local_nS_scaling} (Computational complexity of the exact 1D local algorithm) The computational cost $T$ to obtain the stabilizer ground state of a $n$-qubit, $k$-local sparse Hamiltonian is $T\sim O(n\exp(Ck\log k))$ for some constant $C$.
\end{cor}

\subsection{Stabilizer ground states of infinite periodic Hamiltonians} \label{sec:periodic}

In this section, the stabilizer ground state problem of infinite periodic Hamiltonians (referred to as periodic Hamiltonians) is discussed.
We will show that, for any 1D periodic local Hamiltonian, the stabilizer ground state also has periodic stabilizers in some supercells.
We also show that, the stabilizer ground states of 1D periodic local Hamiltonians can be similarly obtained by the state machine solution in the exact 1D algorithm with an additional periodic boundary condition of state machines in a supercell.
This algorithm is referred as the exact 1D periodic local algorithm.
For general periodic Hamiltonians in higher dimensions, we conjecture that the stabilizer ground states should still have periodic stabilizers.
With this assumption, the formalism in Sec.~\ref{sec:general} is extended to general periodic Hamiltonians.

We first review the properties of the eigenstates of periodic Hamiltonians on an infinitely long 1D lattice. 
Let $T$ be an operator to translate a given $|\psi\rangle$ by some fixed number of sites, and $H$ be a Hamiltonian satisfying $[H,T]=0$. 
Bloch's theorem~\cite{bloch1929quantenmechanik} states that the eigenstates $|\psi\rangle$ of $H$ can be classified
by the eigenvalue of $T$ via $T|\psi\rangle=e^{i\phi}|\psi\rangle$ since $H$ and $T$ can be simultaneously diagonalized.
In numerical treatments, a supercell with size $L$ is usually introduced
and $\phi$ can take discrete values $\phi=2\pi\frac{j}{L}$ for integers
$0\leq j<L$ \cite{kratzer2019basics,martin2020electronic}.

Now we consider stabilizer states in the qubit space. Let $T_{l}$ be the operator to translate qubit $q$ to $q+l$ for any $q$.
A Hamiltonian $H=\sum_{P\in\boldsymbol{P}}w_{P}P$ is defined to be invariant
under $T_{l}$ if $P^{\prime}=T^{\dagger}PT$
satisfies $P^{\prime}\in\boldsymbol{P}$ and $w_{P^{\prime}}=w_{P}$ for any $P\in\boldsymbol{P}$.
However, the stabilizer ground state $|\psi\rangle$ of $H$ might not be an
eigenstate of $T_{l}$, i.e., $T_{l}|\psi\rangle=\lambda|\psi\rangle$
for some $\lambda$. An example is $H=H_0+\epsilon H_I$, where $H_0=-\sum_{n}(X_{3n}X_{3n+1}X_{3n+2}+Z_{3n}Z_{3n+1}Z_{3n+2})$ and $H_I=-\sum_{n}(Z_{3n-1}X_{3n}+X_{3n-1}Z_{3n})$, and $H$ thus has a period of 3.
At $\epsilon=0$, $H$ can be divided into independent subsystems $\{3n,3n+1,3n+2\}$ for each $n$, and each subsystem has degenerate stabilizer ground states with stabilizers $X_{3n}X_{3n+1}X_{3n+2}$ and $Z_{3n}Z_{3n+1}Z_{3n+2}$, respectively. 
At $0<\epsilon\ll1$,
the interaction term $H_I$ breaks the degeneracy, and the stabilizers of
the stabilizer ground state become alternating $X_{3n}X_{3n+1}X_{3n+2}$ and $Z_{3n}Z_{3n+1}Z_{3n+2}$.
This system has a period of 6 and thus does not satisfy $T_{l=3}|\psi\rangle=\lambda|\psi\rangle$. 

Now let us treat the infinite periodic Hamiltonians as $n$-qubit Hamiltonians with $n=\infty$. According to Sec. \ref{sec:1D}, we can efficiently determine the ground state by constructing the state machines $\{A_{m}\}$ from the relaxed transition function $\tilde{F}$, so we get an infinitely long state machine chain $\{A_{m=-\infty}^{\infty}\}$. We introduce an equivalence condition $A\simeq B$, if $A$ and $B$ differ only by translation of some $cl$ lattices, where $c\in\mathbb{Z}$, e.g. $\langle Z_{1},Z_{2}\rangle\simeq\langle Z_{7},Z_{8}\rangle$ with $l=3$. The following theorem presents that the state machine chain $\{A_{m=-\infty}^{\infty}\}$ corresponding to the stabilizer ground state is periodic with some period $cl$, $c\in\mathbb{Z}$. The proof is given in Appendix \ref{appendix:proof_periodic}.
\begin{thm}
\thlabel{periodic_SM} (State machines of CMCS of infinitely periodic Hamiltonians are periodic) Let $H=\sum_{P\in\boldsymbol{P}}w_{P}P$ be an infinitely periodic 1D local sparse Hamiltonian with period $l$, and $\boldsymbol{Q}^{\star}\in\mathscr{S}(\boldsymbol{P})$ be the CMCS of $H$. Let $A_{m}=A_{m}(\boldsymbol{Q}^{\star})$, there exists some $c<N_{A}$ such that $A_{m+cl}\simeq A_{m}$ for any $m$, where $N_{A}$ is given in \thref{local_nS}. Furthermore, the stabilizer ground state energy per site is
\begin{equation}
E_{\text{gs}}^{\text{periodic}}=\min_{c\leq N_{A}}\min_{\{A_{m=0}^{cl}|A_{i+1}\in\tilde{F}(A_{i}),A_{0}\simeq A_{cl}\}}\frac{1}{cl}\sum_{m=1}^{cl}h_{m}(G(A_{m})),\label{eq:gs_periodic}
\end{equation}
which involves a periodic boundary condition $A_{0}\simeq A_{cl}$ , and the boundary state machine must be one of the candidate values in \thref{local_nS}.
\end{thm}
For a fixed supercell size $c$ and boundary state machine $A_{0}$, one can perform the algorithm in \thref{classical_gs} with an additional final state restriction of $A_{cl}\simeq A_{0}$ to obtain the chain $\{A_{0},A_{1},...,A_{cl}\}$ with the lowest energy. With additional loops over $c$ and the boundary state machine $A_{0}$, one can find the stabilizer ground state with a linear scaling of $l$. This algorithm is referred to as the \textbf{exact 1D periodic local algorithm}. Besides, \thref{periodic_SM} implies that the CMCS $\boldsymbol{Q}^{\star}$ is also periodic, and thus the stabilizer ground state $|\psi\rangle$ is also periodic, i.e. $T_{cl}|\psi\rangle=|\psi\rangle$.

We then move on to discuss higher-dimensional Hamiltonians. 
Although we are not able to theoretically prove that the stabilizer ground state $|\psi\rangle$ satisfies $T_{cl}|\psi\rangle=|\psi\rangle$ for translation operators $T_{cl}$ in each dimension, we speculate that at least some approximated (if not exact) stabilizer ground state satisfies such condition due to its similarity to the supercell treatments of
exact eigenstates. 

With the assumption of $T_{cl}|\psi\rangle=|\psi\rangle$, we present the stabilizer ground state theory for general infinite periodic Hamiltonians.
For any stabilizer $P$ with $P|\psi\rangle=|\psi\rangle$, let $P^{\prime}=T_{cl}^{\dagger}PT_{cl}$,  we have $P^{\prime}|\psi\rangle=|\psi\rangle$, and thus $P^{\prime}$ is also a stabilizer of $|\psi\rangle$. 
We only need to consider those $\boldsymbol{Q}\in\mathscr{S}(\boldsymbol{P})$ such that $\boldsymbol{Q}=T_{cl}^{\dagger}\boldsymbol{Q}T_{cl}$.
(Roughly speaking, we only need to determine the stabilizers $\boldsymbol{Q}$ in a single supercell and then copy it to others)
Strictly, we define the \textbf{closed commuting periodic subsets} (CCPS) as
\begin{equation}
\mathscr{S}_{c}(\boldsymbol{P})=\{\boldsymbol{Q}\subseteq\tilde{\boldsymbol{P}}|\boldsymbol{Q}=\langle\boldsymbol{Q}\rangle\cap\tilde{\boldsymbol{P}},-I\notin\langle\boldsymbol{Q}\rangle,\boldsymbol{Q}=T_{cl}^{\dagger}\boldsymbol{Q}T_{cl}\}.
\end{equation}
Since $[Q,T_{cl}^{\dagger}QT_{cl}]=0$ for any $Q\in\boldsymbol{Q}$, $\boldsymbol{Q}\in\mathscr{S}_{c}(\boldsymbol{P})$, a na\"ive but useful simplification of $\mathscr{S}_{c}(\boldsymbol{P})$ is
\begin{equation} \label{eq:periodic_simplification}
\mathscr{S}_{c}(\boldsymbol{P})=\mathscr{S}_{c}(\boldsymbol{P}^{\prime}),
\end{equation}
where $\boldsymbol{P}^{\prime}=\{P\in\boldsymbol{P}|[P,T_{cl}^{\dagger}PT_{cl}]=0\}$.
Consider the example $H=\sum_{i}Z_{i}X_{i+1}$ which has period $l=1$. If the supercell size is $c=1$, then according to $\{Z_{i}X_{i+1},Z_{i+1}X_{i+2}\}=0$, we have $\boldsymbol{P}^{\prime}=\emptyset$ and thus $\mathscr{S}_{c=1}(\boldsymbol{P})=\{\emptyset\}$.
This indicates that $c=1$ is not a good choice.
If the supercell size is $c=2$, we could, for example have $\boldsymbol{Q}=\{Z_{2i}X_{2i+1}|i\in\mathbb{Z}\}$.
Finally, the stabilizer ground state is given by
\begin{equation} \label{eq:periodic_stab_gs}
E_{\text{gs}}=\min_{c, \boldsymbol{Q}\in\mathscr{S}_{c}(\boldsymbol{P})}E_{\text{stab}}(H,\langle\boldsymbol{Q}\rangle).
\end{equation}
The $\boldsymbol{Q}\in\mathscr{S}_{c}(\boldsymbol{P})$ minimizing Eq.~\eqref{eq:periodic_stab_gs} is referred as the \textbf{closed maximally-commuting periodic subset} (CMCPS) of $\tilde{\boldsymbol{P}}$ (or $H$).
For physically reasonable Hamiltonians, it might be enough to search for the minimum stabilizer ground state energy by checking a few small $c$.

\subsection{Applications analysis of stabilizer ground states} \label{sec:applications}

In the following, we discuss the potential applications of the concept and algorithms of stabilizer ground states in simulating many-body physics problems from the aspects of both quantum and classical algorithms. The stabilizer ground states are also compared with other common ground states ansatzes, e.g. the mean-field state, and matrix product states (MPS). Since the concept of "mean-field" has different definitions in different situations, we only consider the product states \cite{gharibian2012approximation,anshu2020beyond}, which is a typical example of mean-field states in the qubit space. However, a similar analysis also works for many other cases. We will also see that, although stabilizer states, product states, and MPS each have their own advantages depending on the physical system, their strengths are not mutually exclusive. In fact, one can seamlessly combine stabilizer states with product states or MPS to leverage the benefits of both, especially for systems that exhibit features characteristic of both classes.

Stabilizer ground states could serve as good candidates for initial states in terms of quantum algorithms. Specifically, we consider two quantum algorithms for many-body ground state problems, i.e., the variational quantum eigensolver (VQE) \cite{cerezo2022variational,tilly2022variational} and quantum phase estimation (QPE) \cite{kitaev1995quantum,kitaev2002classical}, which are the most widely studied algorithms in the noisy intermediate-scale quantum (NISQ) \cite{preskill2018quantum} and fault-tolerant quantum computation (FTQC) \cite{shor1996fault} era, respectively. 
For VQE, a notorious issue is the barren plateau problem \cite{cerezo2021cost,mcclean2018barren,Zhang2023b_z},
implying that the energy gradients with respect to the circuit parameters are exponentially small except in an exponentially small region, and such behavior resulting in the failure of classical optimizations \cite{mcclean2018barren}.
Thus, starting from a physically valid state, such as an approximated ground state in the case of stabilizer ground state, serves as the easiest and the most straightforward solution to mitigate the barren plateau problem. In terms of the QPE algorithm, the output is a random eigenstate. Let the initial state be $|\psi\rangle_{\text{init}}$, the probability of getting each eigenstate $|\psi_{k}\rangle$ follows Born rule $|\langle\psi_{\text{init}}|\psi_{k}\rangle|^{2}$. For randomly chosen $|\psi_{\text{init}}\rangle$, $|\langle\psi_{\text{init}}|\psi_{k}\rangle|^{2}$ also decays exponentially with the system size \cite{lee2023evaluating}. Thus, an initial state having a large overlap with the exact ground state is also vital for the success of QPE. In summary, for both VQE and QPE, an initial state close to the ground state is desired, where stabilizer ground states and the corresponding efficient algorithms serve as good candidates. \cite{robin2025stabilizer}

One of the main advantages of the stabilizer states for initial states on quantum computers is its ability to be efficiently prepared on quantum circuits with up to $O(n^{2}/\log n)$ single-qubit and double-qubit Clifford gates \cite{aaronson2004improved}. Such advantage is further magnified on fault-tolerant quantum computers, where the computational cost is dominated by non-Clifford operations via, e.g., magic state distillation \cite{campbell2017roads}. Besides, we prove in the following that, the gate count number can be further reduced to $O(nk/\log k)$ for the preparation of stabilizer ground states of 1D $k$-local Hamiltonians (assuming the stabilizer ground state is not highly degenerate).
\begin{thm}
\thlabel{local_preparation} (Decomposed Clifford transformations of local stabilizer groups) Given a $n$-qubit stabilizer group $\boldsymbol{S}$, if there exists independent generators $\boldsymbol{P}=\{P_{1},...,P_{n}\}$ (i.e. $\boldsymbol{S}=\langle\boldsymbol{P}\rangle$) that each $P_{i}$ is $k$-local, then we can find up to $L\sim O(nk/\log k)$ single-qubit and double-qubit Clifford transformations $U_{1},...,U_{L}$, such that the combined Clifford transformation $U=\prod_{i=L}^{1}U_{i}$ satisfies $U^{\dagger}\boldsymbol{S}U=\langle Z_{1},Z_{2},...,Z_{n}\rangle$. Furthermore, if there are up to $s$ elements in $\{P_{i}\}$ that are not $k$-local, the above conclusion still holds with $L\sim O(nk^{\prime}/\log k^{\prime})$, where $k^{\prime}=k+s$.
\end{thm}
\begin{proof}
See Appendix \ref{appendix:proof_local_preparation}. The corresponding algorithm is given in Algorithm \ref{alg:local_preparation}.
\end{proof}
\begin{cor} \thlabel{stab_gs_preparation} (Quantum state preparation of stabilizer ground states of 1D local Hamiltonians) Let $H=\sum_{P\in\boldsymbol{P}}w_{P}P$ be a $n$-qubit , $k$-local 1D local Hamiltonian, and the corresponding CMCS is $\boldsymbol{Q}^{\star}\in\mathscr{S}(\boldsymbol{P})$. According to \thref{sparse_gs} the stabilizer ground states are those stabilizer states stabilized by $\langle\boldsymbol{Q}^{\star}\rangle$. If $\langle\boldsymbol{Q}^{\star}\rangle$ has $n-s$ independent generators, then each of the corresponding ground states can be prepared on quantum circuits with up to $L\sim O(nk^{\prime}/\log k^{\prime})$ single-qubit and double-qubit Clifford operations, where $k^{\prime}=k+s$. Specifically, if the stabilizer ground state is non-degenerate, one needs $L\sim O(nk/\log k)$ operations.
\end{cor}
\begin{proof}
Since each $Q\in\boldsymbol{Q}^{\star}$ is $k$-local, $\boldsymbol{S}=\langle\boldsymbol{Q}^{\text{\ensuremath{\star}}}\rangle$ has $n-s$ $k$-local and independent generators, thus it meets the requirement of \thref{local_preparation}, and the conclusion applies. Since each (degenerate) stabilizer ground state satisfies $P|\psi\rangle=|\psi\rangle$ for $P\in\boldsymbol{S}$, we have $U^{\dagger}PU(U^{\dagger}|\psi\rangle)=(U^{\dagger}|\psi\rangle)$. Since $U^{\dagger}\boldsymbol{S}U=\langle Z_{1},Z_{2},...,Z_{n}\rangle$, we have $|\psi^{\prime}\rangle=U^{\dagger}|\psi\rangle=|0\rangle^{\otimes n}$, or $|\psi\rangle=U|0\rangle^{\otimes n}$, i.e. $|\psi\rangle$ can be prepared by $U$, which contains $L$ single and double-qubit Clifford operations. When and only when $s=0$, the stabilizer ground state is unique, thus the non-degenerate stabilizer ground state implies $L\sim O(nk/\log k)$.
\end{proof}
We also compare it with product states and MPS in the following. Since product states can be trivially prepared in $O(1)$ circuit depth and $O(n)$ gates, we mainly discuss the preparation of MPS. MPS is a powerful common ground state ansatz especially for 1D local and gapped systems as the ground state satisfies the area-law entanglement. It has also been shown that an $n$-site MPS with bond dimension $\chi$ can be prepared on quantum computers with $n$ numbers of ($\log_{2}\chi+1$)-qubit unitaries \cite{schon2005sequential,perez2006matrix}. Since $m$-qubit gate can be decomposed to $O(4^{m})$ single-qubit and double-qubit gates \cite{barenco1995elementary}, this gives a total number of $O(n\chi^{2})$ single-qubit and double-qubit gates. (Recent work \cite{malz2024preparation} shows that the circuit depth can be reduced to $O(\log n)$ or $O(\log \log n)$ but with the price of higher scaling in terms of $\chi$, typically $O(\chi^4)$ or $O(\chi^6)$). In the worst case, MPS requires bond dimension $\chi=2^{k}$ to describe a $(k+1)$-local stabilizer state, which leads to an exponentially higher cost of circuit preparation in terms of dependence on the locality $k$. One example is the rainbow state $|\psi\rangle=\sum_{i_{1}....i_{k}}|i_{1}...i_{k}i_{1}...i_{k}\rangle$, which is essentially the product of $k$ bell pairs on qubits $(i,k+i)$, $i=1,...k$.

Even in systems where area-law entanglement holds and MPS or product states are effective—such as 1D gapped Hamiltonians—stabilizer ground states can still provide complementary benefits. One can exploit this by combining stabilizer states with MPS or product states in either the Schrödinger picture (e.g., treating stabilizer states as basis) or the Heisenberg picture (e.g., applying Clifford transformations to simplify the Hamiltonian). Hybrid approaches often offer more expressive power, particularly in systems that exhibit both short-range correlations and stabilizer-like structures. In Sec.~\ref{sec:extended}, we present an example where product states are used as basis transformations to extend stabilizer ground states, which we apply to describe the toric code model in the presence of external fields. Recent works further support the utility of hybrid strategies between stabilizer states and MPS or tensor network. For example, stabilizer tensor networks~\cite{masot2024stabilizer} use stabilizer states as local bases within tensor networks, and Clifford-augmented DMRG~\cite{qian2024augmenting} applies sequential Clifford transformations to enhance MPS representations. In the former case, stabilizer ground states naturally appear as a physical basis for ground state descriptions. In the latter, stabilizer ground states can be used to construct a Clifford transformation \( U_C \) such that \( U_C |0\rangle^{\otimes n} = |\psi_{\text{stab gs}}\rangle \), thereby simplifying the Hamiltonian before applying MPS methods. The various ways in which stabilizer states, product states, and MPS can be combined are illustrated schematically in Fig.~\ref{fig:venn}.

\begin{figure}[thb]
    \centering    \includegraphics[width=0.7\textwidth]{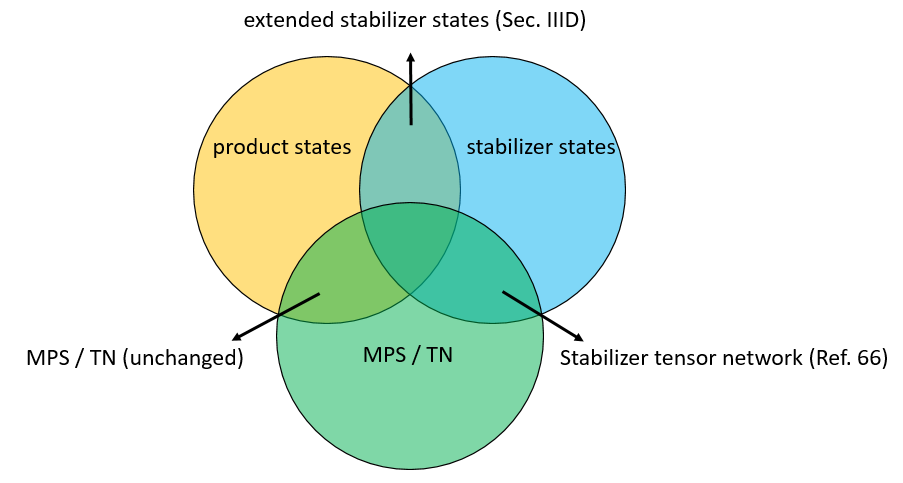}
    \caption{Conceptual relationships among stabilizer states, product states, and matrix product states (MPS) or tensor networks (TN). 
    The intersection of stabilizer states and product states corresponds to \emph{extended stabilizer states} (see Sec.~\ref{sec:extended}), i.e., stabilizer states expressed in rotated local bases. 
    The intersection of stabilizer states and tensor networks corresponds to \emph{stabilizer tensor networks}~\cite{masot2024stabilizer}, where stabilizer states serve as local building blocks of the network. 
    The intersection of product states and MPS/TN is trivially contained within MPS/TN, since product states correspond to MPS with bond dimension one and MPS naturally allow local basis rotations.}
    \label{fig:venn}
\end{figure}

In the context of classical algorithms, stabilizer ground states can be incorporated into both variational and perturbative methods. On the variational side, similar to their role in mitigating barren plateaus in VQE, stabilizer ground states can serve as building blocks for classical ansatzes where energy expectation values must be efficiently computable. While most classical variational methods operate in the computational basis—such as product states, MPS, or variational Monte Carlo (VMC)—the stabilizer basis provides a natural generalization. For example, VMC ansatzes can be extended from the form \( |\psi\rangle = \sum_{\boldsymbol{i}} \psi_{\boldsymbol{i}}(\boldsymbol{\theta}) |\boldsymbol{i}\rangle \) to  
\begin{equation}
|\psi\rangle = \sum_{\psi_{\text{stab}}} C(\psi_{\text{stab}}; \boldsymbol{\theta}) |\psi_{\text{stab}}\rangle,
\end{equation}
where the sum runs over all or a subset of stabilizer states, and \( C(\psi_{\text{stab}}; \boldsymbol{\theta}) \) is a parameterized coefficient function. Since overlaps \( \langle \psi_{\text{stab}}' | H | \psi_{\text{stab}} \rangle \) can be efficiently computed for Pauli Hamiltonians, such stabilizer-based VMC methods may be especially suitable for capturing long-range or strongly entangled features beyond those well-described by the computational basis.

Stabilizer ground states also offer a useful tool in perturbative approaches. These methods typically split the Hamiltonian as \( H = H_0 + H_1 \), where \( H_0 \) is exactly solvable and \( H_1 \) is treated as a perturbation. As discussed in Sec.~\ref{sec:general}, our algorithm identifies a CMCS \( \boldsymbol{Q}^\star \subset \boldsymbol{P} \) for a given Pauli Hamiltonian \( H = \sum_{P \in \boldsymbol{P}} w_P P \). This yields a natural decomposition where \( H_0 = \sum_{P \in \boldsymbol{P} \cap \pm \boldsymbol{Q}^\star} w_P P \), with the stabilizer ground state being the exact ground state of \( H_0 \). The remaining terms define \( H_1 \), and their effects can be treated perturbatively. This approach allows the stabilizer ground state algorithm to not only provide an approximate solution but also guide the \textit{optimal partitioning} of the Hamiltonian for perturbative expansion.
A related recent direction is the development of doped stabilizer states~\cite{gu2024doped}, which utilize the stabilizer nullity structure of a Hamiltonian for efficient approximation. Our method is complementary in that it provides a systematic framework to extract stabilizer-aligned structure even in non-stabilizer Hamiltonians.

\section{Results} \label{sec:results}

In this section, we first perform a few benchmarks on the exact 1D local algorithm, including the computational cost in Sec.~\ref{sec:cost} and the comparison with numerically optimized approximated stabilizer ground states in Sec.~\ref{sec:comparison}.
Furthermore, we demonstrate a few potential applications for stabilizer ground states and the corresponding algorithms on (1) simple qualitative analysis of phase transitions in Sec.~\ref{sec:phase}, (2) serving as the cornerstone of developing advanced ground state ansatzes in Sec.~\ref{sec:extended}, and (3) generation of initial states for VQE problems for better performance in Sec.~\ref{sec:vqe}.

All algorithms introduced in this work are implemented in both Python and C++ in \url{https://github.com/SUSYUSTC/stabilizer_gs}.
The Python code is presented for concept illustration and readability, and the C++ code is used for optimal performance with a simple parallelization.

\subsection{Computational cost of the exact 1D local algorithm} \label{sec:cost}

The scaling of the computation time for this exact 1D local algorithm only has a loose theoretical upper bound (\thref{local_nS} and \thref{local_nS_scaling}) and lacks an exact analytical formula. 
Therefore, we implement the algorithm in C++ and numerically benchmark the computational time.
All corresponding timings are collected on an 8-core i7-9700K Intel CPU.

We consider the following stochastic $k$-nearest Heisenberg model as the example Hamiltonian:
\begin{equation} \label{eq:stochastic_heisenberg}
    H = \sum_{i=1}^n \sum_{j=i+1}^{i+k-1} J^{xx}_{ij} S^x_i S^x_j + J^{yy}_{ij} S^y_i S^y_j + J^{zz}_{ij} S^z_i S^z_j
\end{equation}
with each $J^{xx}_{ij},J^{yy}_{ij},J^{zz}_{ij}\sim \mathcal{N}(0, 1)$, i.e. all these coupling coefficients independently follow the normal distribution. We also consider the case of $J^{zz}_{ij}=0$, $J^{xx}_{ij},J^{yy}_{ij}\sim \mathcal{N}(0, 1)$ for comparisons. The two models are referred to as \{XX,YY,ZZ\} and \{XX,YY\}, respectively. 

Following the procedure of the exact 1D local algorithm, all the possible values of $\{A_m\}$ with $m=1,2,...n+1$ are generated sequentially.
In Figure~\ref{fig:timing}(a), the computational time of generating $\{A_{m+1}\}$ from $\{A_m\}$ are plotted as a function of site $m$ for both the \{XX,YY,ZZ\} and \{XX,YY\} models with $n=25$ and $k=5$.
Except for a few sites near the boundaries, the wall-clock time spent at each site is almost a constant for both models.
This verifies that the computational cost of the 1D local algorithm scales as $O(n)$, as proved in \thref{local_nS}.
Due to the smaller number of Pauli terms in the Hamiltonian, the computational cost of the \{XX,YY\} model is systematically lower than the \{XX,YY,ZZ\} model.

\begin{figure}[thb]
    \centering    \includegraphics[width=1\textwidth]{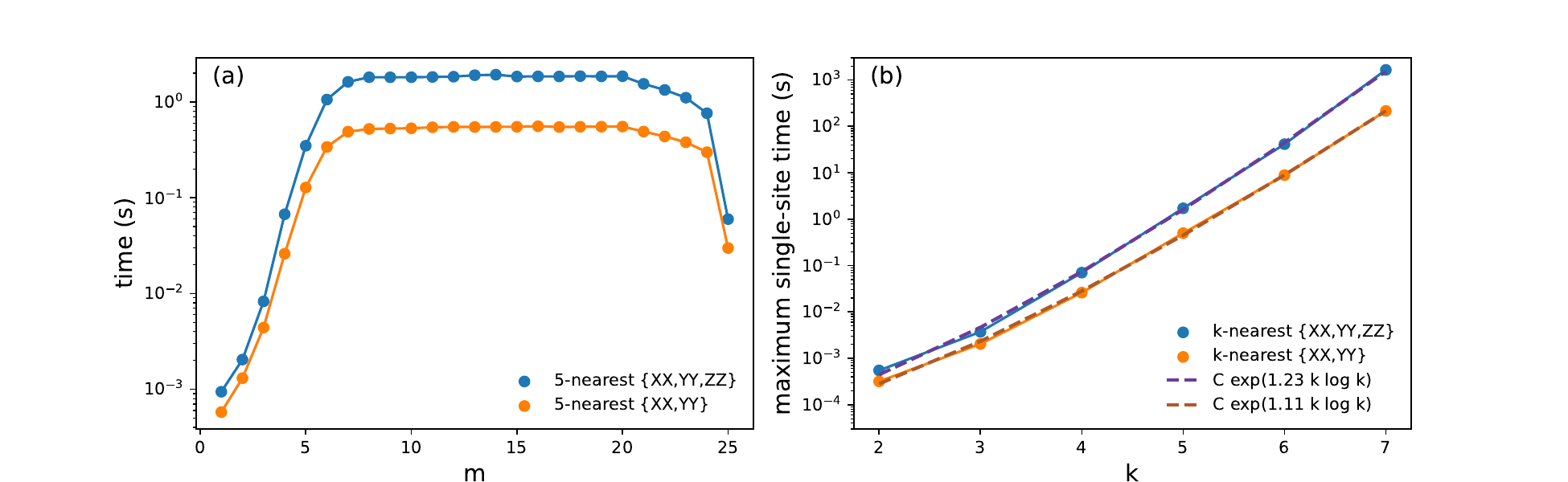}
    \caption{Wall-clock time of the 1D local algorithm on the \{XX,YY,ZZ\} and \{XX,YY\} model (see the main text). (a) Computational time spent on each site $m$ for the two models. (b) Maximum single-site computational time (solid lines) as a function of $k$ for the two models, and fitted curves (dashed lines) with the form of $C \exp (C^\prime k \log k)$.}
    \label{fig:timing}
\end{figure}

After showing that the time spent at each site is a constant except for the sites near the boundaries, we further discuss the scaling of the maximum single-site running time as a function of the locality parameter $k$ for different types of Hamiltonians. 
Figure~\ref{fig:timing}(b) displays the maximum wall-clock time of a single site as a function of $k$ for both models. 
We assume that the form of the scaling function is $C \exp (C^\prime k \log k)$ according to \thref{local_nS_scaling}.
The numerical scaling functions are fitted independently for two models in Figure~\ref{fig:timing}(b).
The resulting fitted scaling curves have the parameters $C^\prime=1.23$ and $C^\prime=1.11$ for the \{XX,YY,ZZ\} and \{XX,YY\} models, respectively, and match the true timing data well. This indicates that the computation time of different Hamiltonians within the same class scales similarly.

\subsection{Comparison with numerical discrete optimizations of stabilizer ground states} \label{sec:comparison}


We demonstrate that numerical optimizations of stabilizer ground states are not scalable and lead to unacceptable energy errors with an increasing number of qubits.
The numerical optimization of stabilizer ground states can be performed by discrete optimizations of the Clifford circuits representing stabilizer states \cite{CAFQA}.

Here, we still use the stochastic $k$-nearest Heisenberg Hamiltonian in Eq.~\eqref{eq:stochastic_heisenberg} (the \{XX,YY,ZZ\} model) as an example.
The Clifford ansatz employed here modifies the hardware-efficient Clifford ansatz in Ref.~\cite{CAFQA} by generalizing the single-qubit Clifford rotations to all single-qubit Clifford operations (24 unique choices in total) \cite{koenig2014efficiently}.
The simulated annealing algorithm is used in the discrete optimization with an exponential decay of temperature from 5 to 0.05 in 2500 steps.
In each step of the simulated annealing, one of the single-qubit Clifford operations is randomly selected and replaced with one of the 24 operations, and the move is accepted with a probability of $\min(\exp(-\Delta E / T), 1)$, where $\Delta E$ is the energy difference.
\begin{figure}[thb]
    \centering
    \includegraphics[width=0.5\textwidth]{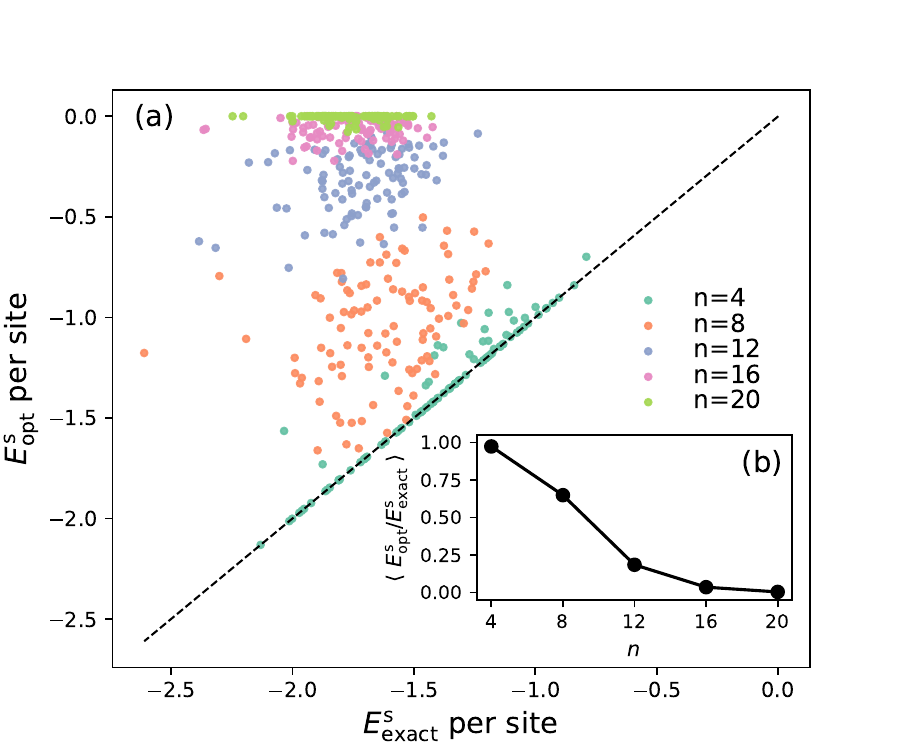}
    \caption{Stabilizer ground state energies of the stochastic $k$-nearest Heisenberg model obtained by the exact 1D local algorithm ($E_\text{exact}^\text{s}$) and the numerical simulated annealing optimization algorithm ($E_\text{opt}^\text{s}$). (a) $E_\text{exact}^\text{s}$ per site versus $E_\text{opt}^\text{s}$ per site with $n=4,8,12,16,20$ and $k=4$. The black dashed line corresponds to $E_\text{exact}^\text{s}=E_\text{opt}^\text{s}$. (b) The mean relative energy error of stabilizer ground state energy captured by the numerical optimization algorithm ($
    \langle E_\text{opt}^\text{s}/E_\text{exact}^\text{s}\rangle$) versus $n$ with locality $k=4$.}
    \label{fig:opt}
\end{figure}

Figure \ref{fig:opt}(a) compares the stabilizer ground state energies obtained from the exact 1D local algorithm ($E_\text{exact}^\text{s}$) and the numerical optimization algorithm ($E_\text{opt}^\text{s}$).
For each $n$, 100 random Hamiltonians are tested.
For every single test, the numerically optimized ground state energy is either equal to or higher than the exact stabilizer ground state energy.
With increasing $n$, the success probability of numerical optimization that results in accurate stabilizer ground state energies decreases and $E_\text{opt}^\text{s}$ approaches zero. 
This indicates that the numerical discrete optimization cannot correctly obtain the stabilizer ground state due to the exponential scaling of the number of stabilizer states and the number of possible Clifford circuits.
Figure~\ref{fig:opt}(b) displays the quantitative statistics of the performance degradation speed of numerical optimization by plotting the averaged relative stabilizer ground state energy $\langle E_\text{opt}^\text{s}/E_\text{exact}^\text{s} \rangle$ versus the number of sites $n$ with $k=4$.
A rapid decay of the energy ratio is observed from 97.4\% at $n=4$ to 0.4\% at $n=20$. 
Therefore, the optimization method fails to bootstrap large-scale variational quantum algorithms via stabilizer initializations as claimed in Ref. \cite{CAFQA} and the challenge is fully solved by our new algorithm at least in the 1D case.

\subsection{Qualitative analysis of phase transitions} \label{sec:phase}

\begin{figure}[hbt]
    \centering
    \includegraphics[width=0.7\linewidth]{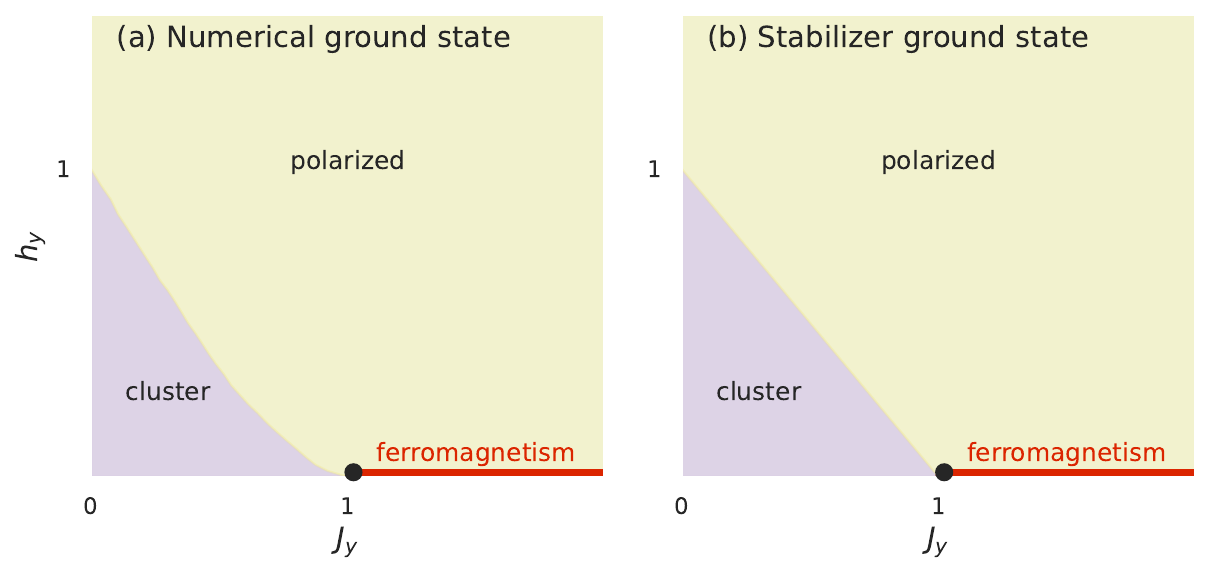}
    \caption{Ground state phase diagram of the Hamiltonian in Eq.~\eqref{eq:alpha-H} obtained by (a) numerical DMRG calculations in Ref.~\cite{PT_example} and (b) stabilizer ground state calculations via the exact 1D periodic local algorithm. There are three phases, including the cluster phase, the polarized phase, and the ferromagnetism phase, for both cases.}
    \label{fig:phase_cluster}
\end{figure}

Similar to the mean-field states, stabilizer ground states can also qualitatively capture the phases and phase transitions in many interesting systems. Specifically, stabilizer ground states are good at capturing topological phases with long-range entanglements.
This capability is demonstrated using an infinite 1D generalized cluster model~\cite{PT_example} as an example, whose Hamiltonian is
\begin{equation} \label{eq:alpha-H}
H=\sum_{n=-\infty}^\infty-X_{n-1}Z_{n}X_{n+1}-J_{y}Y_{n}Y_{n+1}+h_{y}Y_{n}.
\end{equation}

This model is equivalent to the free fermion model at $h_{y}=0$, while it is not dual to any free fermion model at $h_{y}\neq 0$ due to the lack of $Z_{2}$ symmetry. 
This Hamiltonian has been studied by numerical density matrix renormalization group (DMRG) calculations in Ref.~\cite{PT_example}, and the corresponding phase diagram is replotted in Figure ~\ref{fig:phase_cluster}(a).
Three phases are observed in this phase diagram, including the symmetry-protected topological phase at small but positive $J_{y}$ and $h_{y}$, the polarized phase at $J_{y}\rightarrow\infty$, $h_{y}\rightarrow\infty$, and the ferromagnetic phase at $h_{y}=0,J_{y}>1$.

We apply the exact 1D periodic local algorithm to the Hamiltonian to obtain the stabilizer ground states of this model using different parameters $(J_{y},h_{y})$.
Calculations are performed with candidate supercell sizes $c\leq 6$, and the minimum energy is selected as the stabilizer ground state energy.
All possible types of distinct stabilizer ground states are listed in Table~\ref{table:phase}, and the corresponding phase diagram is plotted in Figure~\ref{fig:phase_cluster}(b). 
When comparing Figure~\ref{fig:phase_cluster} (a) and (b), the stabilizer ground state phase diagram matches the numerical ground state phase diagram well except for the shape of the boundary between the cluster phase and the polarized phase. The boundary predicted by stabilizer ground states is a straight line, while the numerical boundary is slightly curved. 
These agreements indicate that stabilizer ground states are useful to qualitatively understand phase transitions in quantum many-body systems and provide a new perspective compared to conventional mean-field approaches. 
The stabilizer ground state at the tricritical point $J_{y}=1,h_{y}=0$ is observed to have two new degenerate stabilizer ground states besides the stabilizer ground states in other phases.
These two new stabilizer ground states have stabilizers $\{X_{3n-1}Z_{3n}X_{3n+1},Y_{3n-1}Y_{3n},Y_{3n}Y_{3n+1}\}$ and $\{X_{3n-1}Z_{3n}X_{3n+1},X_{3n}Z_{3n+1}X_{3n+2},Y_{3n}Y_{3n+1}\}$, respectively.
\begin{table}

\caption{Stabilizer ground states of the Hamiltonian in Eq.~\eqref{eq:alpha-H} in different phases. Note that the conditions of $(J_y,h_y)$ do not strictly contradict each other, which indicates degeneracies of stabilizer ground states in the overlap regions (borders between phases or the tricritical point).}
\label{table:phase}

\begin{tabular}{|c|c|c|}
\hline 
\hline 
Stabilizers & $(J_{y},h_{y})$ & Phase\tabularnewline
\hline 
$\{X_{n-1}Z_{n}X_{n+1}\}$ & $J_{y}+h_{y}\leq1$ & Cluster\tabularnewline
$\{-Y_{n}\}$ & $J_{y}+h_{y}\geq1$,$h_{y}>0$ & Polarized\tabularnewline
$\{Y_{n}Y_{n+1}\}$ & $J_{y}\geq1,h_{y}=0$ & Ferromagnetism\tabularnewline
$\{X_{3n-1}Z_{3n}X_{3n+1},Y_{3n-1}Y_{3n},Y_{3n}Y_{3n+1}\}$ & $J_{y}=1,h_{y}=0$ & Tricritical point \tabularnewline
$\{X_{3n-1}Z_{3n}X_{3n+1},X_{3n}Z_{3n+1}X_{3n+2},Y_{3n}Y_{3n+1}\}$ & $J_{y}=1,h_{y}=0$ & Tricritical point \tabularnewline
\hline 
\hline 
\end{tabular}
\end{table}

\subsection{Extended stabilizer ground states} \label{sec:extended}

Stabilizer ground states can be used as a starting point to develop advanced numerical methods or quantum state ansatz for classical simulation.
As an illustration, we introduce the extended stabilizer ground state and demonstrate its capability of characterizing phase transitions of a 2D generalized toric code model.
The relation between computational basis states, stabilizer states, product states, and extended stabilizer states are discussed in Sec. \ref{sec:applications} and Fig. \ref{fig:venn}.
We first introduce a quantum state ansatz expressed as applying single-qubit rotations on some stabilizer states, i.e.
\begin{equation}
|\psi\rangle = U(\{\boldsymbol{\theta}_j\}) |\psi_\text{stab}\rangle = \prod_j e^{i \boldsymbol{\theta}_j \cdot \boldsymbol{S}_j} |\psi_\text{stab}\rangle,
\end{equation}
where $\boldsymbol{S}_j$ is the vector spin operator on the $j$th qubit.
We then define the extended stabilizer ground state by the state $|\psi\rangle$ with the lowest energy among all possible combinations of $\{\boldsymbol{\theta}_j\}$ and $|\psi_\text{stab}\rangle$.
Instead of directly finding the value of $\{\boldsymbol{\theta}_j\}$ and $|\psi_\text{stab}\rangle$ that minimizes the energy, which requires expensive discrete optimizations, we can effectively transform the Hamiltonian by
\begin{equation} \label{eq:H_transform}
    H \rightarrow H^\prime(\{\boldsymbol{\theta}_j\}) = U^\dagger(\{\boldsymbol{\theta}_j\}) H U(\{\boldsymbol{\theta}_j\}).
\end{equation}
The stabilizer ground state of the Hamiltonian $H^\prime(\{\boldsymbol{\theta}_j\})$ is thus a function of $\{\boldsymbol{\theta}_j\}$.
Since local Hamiltonians after single-site rotations remain local with the same localities $k$, such an extended stabilizer ground state formalism increases the expressive power without significantly complicating the problem, especially when each Pauli operator only nontrivially acts on a limited number of sites.

\begin{figure}[hbt]
    \centering
    \includegraphics[width=0.95\linewidth]{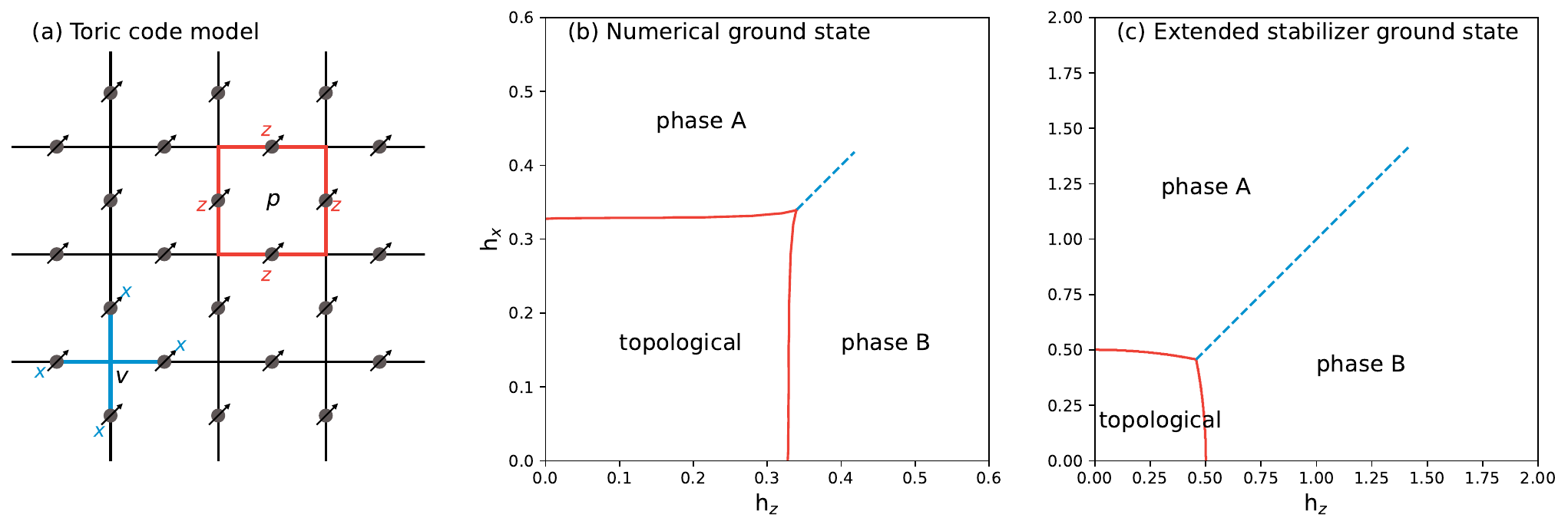}
    \caption{Geometry and phase diagrams of the 2D generalized toric code. (a) The geometry of the 2D generalized toric code model in Eq.~\eqref{eq:toric}. (b) Numerical ground state phase diagram obtained by continuous-time Monte Carlo calculations in Ref.~\cite{toric}. (c) Extended stabilizer ground state phase diagram. For both (b) and (c), three phases are found, including the topological phase, phase A, and phase B. The first-order transition line (dashed blue line) at $h_x=h_z$ begins at $h_x=h_z=0.34$ and ends at $h_x=h_z=0.418$ in (b), while it begins at $h_x=h_z=0.46$ and ends at $h_x=h_z=1.414$ in (c)}
    \label{fig:phase_toric}
\end{figure}

As a demonstration, we consider a 2D generalized toric code model with external magnetic fields. 
The Hamiltonian is
\begin{equation} \label{eq:toric}
H=-\Big(\sum_{v}A_{v}+\sum_{p}B_{p}\Big)-h_{x}\sum_{j}X_{j}-h_{z}\sum_{j}Z_{j},
\end{equation}
which is defined on a torus, where $A_{v}=\prod_{j\in v}X_{j}$ and $B_{p}=\prod_{j\in p}Z_{j}$ represent the product of spin operators on bonds incident to the vertex $v$ and surrounding plaquette $p$, respectively. 
The geometry of the vertices and plaquettes is shown in Figure~\ref{fig:phase_toric}(a). 
This Hamiltonian is studied by continuous-time Monte Carlo simulation in Ref.~\cite{toric} and the phase diagram is reproduced in Figure~\ref{fig:phase_toric}(b). 
At $h_x\rightarrow \infty$ with fixed $h_z$ or $h_z\rightarrow \infty$ with fixed $h_x$, each spin is polarized in the $x$ or $z$ direction, and gives the phase A or B, respectively.
A phase transition happens between the phase A and B at the first-order transition line $h_x=h_z$, which begins at $h_{x}=h_{z}=0.34$ and ends at $h_{x}=h_{z}=0.418$.
In the limit of $h_x\rightarrow \infty$ and $h_z \rightarrow \infty$, the polarization of the system varies continuously between phases A and B, thus no phase transition occurs.

Now we consider the extended stabilizer ground state of this Hamiltonian.
Since the Hamiltonian only contains X and Z, the single-qubit rotations can be restricted to the form of $U(\{\theta_j\}) = \prod_j e^{\frac{1}{2}i\theta_j Y_j}$.
As stated previously, we need to transform the Hamiltonian in Eq.~\eqref{eq:toric} by $U(\{\theta_j\})$ and then determine the stabilizer ground state.
As discussed in Sec.~\ref{sec:periodic}, the stabilizer ground state of a periodic local Hamiltonian should be periodic over supercells with some size $c$.
For simplification, the stabilizer ground state is assumed to have period 1. We set $\theta_j=\alpha$ and $\theta_j=\beta$ for sites $j$ on vertical bonds and horizontal bonds, respectively, and thus the total rotation operator can be written as $U(\alpha,\beta)$.

With fixed supercell size $c=1$, the stabilizer ground state of the rotated Hamiltonian $U(\alpha,\beta)^\dagger HU(\alpha,\beta)$ can be found via Eq.~\eqref{eq:periodic_stab_gs} for each set of rotation angles $\alpha,\beta$. 
The corresponding stabilizer ground state energy per site is written as $E(h_{x},h_{z},\alpha,\beta)$. 
The extended stabilizer ground state energy per site is then given by $E(h_{x},h_{z})=\min_{\alpha,\beta}E(h_{x},h_{z},\alpha,\beta)$.
In the following analysis, we apply the simplification process in Eq.~\eqref{eq:periodic_simplification} for convenience, which allows us to exclude Pauli terms like $P=X_{l}Z_{r}X_{u}Z_{d}$, where the subscripts $l,r,u,d$ stand for the left, right, up, and down site of either a vertex or a plaquette. 
The valid Pauli terms of the rotated Hamiltonians include (1) $X$ and $Z$ on each site; and (2) $X_{l}X_{r}X_{u}X_{d}$, $X_{l}X_{r}Z_{u}Z_{d}$, $Z_{l}Z_{r}X_{u}X_{d}$, and $Z_{l}Z_{r}Z_{u}Z_{d}$ on each vertex or plaquette.

The resulting extended stabilizer ground state phase diagram is plotted in Figure~\ref{fig:phase_toric}(c) and it matches the exact phase diagram qualitatively. 
In the topologically ordered phase, the stabilizers are the set of all $X_{l}X_{r}X_{u}X_{d}$, $X_{l}X_{r}Z_{u}Z_{d}$, $Z_{l}Z_{r}X_{u}X_{d}$, $Z_{l}Z_{r}Z_{u}Z_{d}$ on all vertices and plaquettes.
We find that the corresponding $E(h_{x},h_{z},\alpha,\beta)$ is a constant with respect to $\alpha$ and $\beta$, attributed to the fact that $U(\alpha,\beta)$ is a symmetry operation of this stabilizer state.
In phase A and B, the stabilizers are simply $X$ on each site and $X_lX_rX_uX_d$ on each vertex and plaquette, which corresponds to the product of single-site polarized states in the picture of the unrotated Hamiltonian.
We consider $h_x=h_z$, in which case the extended stabilizer ground state is always found at $\alpha=\beta$.
The corresponding per-site energy function $E(h,\alpha) = E(h_x=h,h_z=h,\alpha,\beta=\alpha)$ is given by 
\begin{equation}
    E(h,\alpha)=-\frac{1}{2}(\cos^{4}\alpha+\sin^{4}\alpha)-h(\cos\alpha+\sin\alpha),
\end{equation}
which is symmetric under $\alpha \rightarrow \frac{\pi}{2}-\alpha$.
No phase transition happens for large $h$ since only one minimum $\alpha=\frac{\pi}{4}$, while two minimums can be found for small $h$. 
The behavior could be better understood by Taylor expansion around $\alpha=\frac{\pi}{4}$, which gives
\begin{equation}
E(h,\alpha=\frac{\pi}{4}+\theta)=\text{const}+(\frac{h}{\sqrt{2}}-1)\theta^{2}+\frac{32-\sqrt{2}h}{24}\theta^{4}+O(\theta^{6}).
\end{equation}
This indicates that the first-order transition ends at $h=\sqrt{2}$, where a second-order phase transition happens due to the change of sign of the quadratic term.
Furthermore, at $h < h_c \approx 0.46$, $E(h,\alpha)$ is higher than the energy of the topologically ordered state for all $\alpha$.
Thus, we claim that the corresponding first-order transition line begins at $h_x=h_z\approx 0.46$ and ends at $h_x=h_z=\sqrt{2}\approx 1.414$.
Although the exact transition values are different, the qualitative picture captured by the extended stabilizer ground state is consistent with the ground truth.

\subsection{Initial state for VQE problems} \label{sec:vqe}

\begin{figure}[thb]
    \centering
    \includegraphics[width=0.8\textwidth]{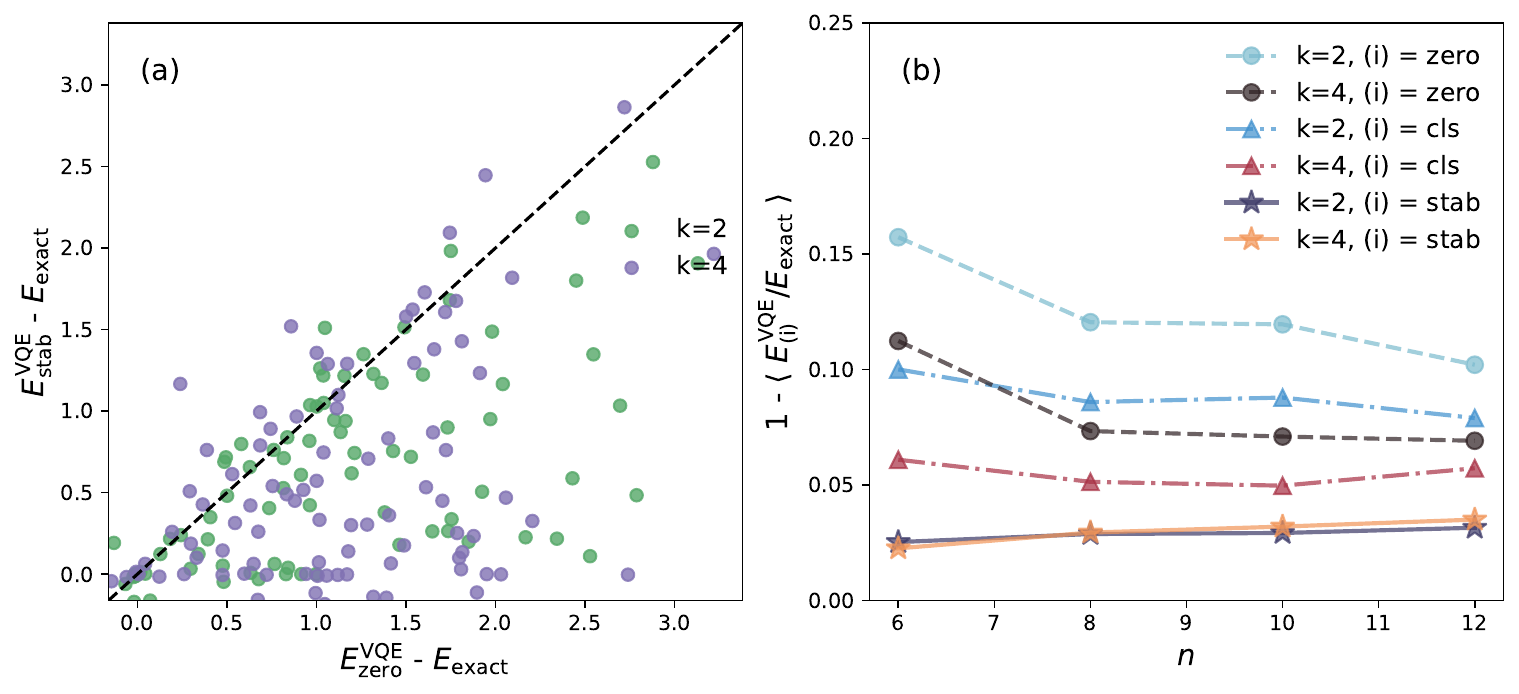}
    \caption{Errors of optimized VQE energies computed using the stochastic $k$-nearest Heisenberg model. $E_\text{exact}$ represents the exact ground state energy, and $E_\text{(i)}^\text{VQE}$ is the optimized VQE energy, where $(i)$ is the employed initialization scheme. Three initialization schemes, i.e. zero state, classical ground state, and stabilizer ground state, are compared. They are abbreviated as $(i)=$zero, $(i)=$cls, $(i)=$stab, respectively. (a) $E_\text{stab}^\text{VQE}$ - $E_\text{exact}$ plots versus $E_\text{zero}^\text{VQE}$ - $E_\text{exact}$ for $N=100$ random Hamiltonians with $n=6$ and $k=2, 4$. (b) Mean relative errors of energies ($1-\langle E_\text{(i)}^\text{VQE}/E_\text{exact} \rangle$) versus the number of sites $n$ for the three initializations with $k=2, 4$.}
    \label{fig:vqe}
\end{figure}

Stabilizer states have been recently used as initial states \cite{CAFQA,ising_stab_gs} for VQE problems to mitigate the notorious barren plateau issue \cite{cerezo2021cost,mcclean2018barren, Zhang2023b_z}.
The stabilizer initial states can be prepared on quantum circuits by efficient decomposition to up to $O(n^2/\log n)$ single-qubit and double-qubit Clifford gates \cite{aaronson2004improved}, or even fewer gates for stabilizer ground states of 1D local Hamiltonians (see \thref{stab_gs_preparation}).
The effective VQE ansatz is
\begin{equation} \label{eq:vqe_approach1}
|\psi(\boldsymbol{\theta})\rangle=U(\boldsymbol{\theta})|\psi_{\text{stab}}\rangle=U(\boldsymbol{\theta})U_{C}|0\rangle^{\otimes n},
\end{equation}
where the stabilizer initial state $|\psi_{\text{stab}}\rangle$ is decomposed to 
$U_{C}|0\rangle^{\otimes n}$. 
Another approach is to employ the quantum state ansatz as follows:
\begin{equation} \label{eq:vqe_approach2}
|\psi(\boldsymbol{\theta})\rangle=U_{C}U(\boldsymbol{\theta})|0\rangle^{\otimes n}.
\end{equation}
The advantage of the latter approach is that one can equivalently transform the Hamiltonian by $H\rightarrow H^{\prime}=U_{C}^{\dagger}HU_{C}$ classically, and thus only the $U(\boldsymbol{\theta})$ part needs to be performed on the quantum circuit \cite{sun2024toward,Zhang2021d_z, shang2023schrodinger}.
However, its disadvantage is that it might break the locality of the Hamiltonian and cause additional overhead on the hardware that cannot support nonlocal operations \cite{ibm_eagle,sycamore}.
Therefore, we adopt the former strategy in Eq.~\eqref{eq:vqe_approach1} for the following benchmark.

The \{XX,YY,ZZ\} model in Eq.~\eqref{eq:stochastic_heisenberg} still serves as the example Hamiltonian, and the variational Hamiltonian ansatz \cite{wiersema2020exploring} is used as the example VQE circuit for ground state optimization.  
By rewriting the Hamiltonian as $H=\sum_P w_P P$ with $P=S^{x}_{i}S^{x}_{j},S^{y}_{i}S^{y}_{j},S^{z}_{i}S^{z}_{j}$, the corresponding quantum circuit ansatz is as follows:
\begin{equation}
    |\psi(\boldsymbol{\theta})\rangle = \prod_P e^{i\theta_P P} |\psi_\text{init}\rangle.
\end{equation}
We compare three choices of the initial state $|\psi_\text{init}\rangle$, including the $|0\rangle^{\otimes n}$ state (referred to as zero state), the classical ground state, and the stabilizer ground state obtained by the exact 1D local algorithm.
The relations between classical ground states, stabilizer ground states, and product states are discussed in Sec. \ref{sec:mapping}.
The quantum circuit simulations are conducted via the \textit{TensorCircuit} software \cite{Zhang2022_z}.
The optimization of parameters $\boldsymbol{\theta}$ is performed by the default L-BFGS-B \cite{scipy_LBFGS_B} optimizer in \textit{SciPy} \cite{2020SciPy-NMeth} with zero initial values.

Figure \ref{fig:vqe}(a) displays the distributions of optimized energy errors of the zero state and stabilizer ground state initialization strategies tested on 100 random Hamiltonians with $n=6$.
Stabilizer state initializations result in lower VQE errors compared with those via the zero state in 82\% and 92\% of the 100 tests for $k=2$ and $k=4$, respectively. 
There are few points in the region of $E_\text{zero}^\text{VQE}<E_\text{stab}^\text{VQE}$,
which is attributed to the fact that an initialization state with a lower energy does not guarantee a lower final energy after VQE optimizations.
Figure \ref{fig:vqe}(b) also shows the mean relative errors of energies $1 - \langle E^\text{VQE}_{\text{(i)}} / E_\text{exact}\rangle$ for increasing $n$ and $k=2,4$, where (i) represents each of the three initialization strategies.
Initializations via stabilizer ground states are observed to systematically provide better energy estimations for both $k=2$ and $k=4$.

\section{Conclusions and Outlook}\label{sec:conclusion}
In this work, we introduce stabilizer ground states as a versatile toolkit of both qualitative analysis of quantum systems and cornerstone of developing advanced quantum state anstazes on classical or quantum computers. For general Hamiltonians, we establish the equivalence between the stabilizer ground state and the closed maximally-commuting Pauli subset. For 1D local Hamiltonians, we additionally develop an exact and efficient algorithm to obtain the exact stabilizer ground state with linear scaling. Besides, we prove that the stabilizer ground state of 1D local Hamiltonians can be prepared on quantum circuits with a linear scaled circuit depth. Furthermore, both the equivalence formalism for general Hamiltonians and the linear-scaled algorithm for 1D local Hamiltonians can be extended to infinite periodic systems. We also compare stabilizer ground states, mean-field ground states, and MPS ground states in terms of the applications for constructing quantum and classical ansatzes. By benchmarking on example Hamiltonians, we verified the computational scaling of the exact 1D local algorithm and demonstrated the substantial performance gain over the traditional discrete optimization strategies. We also illustrate that stabilizer ground states are promising tools for various applications, including qualitative analysis of phase transitions, generating better heuristics for VQE problems, and developing more expressive classical ground state ansatzes.

Looking forward, future studies can fruitfully branch into three major directions. The first avenue is to develop algorithms for stabilizer ground states of other types of Hamiltonians, including Hamiltonians with other quasi-1D structures and local Hamiltonians in higher dimensions. For the latter, finding the exact stabilizer ground state is NP-hard, evidenced by the NP-hardness of one of its simplified cases, i.e., the ground state problem of 2D classical spin models with random magnetic fields \cite{Barahona1982a_z,zhang2019classification}. However, approximate or heuristic algorithms \cite{Zhang2020b_z,gu2023zero,munoz2023low} for stabilizer ground states may still be practically useful for higher-dimensional systems. The second avenue extends the concept of stabilizer ground states to other physically interesting properties, such as excited states, mixed states, and thermal state sampling. These extensions are plausible, as the automaton structure of the 1D algorithm shares similarities with an ensemble of quantum states. The third avenue involves the exploration of more downstream applications for stabilizer ground states, especially via combining with other well-established quantum state ansatzes, such as tensor network \cite{masot2024stabilizer}, perturbation theory \cite{gu2024doped,beguvsic2023simulating}, variational quantum Monte Carlo, or low-rank (or low-energy) stabilizer decomposition \cite{bravyi2019simulation}.


\section*{Acknowledgement}
The authors thank the insightful discussions from Prof. Garnet K.L. Chan and Dr. Tomislav Begušić, and appreciate the valuable suggestions and comments from the reviewers. SXZ acknowledges the support from a start-up grant at IOP-CAS (E4VK061).

\vspace{1cm}
\section{Appendix}

\subsection{Proof of \thref{sparse_gs} \label{appendix:proof_sparse}}

Let $\boldsymbol{S}=\langle\boldsymbol{Q}\rangle$, $\boldsymbol{Q}\in\mathscr{S}(\boldsymbol{P})$ be one of the stabililizer group such that $E_{\text{stab}}(H,\boldsymbol{S})=E_{\text{gs}}$, and $|\psi\rangle$ is any stabilizer state stabilized by $\boldsymbol{S}$. Let $\boldsymbol{S}_{\psi}=\langle\text{Stab}(|\psi\rangle)\cap\tilde{\boldsymbol{P}}\rangle$, obviously we have $\boldsymbol{S}\subseteq\boldsymbol{S}_{\psi}$. Let $\boldsymbol{S}_{\psi}=\langle\boldsymbol{S},P_{1},P_{2},...,P_{k}\rangle$ where each $P_{i}\in\tilde{\boldsymbol{P}}$, $1\leq i\leq k$. We consider the sequence $\boldsymbol{S}_{i}=\langle\boldsymbol{S},P_{1},P_{2},...,P_{i}\rangle$ with $i=0,1,...,k$. Clearly we have $\boldsymbol{S}_{i}\cap\tilde{\boldsymbol{P}}\in\mathscr{S}(\boldsymbol{P})$ for each $\boldsymbol{S}_{i}$.

First we prove that $E_{\text{stab}}(H,\boldsymbol{S}_{i})=E_{\text{gs}}$ for each $\boldsymbol{S}_{i}$. Obviously, it is true for $i=0$. Given it is true for $\boldsymbol{S}_{i}$, we then consider $\boldsymbol{S}_{i+1}=\langle\boldsymbol{S}_{i},P_{i+1}\rangle$. If $P_{i+1}\in\boldsymbol{S}_{i}$, then $\boldsymbol{S}_{i+1}=\boldsymbol{S}_{i}$ so $E_{\text{stab}}(H,\boldsymbol{S}_{i+1})=E_{\text{gs}}$. Otherwise we consider $E_{i}=E_{\text{stab}}(H,\boldsymbol{S}_{i})$, $E_{+}=E_{\text{stab}}(H,\langle\boldsymbol{S}_{i},P_{i+1}\rangle)$ and $E_{-}=E_{\text{stab}}(H,\langle\boldsymbol{S}_{i},-P_{i+1}\rangle)$. For each $P\in\boldsymbol{P}$, it falls into one of the following situations: (1) $P\in\pm\boldsymbol{S}_{i}$ so it contributes equally to $E_{+}$ and $E_{-}$, (2) $P\notin\pm\boldsymbol{S}_{i}$ but $P\in\pm\langle\boldsymbol{S}_{i},P_{i+1}\rangle$ so it contributes no energy to $E_{i}$ and opposite energy to $E_{+}$ and $E_{-}$, (3) $P\notin\pm\langle\boldsymbol{S}_{i},P_{i+1}\rangle$ so it contributes no energy to $E_{i}$, $E_{+}$ and $E_{-}$. Thus we have $2E_{i}=E_{+}+E_{-}$. However $E_{i}$ is already the minimum of $E_{\text{stab}}(H,\boldsymbol{S})$ for $\boldsymbol{S}\in\mathscr{S}(\boldsymbol{P})$, thus $E_{+}=E_{-}=E_{i}=E_{\text{gs}}$, i.e. $E_{\text{stab}}(H,\boldsymbol{S}_{i+1})=E_{\text{gs}}$. Now we can conclude that $E_{\text{stab}}(H,\boldsymbol{S}_{i})=E_{\text{gs}}$ for each $\boldsymbol{S}_{i}$, which implies $E_{\text{stab}}(H,\boldsymbol{S}_{\psi})=E_{\text{gs}}$.

Next we prove $|\psi\rangle$ is a (degenerate) stabilizer ground state. According to \thref{stab_energy}, we have $\langle\psi|H|\psi\rangle=E_{\text{stab}}(H,\boldsymbol{S}_{\psi})=E_{\text{gs}}$. If there exists stabilizer state $|\psi^{\prime}\rangle$ such that $\langle\psi^{\prime}|H|\psi^{\prime}\rangle<\langle\psi|H|\psi\rangle=E_{\text{gs}}$, we should have $E_{\text{stab}}(H,\langle\text{Stab}(|\psi^{\prime}\rangle)\cap\tilde{\boldsymbol{P}}\rangle)=\langle\psi^{\prime}|H|\psi^{\prime}\rangle<E_{\text{gs}}$, which conflicts with the definition of $E_{\text{gs}}$. Thus we conclude that $|\psi\rangle$ is a (degenerate) stabilizer ground state.

\subsection{Proof of Eq. \ref{eq:Markov_independence_s} from Eq. \ref{eq:Markov_independence} \label{appendix:proof_Markov_independence_s}}

We first introduce an important fact $\boldsymbol{s}\leftrightarrow\{A_{1},...,A_{n}\}$, where $\leftrightarrow$ means bijective mapping. This is because (1) $G(A_{m})=s_{m}$ gives the mapping from the latter one to the former one, and (2) $A_{m}(\boldsymbol{s})$ itself is the function that maps from the former one to the latter one. Based on this, we have
\begin{equation}
\begin{aligned}\{\boldsymbol{s}|A_{m}\} & \leftrightarrow\{A_{1},...,A_{n}|A_{m}\}\\
 & =\{A_{1},...,A_{m}|A_{m}\}\otimes\{A_{m+1},...,A_{n}|A_{m}\}\\
 & \rightarrow\{\boldsymbol{s}_{\leq m}|A_{m}\}\otimes\{\boldsymbol{s}_{>m}|A_{m}\}\\
 & \supseteq\{\boldsymbol{s}|A_{m}\},
\end{aligned}
\label{eq:4step}
\end{equation}
where in the first line used $\boldsymbol{s}\leftrightarrow\{A_{1},...,A_{n}\}$, the second line used Eq. \ref{eq:Markov_independence}, the third line used the mapping $G(A_{m})=s_{m}$, and $\rightarrow$ means the surjective mapping. Thus Eq. \ref{eq:4step} gives a four-step way to map $\{\boldsymbol{s}|A_{m}\}$ to itself. Therefore the mapping in each step must be bijective. Specifically, the mapping from the third line to the fourth line is bijective, which exactly gives Eq. \ref{eq:Markov_independence_s}, i.e.
\begin{equation}
\{\boldsymbol{s}|A_{m}\}=\{\boldsymbol{s}_{\leq m}|A_{m}\}\otimes\{\boldsymbol{s}_{>m}|A_{m}\}.
\end{equation}

\subsection{Derivation of the state machine \label{appendix:SM}}

We first introduce the following lemma related to the projection operations, which will be frequently used:
\begin{lem}
\thlabel{projectAB} If $\boldsymbol{B}\subseteq\mathcal{P}_{I}$, then $\mathbb{P}_{I}(\langle\boldsymbol{A},\boldsymbol{B}\rangle)=\langle\mathbb{P}_{I}(\langle\boldsymbol{A}\rangle),\boldsymbol{B}\rangle$.
\end{lem}
\begin{proof}
We only need to prove $\langle\boldsymbol{A},\boldsymbol{B}\rangle\cap\mathcal{P}_{I}=\langle\boldsymbol{A}\cap\mathcal{P}_{I},\boldsymbol{B}\rangle$. This can be seen by (1) $\langle\boldsymbol{A}\cap\mathcal{P}_{I},\boldsymbol{B}\rangle\subseteq\langle\boldsymbol{A},\boldsymbol{B}\rangle$, (2) $\langle\boldsymbol{A}\cap\mathcal{P}_{I},\boldsymbol{B}\rangle\subseteq\mathcal{P}_{I}$ since $\boldsymbol{A}\cap\mathcal{P}_{I}\subseteq\mathcal{P}_{I}$ and $\boldsymbol{B}\subseteq\mathcal{P}_{I}$, and (3) for any $P=\langle\boldsymbol{A},\boldsymbol{B}\rangle\cap\mathcal{P}_{I}$ we can write $P=AB$ where $A\in\boldsymbol{A}$ and $B\in\boldsymbol{B}$. Since $B\in\mathcal{P}_{I}$, $P\in\mathcal{P}_{I}$ we also have $A\in\mathcal{P}_{I}$. Thus $P=AB\in\langle\mathbb{P}_{I}(\langle\boldsymbol{A}\rangle),\boldsymbol{B}\rangle$.
\end{proof}
We have already proved that Eq. \ref{eq:Markov_independence_s} is a necessary condition of Eq. \ref{eq:Markov_independence}. In fact, if we additionally have
\begin{equation}
\begin{aligned}\{\boldsymbol{s}_{\leq m},A_{m}\} & \leftrightarrow\{A_{1},...,A_{m}\},\\
\{\boldsymbol{s}_{>m},A_{m}\} & \leftrightarrow\{A_{m},...,A_{n}\}
\end{aligned}
\label{eq:Am_restriction}
\end{equation}
they are equivalent. With the mapping $\boldsymbol{s}\rightarrow\boldsymbol{Q}$ given in Sec. \ref{sec:mapping}, Eq. \ref{eq:Am_restriction} becomes
\begin{equation}
\begin{aligned}\{\boldsymbol{Q}_{\leq m},A_{m}\} & \leftrightarrow\{A_{1},...,A_{m}\},\\
\{\boldsymbol{Q}_{>m},A_{m}\} & \leftrightarrow\{A_{m},...,A_{n}\},
\end{aligned}
\label{eq:Am_restriction_Q}
\end{equation}
and Eq. \ref{eq:Markov_independence_s} becomes
\begin{equation}
\{\boldsymbol{Q}|A_{m}\}=\{\boldsymbol{Q}_{\leq m}|A_{m}\}\otimes\{\boldsymbol{Q}_{>m}|A_{m}\},\label{eq:Markov_independence_Q2}
\end{equation}
i.e. $\boldsymbol{Q}_{\leq m}$ must be decoupled from $\boldsymbol{Q}_{>m}$ given $A_{m}$. The coupling between $\boldsymbol{Q}_{\leq m}$ and $\boldsymbol{Q}_{>m}$ is given in the following two lemmas:
\begin{lem}
\thlabel{Q_to_Qpast} If $\boldsymbol{Q}\in\mathscr{S}(\boldsymbol{P})$, then $\boldsymbol{Q}_{\leq m}\in\mathscr{S}(\boldsymbol{P}_{\leq m})$ for each $m$.
\end{lem}
\begin{proof}
We recall that $\mathscr{S}(\boldsymbol{P})=\{\boldsymbol{Q}|\boldsymbol{Q}=\langle\boldsymbol{Q}\rangle\cap\tilde{\boldsymbol{P}},-I\notin\langle\boldsymbol{Q}\rangle\}$. Then $-I\notin\langle\boldsymbol{Q}_{\leq m}\rangle$ is obvious because $\boldsymbol{Q}_{\leq m}\subseteq\boldsymbol{Q}$. Given that $\boldsymbol{Q}=\langle\boldsymbol{Q}\rangle\cap\tilde{\boldsymbol{P}}$, we have $\langle\boldsymbol{Q}_{\leq m}\rangle\cap\tilde{\boldsymbol{P}}_{\leq m}\subseteq\langle\boldsymbol{Q}\rangle\cap\tilde{\boldsymbol{P}}_{\leq m}=\langle\boldsymbol{Q}\rangle\cap\tilde{\boldsymbol{P}}\cap\tilde{\boldsymbol{P}}_{\leq m}=\boldsymbol{Q}\cap\tilde{\boldsymbol{P}}_{\leq m}=\boldsymbol{Q}_{\leq m}$. On the other hand, we must have $\boldsymbol{Q}_{\leq m}\subseteq\langle\boldsymbol{Q}_{\leq m}\rangle\cap\tilde{\boldsymbol{P}}_{\leq m}$. Thus $\langle\boldsymbol{Q}_{\leq m}\rangle\cap\tilde{\boldsymbol{P}}_{\leq m}=\boldsymbol{Q}_{\leq m}$, i.e. $\boldsymbol{Q}_{\leq m}\in\mathscr{S}(\boldsymbol{P}_{\leq m})$.
\end{proof}
For simplicify, given $\boldsymbol{Q}_{\leq m}\in\mathscr{S}(\boldsymbol{P}_{\leq m})$, we say $\boldsymbol{Q}_{>m}\subseteq\tilde{\boldsymbol{P}}_{>m}$ is valid if $\boldsymbol{Q}=\boldsymbol{Q}_{\leq m}\cup\boldsymbol{Q}_{>m}\in\mathscr{S}(\boldsymbol{P})$. The requirement of $\boldsymbol{Q}_{>m}$ can be rewritten as follows:
\begin{cor}
\thlabel{Qpast_to_Qfuture} Given $\boldsymbol{Q}_{\leq m}\in\mathscr{S}(\boldsymbol{P}_{\leq m})$ , $\boldsymbol{Q}_{>m}\subseteq\tilde{\boldsymbol{P}}_{>m}$ is valid if and only if
\begin{enumerate}
\item $-I\notin\langle\boldsymbol{Q}_{>m}\rangle$, i.e. $\langle\boldsymbol{Q}_{>m}\rangle$ is a stabilizer group
\item $[\boldsymbol{Q}_{>m},\boldsymbol{Q}_{\leq m}]=0$
\item $\langle\boldsymbol{Q}_{\leq m},\boldsymbol{Q}_{>m}\rangle\cap\tilde{\boldsymbol{P}}_{>m}=\boldsymbol{Q}_{>m}$
\item $\langle\boldsymbol{Q}_{\leq m},\boldsymbol{Q}_{>m}\rangle\cap\tilde{\boldsymbol{P}}_{\leq m}=\boldsymbol{Q}_{\leq m}$
\end{enumerate}
\end{cor}
Condition 1 is only dependent on $\boldsymbol{Q}_{>m}$ itself so it does not couple $\boldsymbol{Q}_{>m}$ with $\boldsymbol{Q}_{\leq m}$. Condition 2 requires that the added stabilizers $\boldsymbol{Q}_{>m}$ should commute with all previous stabilizers $\boldsymbol{Q}_{\leq m}$. Conditions 3 and 4 require that group multiplication operations between $\boldsymbol{Q}_{\leq m}$ and $\boldsymbol{Q}_{>m}$ does not generate new elements in $\tilde{\boldsymbol{P}}$, which can be further decomposed to $\tilde{\boldsymbol{P}}_{>m}$ and $\tilde{\boldsymbol{P}}_{\leq m}$. Thus conditions 2, 3, 4 give the coupling between $\boldsymbol{Q}_{>m}$ and $\boldsymbol{Q}_{\leq m}$. Now we want to construct $A_{m}$ to decouple them. With 
\begin{equation}
\tilde{\boldsymbol{P}}_{\text{invalid}}^{m}=\tilde{\boldsymbol{P}}_{\text{invalid}}^{m}(\boldsymbol{Q}_{\leq m})=\{P\in\tilde{\boldsymbol{P}}_{>m}|[P,\boldsymbol{Q}_{\leq m}]\neq0\}
\end{equation}
(i.e. Eq. \ref{eq:Pinvalid}), condition 2 can be rewritten as 
\begin{equation}
\begin{aligned} & [\boldsymbol{Q}_{\leq m},\boldsymbol{Q}_{>m}]=0\\
\Leftrightarrow & \boldsymbol{Q}_{>m}\cap\{P\in\tilde{\boldsymbol{P}}_{>m}|[P,\boldsymbol{Q}_{\leq m}]\neq0\}=\emptyset\\
\Leftrightarrow & \boldsymbol{Q}_{>m}\cap\tilde{\boldsymbol{P}}_{\text{invalid}}^{m}=\emptyset
\end{aligned}
\label{eq:Pinvalid_to_future}
\end{equation}
which is only dependent on $\boldsymbol{Q}_{>m}$. With
\begin{equation}
\boldsymbol{S}_{\text{proj}}^{m}=\boldsymbol{S}_{\text{proj}}^{m}(\boldsymbol{Q}_{\leq m})=\mathbb{P}_{>m-k}(\langle\boldsymbol{Q}_{\leq m}\rangle)
\end{equation}
(i.e. Eq. \ref{eq:Sproj}), condition 3 can be rewritten as
\begin{equation}
\begin{aligned}\boldsymbol{Q}_{>m} & =\langle\boldsymbol{Q}_{\leq m},\boldsymbol{Q}_{>m}\rangle\cap\tilde{\boldsymbol{P}}_{>m}\\
 & =\mathbb{P}_{>m-k}(\langle\boldsymbol{Q}_{\leq m},\boldsymbol{Q}_{>m}\rangle)\cap\tilde{\boldsymbol{P}}_{>m}\\
 & =\langle\mathbb{P}_{>m-k}(\langle\boldsymbol{Q}_{\leq m}\rangle),\boldsymbol{Q}_{>m}\rangle\cap\tilde{\boldsymbol{P}}_{>m}\\
 & =\langle\boldsymbol{S}_{\text{proj}}^{m},\boldsymbol{Q}_{>m}\rangle\cap\tilde{\boldsymbol{P}}_{>m},
\end{aligned}
\label{eq:Sproj_to_future}
\end{equation}
which is only dependent on $\boldsymbol{Q}_{>m}$. Here we used $\boldsymbol{Q}_{>m}\subseteq\mathcal{P}_{>m-k}$ and \thref{projectAB} in the third line. With
\begin{equation}
\boldsymbol{S}_{\text{right}}^{m}=\boldsymbol{S}_{\text{right}}^{m}(\boldsymbol{Q}_{\geq m})=\mathbb{P}_{\leq m}(\langle\boldsymbol{Q}_{\geq m}\rangle)
\end{equation}
(i.e. Eq. \ref{eq:Sright}), condition 4 can be rewritten as
\begin{equation}
\begin{aligned}\boldsymbol{Q}_{\leq m} & =\langle\boldsymbol{Q}_{\leq m},\boldsymbol{Q}_{>m}\rangle\cap\tilde{\boldsymbol{P}}_{\leq m}\\
 & =\langle\boldsymbol{Q}_{\leq m},\boldsymbol{Q}_{\geq m}\rangle\cap\tilde{\boldsymbol{P}}_{\leq m}\\
 & =\mathbb{P}_{\leq m}(\langle\boldsymbol{Q}_{\leq m},\boldsymbol{Q}_{\geq m}\rangle)\cap\tilde{\boldsymbol{P}}_{\leq m}\\
 & =\langle\boldsymbol{Q}_{\leq m},\mathbb{P}_{\leq m}(\langle\boldsymbol{Q}_{\geq m}\rangle\rangle\cap\tilde{\boldsymbol{P}}_{\leq m}\\
 & =\langle\boldsymbol{Q}_{\leq m},\boldsymbol{S}_{\text{right}}^{m}\rangle\cap\tilde{\boldsymbol{P}}_{\leq m},
\end{aligned}
\label{eq:Sright_to_past_der}
\end{equation}
which is only dependent on $\boldsymbol{Q}_{\leq m}$. Here we used $\boldsymbol{Q}_{\leq m}\subseteq\mathcal{P}_{\leq m}$ and \thref{projectAB} in the fourth line. As a summary, given 
\begin{equation}
A_{m}(\boldsymbol{Q})=(\boldsymbol{S}_{\text{proj}}^{m}(\boldsymbol{Q}_{\leq m}),\tilde{\boldsymbol{P}}_{\text{invalid}}^{m}(\boldsymbol{Q}_{\leq m}),\boldsymbol{S}_{\text{right}}^{m}(\boldsymbol{Q}_{\geq m})),
\end{equation}
the four conditions are now
\begin{equation}
\begin{aligned}-I & \notin\langle\boldsymbol{Q}_{>m}\rangle\\
\boldsymbol{Q}_{>m}\cap\tilde{\boldsymbol{P}}_{\text{invalid}}^{m} & =\emptyset\\
\langle\boldsymbol{S}_{\text{proj}}^{m},\boldsymbol{Q}_{>m}\rangle\cap\tilde{\boldsymbol{P}}_{>m} & =\boldsymbol{Q}_{>m}.
\end{aligned}
\label{eq:Qfuture_all}
\end{equation}
for conditions 1, 2, 3, which depend on $\boldsymbol{Q}_{>m}$, and
\begin{equation}
\langle\boldsymbol{Q}_{\leq m},\boldsymbol{S}_{\text{right}}^{m}\rangle\cap\tilde{\boldsymbol{P}}_{\leq m}=\boldsymbol{Q}_{\leq m}.\label{eq:Sright_to_past}
\end{equation}
for condition 4, which depends on $\boldsymbol{Q}_{\leq m}$. Recalling that \thref{Qpast_to_Qfuture} additionally requires $\boldsymbol{Q}_{\leq m}\in\mathscr{S}(\boldsymbol{P}_{\leq m})$, we conclude that Eq. \ref{eq:Markov_independence_Q2} holds with
\begin{equation}
\{\boldsymbol{Q}_{\leq m}|A_{m}\}=\{\boldsymbol{Q}_{\leq m}\in\mathscr{S}(\boldsymbol{P}_{\leq m})\text{ satisfying Eq. \ref{eq:Sright_to_past}}\}\label{eq:expr_Am_to_Qpast}
\end{equation}
and
\begin{equation}
\{\boldsymbol{Q}_{>m}|A_{m}\}=\{\boldsymbol{Q}_{>m}\subseteq\tilde{\boldsymbol{P}}_{>m}\text{ satisfying Eq. \ref{eq:Qfuture_all}}\}.\label{eq:expr_Am_to_Qfuture}
\end{equation}

Now we show that $A_{m}$ is a state machine. The existence of the state function $G$ is already given in Eq. \ref{eq:G}. Therefore we just need to prove Eq. \ref{eq:Am_restriction_Q}. The ``$\leftarrow"$ part is trivial due to the state function $G$. For the ``$\rightarrow"$ part, $\{\boldsymbol{Q}_{\leq m},A_{m}\}\rightarrow\{A_{1},...,A_{m}\}$ can be obtained by $A_{i}=(\boldsymbol{S}_{\text{proj}}^{i}(\boldsymbol{Q}_{\leq i}),\tilde{\boldsymbol{P}}_{\text{invalid}}^{i}(\boldsymbol{Q}_{\leq i}),\boldsymbol{S}_{\text{right}}^{i}=\mathbb{P}_{\leq i}(\boldsymbol{Q}_{\geq i}\rangle)=\mathbb{P}_{\leq i}(\boldsymbol{S}_{\text{right}}^{m}))$ for $i\leq m$. To prove $\{\boldsymbol{Q}_{>m},A_{m}\}\rightarrow\{A_{m},...,A_{n}\}$, we show that $A_{m+1}$ can be obtained by $A_{m}$ and $\boldsymbol{S}_{\text{right}}^{m+1}$, i.e
\begin{equation}
A_{m+1}=A_{m+1}(A_{m},\boldsymbol{S}_{\text{right}}^{m+1}).
\end{equation}
Once it is true, $\{A_{m},...,A_{n}\}$ can be obtained by applying $A_{i+1}=A_{i+1}(A_{i},\boldsymbol{S}_{\text{right}}^{i+1}(\boldsymbol{Q}_{\geq i+1}))$ for $i=m,...,n-1$. To construct $A_{m+1}(A_{i},\boldsymbol{S}_{\text{right}}^{m+1})$, we only need to derive $\boldsymbol{S}_{\text{proj}}^{m+1}$ and $\tilde{\boldsymbol{P}}_{\text{invalid}}^{m+1}$. Specifically, we have:

\begin{equation}
\begin{aligned}\boldsymbol{S}_{\text{proj}}^{m+1}(\boldsymbol{Q}_{\leq m+1}) & =\mathbb{P}_{>m-k+1}(\langle\boldsymbol{Q}_{\leq m},\boldsymbol{Q}_{m+1}\rangle)\\
 & =\langle\mathbb{P}_{>m-k+1}(\langle\boldsymbol{Q}_{\leq m}\rangle),\boldsymbol{Q}_{m+1}\rangle\\
 & =\langle\mathbb{P}_{>m-k+1}(\boldsymbol{S}_{\text{proj}}^{m}(\boldsymbol{Q}_{\leq m})),\boldsymbol{Q}_{m+1}\rangle\\
 & :=\boldsymbol{S}_{\text{proj}}^{m+1}(\boldsymbol{S}_{\text{proj}}^{m},\boldsymbol{Q}_{m+1}),
\end{aligned}
\label{eq:Sproj_next_der}
\end{equation}
where the second line used $\boldsymbol{Q}_{m+1}\subseteq\mathcal{P}_{>m-k+1}$ and \thref{projectAB}, and

\begin{equation}
\begin{aligned}\tilde{\boldsymbol{P}}_{\text{invalid}}^{m+1}(\boldsymbol{Q}_{\leq m+1}) & =\{P\in\tilde{\boldsymbol{P}}_{>m+1}|[P,\boldsymbol{Q}_{\geq m+1}]\neq0\}\\
 & =\{P\in\tilde{\boldsymbol{P}}_{\geq m+1}|P\in\tilde{\boldsymbol{P}}_{\text{invalid}}^{m}(\boldsymbol{Q}_{\leq m})\text{ or }[P,\boldsymbol{Q}_{m+1}]\neq0\}\\
 & :=\tilde{\boldsymbol{P}}_{\text{invalid}}^{m+1}(\tilde{\boldsymbol{P}}_{\text{invalid}}^{m},\boldsymbol{Q}_{m+1}),
\end{aligned}
\label{eq:Pinvalid_next_der}
\end{equation}
which are the same with Eq. \ref{eq:Sproj_next} and Eq. \ref{eq:Pinvalid_next}. Finally, $\boldsymbol{Q}_{m+1}$ is given by $\boldsymbol{S}_{\text{right}}^{m+1}\cap\tilde{\boldsymbol{P}}_{m+1}$ according to Eq. \ref{eq:G}, thus we conclude that $A_{m+1}$ can indeed be rewritten as a function of $A_{m}$ and $\boldsymbol{S}_{\text{right}}^{m+1}$.

\subsection{Derivation of Eq. \ref{eq:Sright_next_conditions} \label{appendix:proof_F}}

\thref{relaxed_F} requires that, for any path $\{A_{1},...,A_{m+1}|A_{i+1}\in\tilde{F}(A_{i})\}$, if $A_{m+1}\in\mathcal{A}_{m+1}$, then $A_{m}\in\mathcal{A}_{m}$. We require a stronger version: if $A_{m+1}=A_{m+1}(\boldsymbol{Q})$ for some \textbf{$\boldsymbol{Q}\in\mathscr{S}(\boldsymbol{P})$}, then $A_{m}=A_{m}(\boldsymbol{Q}^{\prime})$, where $\boldsymbol{Q}^{\prime}=\boldsymbol{Q}_{>m+1}\cup\boldsymbol{Q}_{\leq m}^{\prime}$, $\boldsymbol{Q}_{\leq m}^{\prime}=\cup_{i=1}^{m}G(A_{m})$. Clearly it implies $A_{m}\in\mathcal{A}_{m}$ given $A_{m+1}\in\mathcal{A}_{m+1}$. Additionally, we require that $\boldsymbol{Q}_{\leq m}^{\prime}\in\{\boldsymbol{Q}_{\leq m}|A_{m}\}$. As shown in Eq. \ref{eq:expr_Am_to_Qpast}, it is equivalent to

\begin{equation}
\begin{aligned}\boldsymbol{Q}_{\leq m}^{\prime} & \in\mathscr{S}(\boldsymbol{P}_{\leq m}),\\
\langle\boldsymbol{Q}_{\leq m}^{\prime},\boldsymbol{S}_{\text{right}}^{m}\rangle\cap\tilde{\boldsymbol{P}}_{\leq m} & =\boldsymbol{Q}_{\leq m}^{\prime}.
\end{aligned}
\label{eq:relaxed_Am_requirement}
\end{equation}

Now we construct $\tilde{F}(A_{m})$ to satisfy the above two requirements. Following the logic of Sec. \ref{sec:F}, we essentially need to construct $\tilde{\mathcal{S}}_{\text{right}}^{m+1}(A_{m})$, and $\tilde{F}(A_{m})$ is then given in Eq. \ref{eq:F_Am_relaxed}. The latter requirement $\boldsymbol{Q}_{\leq m}^{\prime}\in\{\boldsymbol{Q}_{\leq m}|A_{m}\}$ is equivalent to say that, if Eq. \ref{eq:relaxed_Am_requirement} holds for $m$, it also holds for $m+1$. Given $\boldsymbol{Q}_{\leq m}^{\prime}=\cup_{i=1}^{m}G(A_{i})\in\mathscr{S}(\boldsymbol{P}_{\leq m})$, we want to construct $\boldsymbol{Q}_{m+1}^{\prime}=\boldsymbol{S}_{\text{right}}^{m+1}\cap\tilde{\boldsymbol{P}}_{m+1}$ such that $\boldsymbol{Q}_{\leq m+1}^{\prime}=\boldsymbol{Q}_{\leq m}^{\prime}\cup\boldsymbol{Q}_{m+1}^{\prime}\in\mathscr{S}(\boldsymbol{P}_{\leq m+1})$. We can derive it by replacing all subscripts $>m$ to $m+1$ in \thref{Qpast_to_Qfuture} and the following derivations. With such replacements, we require $\boldsymbol{Q}_{m+1}^{\prime}$ to satisfy
\begin{equation}
\begin{aligned}-I & \notin\langle\boldsymbol{Q}_{m+1}^{\prime}\rangle\\
\boldsymbol{Q}_{m+1}^{\prime}\cap\tilde{\boldsymbol{P}}_{\text{invalid}}^{m} & =\emptyset\\
\langle\boldsymbol{S}_{\text{proj}}^{m},\boldsymbol{Q}_{m+1}^{\prime}\rangle\cap\tilde{\boldsymbol{P}}_{m+1} & =\boldsymbol{Q}_{m+1}^{\prime}
\end{aligned}
\label{eq:Qnext_all}
\end{equation}
for the first three conditions, and
\begin{equation}
\langle\boldsymbol{Q}_{\leq m}^{\prime},\boldsymbol{Q}_{m+1}^{\prime}\rangle\cap\tilde{\boldsymbol{P}}_{\leq m}=\boldsymbol{Q}_{\leq m}^{\prime}\label{eq:Qnext_last}
\end{equation}
for the last condition. Similarly, given $\langle\boldsymbol{Q}_{\leq m}^{\prime},\boldsymbol{S}_{\text{right}}^{m}\rangle\cap\tilde{\boldsymbol{P}}_{\leq m}=\boldsymbol{Q}_{\leq m}^{\prime}$, we require $\langle\boldsymbol{Q}_{\leq m+1}^{\prime},\boldsymbol{S}_{\text{right}}^{m+1}\rangle\cap\tilde{\boldsymbol{P}}_{\leq m+1}=\boldsymbol{Q}_{\leq m+1}^{\prime}$, which can be decomposed to the $\tilde{\boldsymbol{P}}_{\leq m}$ part and $\tilde{\boldsymbol{P}}_{\leq m}$ part as:
\begin{equation}
\begin{aligned}\boldsymbol{Q}_{m+1}^{\prime} & =\langle\boldsymbol{Q}_{\leq m+1}^{\prime},\boldsymbol{S}_{\text{right}}^{m+1}\rangle\cap\tilde{\boldsymbol{P}}_{m+1}\\
 & =\langle\mathbb{P}_{\leq m+1}(\boldsymbol{Q}_{\leq m+1}^{\prime}),\boldsymbol{S}_{\text{right}}^{m+1}\rangle\cap\tilde{\boldsymbol{P}}_{m+1}\\
 & =\langle\boldsymbol{S}_{\text{proj}}^{m+1}(\boldsymbol{Q}_{\leq m+1}^{\prime}),\boldsymbol{S}_{\text{right}}^{m+1}\rangle\cap\tilde{\boldsymbol{P}}_{m+1}\\
 & =\langle\boldsymbol{S}_{\text{proj}}^{m+1}(\boldsymbol{S}_{\text{proj}}^{m},\boldsymbol{Q}_{m+1}^{\prime}),\boldsymbol{S}_{\text{right}}^{m+1}\rangle\cap\tilde{\boldsymbol{P}}_{m+1}
\end{aligned}
\label{eq:F_for_Qnext}
\end{equation}
where $\boldsymbol{S}_{\text{proj}}^{m+1}(\boldsymbol{S}_{\text{proj}}^{m},\boldsymbol{Q}_{m+1}^{\prime})$ is given in Eq. \ref{eq:Sproj_to_future}, and 
\begin{equation}
\begin{aligned}\boldsymbol{Q}_{\leq m}^{\prime} & =\langle\boldsymbol{Q}_{\leq m+1}^{\prime},\boldsymbol{S}_{\text{right}}^{m+1}\rangle\cap\tilde{\boldsymbol{P}}_{\leq m}\\
 & =\langle\boldsymbol{Q}_{\leq m}^{\prime},\boldsymbol{S}_{\text{right}}^{m+1}\rangle\cap\tilde{\boldsymbol{P}}_{\leq m}\\
 & =\langle\boldsymbol{Q}_{\leq m}^{\prime},\langle\mathbb{P}_{\leq m}(\boldsymbol{S}_{\text{right}}^{m+1}),\boldsymbol{Q}_{m}^{\prime}\rangle\rangle\cap\tilde{\boldsymbol{P}}_{\leq m},
\end{aligned}
\label{eq:F_for_Qpast}
\end{equation}
where in the second line we used $\boldsymbol{Q}_{m+1}^{\prime}=\boldsymbol{S}_{\text{right}}^{m+1}\cap\tilde{\boldsymbol{P}}_{m+1}\subseteq\boldsymbol{S}_{\text{right}}^{m+1}$. By comparing with $\langle\boldsymbol{Q}_{\leq m}^{\prime},\boldsymbol{S}_{\text{right}}^{m}\rangle\cap\tilde{\boldsymbol{P}}_{\leq m}=\boldsymbol{Q}_{\leq m}^{\prime}$, Eq. \ref{eq:F_for_Qpast} is satisfied if
\begin{equation}
\langle\mathbb{P}_{\leq m}(\boldsymbol{S}_{\text{right}}^{m+1}),\boldsymbol{Q}_{m}^{\prime}\rangle=\boldsymbol{S}_{\text{right}}^{m}.\label{eq:connect_Sright}
\end{equation}

Now we consider the former requirement $A_{m}=A_{m}(\boldsymbol{Q}^{\prime})$, which is equivalent to $\boldsymbol{S}_{\text{proj}}^{m}=\boldsymbol{S}_{\text{proj}}^{m}(\boldsymbol{Q}_{\leq m}^{\prime})$, $\tilde{\boldsymbol{P}}_{\text{invalid}}^{m}=\tilde{\boldsymbol{P}}_{\text{invalid}}^{m}(\boldsymbol{Q}_{\le m}^{\prime})$, and $\boldsymbol{S}_{\text{right}}^{m}=\boldsymbol{S}_{\text{right}}^{m}(\boldsymbol{Q}_{\geq m}^{\prime}=\boldsymbol{Q}_{>m+1}\cup G(A_{m+1}))$ given $A_{m+1}=A_{m+1}(\boldsymbol{Q})$, where $\boldsymbol{Q}^{\prime}=\boldsymbol{Q}_{>m+1}\cup\boldsymbol{Q}_{\leq m}^{\prime}$, $\boldsymbol{Q}_{\leq m}^{\prime}=\cup_{i=1}^{m}G(A_{m})$. The first two are trivial since $\boldsymbol{S}_{\text{proj}}^{m}$ and $\tilde{\boldsymbol{P}}_{\text{invalid}}^{m}$ are recursively constructed by $\boldsymbol{S}_{\text{proj}}^{i+1}=\boldsymbol{S}_{\text{proj}}^{i+1}(\boldsymbol{S}_{\text{proj}}^{i},\boldsymbol{Q}_{i+1}^{\prime})$ and $\tilde{\boldsymbol{P}}_{\text{invalid}}^{i+1}(\tilde{\boldsymbol{P}}_{\text{invalid}}^{i},\boldsymbol{Q}_{i+1}^{\prime})$. Combined with Eq. \ref{eq:Sproj_next_der} and Eq. \ref{eq:Pinvalid_next_der}, we can recursive derive $\boldsymbol{S}_{\text{proj}}^{i+1}=\boldsymbol{S}_{\text{proj}}^{i+1}(\boldsymbol{Q}_{\leq i+1}^{\prime})$ from $\boldsymbol{S}_{\text{proj}}^{i}=\boldsymbol{S}_{\text{proj}}^{i+1}(\boldsymbol{Q}_{\leq i}^{\prime})$, $\tilde{\boldsymbol{P}}_{\text{invalid}}^{i+1}=\tilde{\boldsymbol{P}}_{\text{invalid}}^{i+1}(\boldsymbol{Q}_{\le i+1}^{\prime})$ from $\tilde{\boldsymbol{P}}_{\text{invalid}}^{i}=\tilde{\boldsymbol{P}}_{\text{invalid}}^{i}(\boldsymbol{Q}_{\le i}^{\prime})$. We show that $\boldsymbol{S}_{\text{right}}^{m}=\boldsymbol{S}_{\text{right}}^{m}(\boldsymbol{Q}_{\geq m}^{\prime}=\boldsymbol{Q}_{>m+1}\cup G(A_{m+1}))$ can be derived from Eq. \ref{eq:connect_Sright} as
\begin{equation}
\begin{aligned}\boldsymbol{S}_{\text{right}}^{m}(\boldsymbol{Q}_{\geq m}^{\prime}) & =\mathbb{P}_{\leq m}(\langle\boldsymbol{Q}_{\geq m}^{\prime}\rangle)\\
 & =\langle\mathbb{P}_{\leq m}(\langle\boldsymbol{Q}_{\geq m+1}^{\prime}\rangle),\boldsymbol{Q}_{m}^{\prime}\rangle\\
 & =\langle\mathbb{P}_{\leq m}(\mathbb{P}_{\leq m+1}(\langle\boldsymbol{Q}_{\geq m+1}^{\prime}\rangle)),\boldsymbol{Q}_{m}^{\prime}\rangle\\
 & =\langle\mathbb{P}_{\leq m}(\boldsymbol{S}_{\text{right}}^{m+1}(\boldsymbol{Q}_{\geq m+1}^{\prime})),\boldsymbol{Q}_{m}^{\prime}\rangle\\
 & =\langle\mathbb{P}_{\leq m}(\boldsymbol{S}_{\text{right}}^{m+1}),\boldsymbol{Q}_{m}^{\prime}\rangle\\
 & =\boldsymbol{S}_{\text{right}}^{m},
\end{aligned}
\end{equation}
where in the last line we used Eq. \ref{eq:connect_Sright}.

Until now, we have shown that $A_{m+1}\in\mathcal{A}_{m+1}$ implies $A_{m}\in\mathcal{A}_{m}$ if Eq. \ref{eq:Qnext_all}, Eq. \ref{eq:Qnext_last}, Eq. \ref{eq:F_for_Qnext}, and Eq. \ref{eq:connect_Sright} are satisfied. We now simplify these conditions as follows. First we show that Eq. \ref{eq:Qnext_last} can be derived from Eq. \ref{eq:connect_Sright} and $\langle\boldsymbol{Q}_{\leq m}^{\prime},\boldsymbol{S}_{\text{right}}^{m}\rangle\cap\tilde{\boldsymbol{P}}_{\leq m}=\boldsymbol{Q}_{\leq m}^{\prime}$ by observing
\begin{equation}
\begin{aligned}\boldsymbol{Q}_{\leq m}^{\prime} & =\langle\boldsymbol{Q}_{\leq m}^{\prime}\rangle\cap\tilde{\boldsymbol{P}}_{\leq m}\\
 & \subseteq\langle\boldsymbol{Q}_{\leq m}^{\prime},\boldsymbol{Q}_{m+1}^{\prime}\rangle\cap\tilde{\boldsymbol{P}}_{\leq m}\\
 & \subseteq\langle\boldsymbol{Q}_{\leq m}^{\prime},\boldsymbol{S}_{\text{right}}^{m+1}\rangle\cap\tilde{\boldsymbol{P}}_{\leq m}\\
 & =\langle\boldsymbol{Q}_{\leq m}^{\prime},\mathbb{P}_{\leq m}(\boldsymbol{S}_{\text{right}}^{m+1})\rangle\cap\tilde{\boldsymbol{P}}_{\leq m}\\
 & =\langle\boldsymbol{Q}_{\leq m}^{\prime},\boldsymbol{S}_{\text{right}}^{m}\rangle\cap\tilde{\boldsymbol{P}}_{\leq m}\\
 & =\boldsymbol{Q}_{\leq m}^{\prime},
\end{aligned}
\end{equation}
which implies that all the $\subseteq$ relations must be $=$. Similarly, the third line in Eq. \ref{eq:Qnext_all} can be derived from Eq. \ref{eq:F_for_Qnext} by observing
\begin{equation}
\begin{aligned}\boldsymbol{Q}_{m+1}^{\prime} & =\langle\boldsymbol{Q}_{m+1}^{\prime}\rangle\cap\tilde{\boldsymbol{P}}_{m+1}\\
 & \subseteq\langle\boldsymbol{S}_{\text{proj}}^{m},\boldsymbol{Q}_{m+1}^{\prime}\rangle\cap\tilde{\boldsymbol{P}}_{m+1}\\
 & \subseteq\langle\boldsymbol{S}_{\text{proj}}^{m+1},\boldsymbol{S}_{\text{right}}^{m+1}\rangle\cap\tilde{\boldsymbol{P}}_{m+1}\\
 & =\boldsymbol{Q}_{m+1}^{\prime}.
\end{aligned}
\end{equation}
Thus it simplies to the four conditions in Eq. \ref{eq:Sright_next_conditions}, i.e.
\begin{equation}
\begin{aligned}-I & \notin\langle\boldsymbol{Q}_{m+1}^{\prime}\rangle,\\
\boldsymbol{Q}_{m+1}^{\prime}\cap\tilde{\boldsymbol{P}}_{\text{invalid}}^{m} & =\emptyset,\\
\langle\mathbb{P}_{\leq m}(\boldsymbol{S}_{\text{right}}^{m+1}),\boldsymbol{Q}_{m}^{\prime}\rangle & =\boldsymbol{S}_{\text{right}}^{m},\\
\boldsymbol{S}_{\text{proj}}^{m+1}\cap\boldsymbol{S}_{\text{right}}^{m+1} & =\boldsymbol{Q}_{m+1}^{\prime},
\end{aligned}
\end{equation}
where $\boldsymbol{S}_{\text{proj}}^{m+1}=\boldsymbol{S}_{\text{proj}}^{m+1}(\boldsymbol{S}_{\text{proj}}^{m},\boldsymbol{Q}_{m+1}^{\prime}=\boldsymbol{S}_{\text{right}}^{m+1}\cap\tilde{\boldsymbol{P}}_{m+1})$ is a function of $\boldsymbol{S}_{\text{right}}^{m+1}$. This expression is the same with Eq. \ref{eq:Sright_next_conditions}.

Besides, one can easily verify that Eq. \ref{eq:Sright_next_conditions} is satisfied for $A_{m}(\boldsymbol{Q})$ and $A_{m+1}(\boldsymbol{Q})$ if $\boldsymbol{Q}\in\mathscr{S}(\boldsymbol{P})$. This means that $\tilde{F}$ defined by Eq. \ref{eq:F_Am_relaxed} and Eq. \ref{eq:relaxed_Sright_next} has the same behavior with $F$ in the valid region $\mathcal{A}_{m}$, which is the first requirement of \thref{relaxed_F}. Finally, we can conclude that such $\tilde{F}$ is a relaxed transition function according to \thref{relaxed_F}.

\subsection{Proof of \thref{local_nS} \label{appendix:proof_local_nS} }
\begin{proof}
Following Sec. \ref{sec:F}, and especially Eq. \ref{eq:relaxed_Am_requirement}, for each (partial) path $\{A_{i=1}^{m}|A_{i+1}\in\tilde{F}(A_{i})\}$, let $\boldsymbol{Q}_{\leq m}=\cup_{i\leq m}\boldsymbol{Q}_{i}=\cup_{i\leq m}\boldsymbol{S}_{\text{right}}^{i}\cap\tilde{\boldsymbol{P}}_{i}$, we always have $\boldsymbol{Q}_{\leq m}\in\mathscr{S}(\boldsymbol{P}_{\leq m})$, $\boldsymbol{S}_{\text{proj}}^{m}=\boldsymbol{S}_{\text{proj}}^{m}(\boldsymbol{Q}_{\leq m})$ and $\tilde{\boldsymbol{P}}_{\text{invalid}}^{m}=\tilde{\boldsymbol{P}}_{\text{invalid}}^{m}(\boldsymbol{Q}_{\leq m})$. In the following we give the upper bound of the number of pairs $(\boldsymbol{S}_{\text{proj}}^{m},\tilde{\boldsymbol{P}}_{\text{invalid}}^{m})$ and the number of $\boldsymbol{S}_{\text{right}}^{m}$, respectively.

First, we can prove that the pair $(\boldsymbol{S}_{\text{proj}}^{m},\tilde{\boldsymbol{P}}_{\text{invalid}}^{m})$ can be determined from 
\begin{equation}
\boldsymbol{S}^{\prime}=\mathbb{P}_{>m-2k+2}(\langle\boldsymbol{Q}_{\leq m}\rangle).
\end{equation}
For $\boldsymbol{S}_{\text{proj}}^{m}$ we have
\begin{align}
\boldsymbol{S}_{\text{proj}}^{m} & =\mathbb{P}_{>m-k}(\langle\boldsymbol{Q}_{\leq m}\rangle)\\
 & =\mathbb{P}_{>m-k}(\boldsymbol{S}^{\prime}).
\end{align}
For $\tilde{\boldsymbol{P}}_{\text{invalid}}^{m}$, since $P\in\tilde{\boldsymbol{P}}_{>m}$ has $q_{P}^{\text{begin}}>m-k+1$, any $Q\in\boldsymbol{Q}_{<m}$ with $\{P,Q\}=0$ must have $q_{P}^{\text{last}}>m-k+1$ and thus $q_{P}^{\text{begin}}>m-2k+2$, so we have $Q\in\boldsymbol{S}^{\prime}$. Thus
\begin{align}
\tilde{\boldsymbol{P}}_{\text{invalid}}^{m} & =\{P\in\tilde{\boldsymbol{P}}_{>m}|[P,\boldsymbol{Q}_{\leq m}]\neq0\}\\
 & =\{P\in\tilde{\boldsymbol{P}}_{>m}|[P,\boldsymbol{S}^{\prime}]\neq0\}.
\end{align}
For simplify we use $\boldsymbol{S}_{\text{proj}}^{m}(\boldsymbol{S}^{\prime})$ and $\tilde{\boldsymbol{P}}_{\text{invalid}}^{m}(\boldsymbol{S}^{\prime})$ as the shorthands of the above two formulas.

Next, we have
\begin{equation}
\boldsymbol{S}^{\prime}\in\{\langle\boldsymbol{Q}\rangle|\boldsymbol{Q}\in\mathscr{S}(\mathbb{T}_{>m-2k+2}(\tilde{\boldsymbol{P}}_{\leq m}))\}\label{eq:num_Sprime}
\end{equation}
Since any $P\in\boldsymbol{P}_{\leq m}$ with $\mathbb{T}_{>m-2k+2}\neq I$ must have $q_{P}^{\text{last}}>m-2k+2$, thus $\mathbb{T}_{>m-2k+2}(\tilde{\boldsymbol{P}}_{\leq m})$ has at most $(2k-2)M$ non-identity elements. Also $\boldsymbol{S}^{\prime}$ is a $2k-2$ qubit stabilizer group. According to \thref{nS}, there are at most $(2(2k-2)M)^{2k-1}$ choices of $\boldsymbol{S}^{\prime}$.

Finally, according to Eq. \ref{eq:relaxed_Sright_next}, we have 
\begin{equation}
\boldsymbol{S}_{\text{right}}^{m}\in\mathcal{T}_{\text{right}}^{m}\equiv\{\langle\boldsymbol{Q}\rangle|\boldsymbol{Q}\in\mathscr{S}(\mathbb{T}_{\leq m}(\boldsymbol{P}_{\geq m}))\}\label{eq:num_Sright}
\end{equation}
For the similar reason, there are at most $kM$ non-identity elements in $\mathbb{T}_{\leq m}(\tilde{\boldsymbol{P}}_{\geq m})$. Thus there are at most $(2kM)^{k+1}$ choices of $\boldsymbol{S}_{\text{right}}^{m}$.

Combining the above results, we can construct
\begin{equation}
\tilde{\mathcal{A}}_{m}^{\prime}=\{(\boldsymbol{S}_{\text{proj}}^{m}(\boldsymbol{S}^{\prime}),\tilde{\boldsymbol{P}}_{\text{invalid}}^{m}(\boldsymbol{S}^{\prime}),\boldsymbol{S}_{\text{right}}^{m})|\boldsymbol{S}^{\prime},\boldsymbol{S}_{\text{right}}^{m}\text{\text{satisfy Eq. \ref{eq:num_Sprime} and Eq. \ref{eq:num_Sright}}}\}.
\end{equation}
Clearly we have $\tilde{\mathcal{A}}_{m}^{\prime}\supseteq\tilde{\mathcal{A}}_{m}$ and $|\tilde{\mathcal{A}}_{m}^{\prime}|\leq(2(2k-2)M)^{2k-1}\cdot(2kM)^{k+1}<(4kM)^{3k}$. Also both $\mathbb{T}_{\leq m}(\tilde{\boldsymbol{P}}_{\geq m})$ and $\mathbb{T}_{>m-2k+2}(\tilde{\boldsymbol{P}}_{<m})$ can be determined by $\mathbb{T}_{m-2k+1,m}(\tilde{\boldsymbol{P}})$. Thus $\tilde{\mathcal{A}}_{m}^{\prime}$ can be determined by $\mathbb{T}_{m-2k+1,m}(\tilde{\boldsymbol{P}})$.
\end{proof}

\subsection{Proof of \thref{periodic_SM} \label{appendix:proof_periodic}}

We prove that the stabilizer ground state per site of an infinite periodic local and sparse Hamiltonian with period $l$ is given by:
\begin{equation}
E_{\text{gs}}^{\text{periodic}}=\min_{c,A,\{A_{m=0}^{cl}|A_{i+1}\in\tilde{F}(A_{i}),A_{0}\simeq A_{cl}\simeq A\}}\frac{1}{cl}\sum_{m=1}^{cl}h_{m}(G(A_{m})).
\end{equation}
Strictly speaking, let $H_{n}$ be the $n$-qubit Hamiltonian truncated from the infinite Hamiltonian (with arbitrary truncation strategy at edges), $E_{\text{gs}}^{\text{periodic}}$ above is defined by the minimum per-site stabilizer ground state of $H_{n}$ in the limit of $n\rightarrow\infty$.
\begin{equation}
E_{\text{gs}}^{\text{periodic}}=\lim_{n\rightarrow\infty}\frac{1}{n}\min_{\text{\ensuremath{n}-qubit stabilizer state}|\psi\rangle}\langle\psi|H_{n}|\psi\rangle.
\end{equation}

We first prove that, for any chain of state machine $A_{0},A_{1},...A_{n}$ such that $A_{m+1}\in F(A_{m})$ for each $m$, if $n\geq N_{A}l$ we can always find $c_{1}$ and $c_{2}$ that $A_{c_{1}l}\simeq A_{c_{2}l}$, and $|c_{1}-c_{2}|\leq N_{A}$, where $N_{A}$ is defined in \thref{local_nS}. According to \thref{local_nS}, each possible $A_{m}$ is taken from candidate values $\tilde{\mathcal{A}}_{m}^{\prime}\supseteq\tilde{\mathcal{A}}_{m}$, and $\tilde{\mathcal{A}}_{m}^{\prime}$ is periodic with period $l$, i.e. $\tilde{\mathcal{A}}_{m}^{\prime}\simeq\tilde{\mathcal{A}}_{m+l}^{\prime}$, since it is determined by local terms $\mathbb{T}_{m-2k+1,m}(\tilde{\boldsymbol{P}})$ in the Hamiltonian. Now we consider the equivalence set $\{A_{cl}|c\in\mathbb{Z}\}$. Since each $A_{cl}\in\tilde{\mathcal{A}}_{cl}\simeq\tilde{\mathcal{A}}_{0}$ for each $c$, and $\tilde{\mathcal{A}}_{0}$ has up to $N_{A}$ values according to \thref{local_nS}, we can find some $A_{c_{1}l}\simeq A_{c_{2}l}$ with $|c_{1}-c_{2}|\leq N_{A}$.

If $E_{\text{gs}}^{\text{periodic}}$ is not the stabilizer ground state energy per site, then for sufficiently large $n$, the stabilizer ground state energy $E_{\min}^{(n)}$ of $H_{n}=\sum_{P\in\boldsymbol{P}}w_{P}P$ can be lower than $nE_{\text{gs}}^{\text{periodic}}$ by arbitrary amount of energy $E_{C}$, i.e. $E_{\min}^{(n)}\leq nE_{\text{gs}}^{\text{periodic}}-E_{C}$. Let the stabilizer group of the corresponding $n$-qubit stabilizer state be $\boldsymbol{S}=\langle\boldsymbol{Q}^{\star}\rangle$, $\boldsymbol{Q}^{\star}\in\mathscr{S}(\boldsymbol{P})$, and let $A_{m}=A_{m}(\boldsymbol{Q}^{\star})$ for each $m$. The expectation energy is given by $E=\sum_{m=1}^{n}h_{m}(G(A_{m}))$. If $n>N_{A}l$, we can find $0\leq c_{1}\leq c_{2}\leq n$ such that $A_{c_{1}l}\simeq A_{c_{2}l}\simeq A$. Then we have $\Delta E=\sum_{m=c_{1}l+1}^{c_{2}l}h_{m}(G(A_{m}))\geq(c_{2}-c_{1})lE_{\text{gs}}^{\text{periodic}}$ according to the definition of $E_{\text{gs}}^{\text{periodic}}$. Now we remove $\{A_{i=c_{1}l+1}^{c_{2}l}\}$ from the path $\{A_{m=1}^{n}\}$, and consider the new path $\{A_{i}^{\prime}|0\leq i\leq n^{\prime}\}$, such that $A_{i}^{\prime}=A_{i}$ for $i\leq c_{1}l$ and $A_{i}^{\prime}\simeq A_{i+(c_{2}-c_{1})l}$ for $i>c_{1}l$, where $n^{\prime}=n-(c_{2}-c_{1})l$. Obviously, it still satisfies $A_{m+1}^{\prime}\in F(A_{m}^{\prime})$ for each $m$, thus it maps to some $n^{\prime}$-qubit stabilizer state with energy $E^{(n^{\prime})}=E_{\min}^{(n)}-\Delta E\leq E_{\min}^{(n)}-(c_{2}-c_{1})\tilde{E}_{\text{min}}\leq n^{\prime}\tilde{E}_{\min}-E_{C}$. We can then continue the above deleting process and create $\{A_{i}^{\prime\prime}|0\leq i\leq n^{\prime\prime}\}$, $\{A_{i}^{\prime\prime\prime}|0\leq i\leq n^{\prime\prime\prime}\}$, ..., until we end up with some $n^{*}\leq N_{A}l$ and the corresponding path $\{A_{i}^{*}|0\leq i\leq n^{*}\}$. Such path can map to some $n^{*}$-qubit stabilizer state $|\psi^{*}\rangle$ with energy $\langle\psi^{*}|H_{n^{*}}|\psi^{*}\rangle=E^{(n^{*})}\leq n^{*}\tilde{E}_{\min}-E_{C}$. However $\langle\psi^{*}|H_{n^{*}}|\psi^{*}\rangle$ should have a finite lower bound in the order of $O(N_{A}l)$, which is independent of $n$. Thus $E_{C}$ cannot be arbitrarily large, which conflicts with the assumption. Thus we conclude that $E_{\text{gs}}^{\text{periodic}}$ is the stabilizer ground state energy per site.

\subsection{Proof of \thref{local_preparation} \label{appendix:proof_local_preparation}}

\begin{figure}[tbh]
\centering

\includegraphics[width=0.75\linewidth]{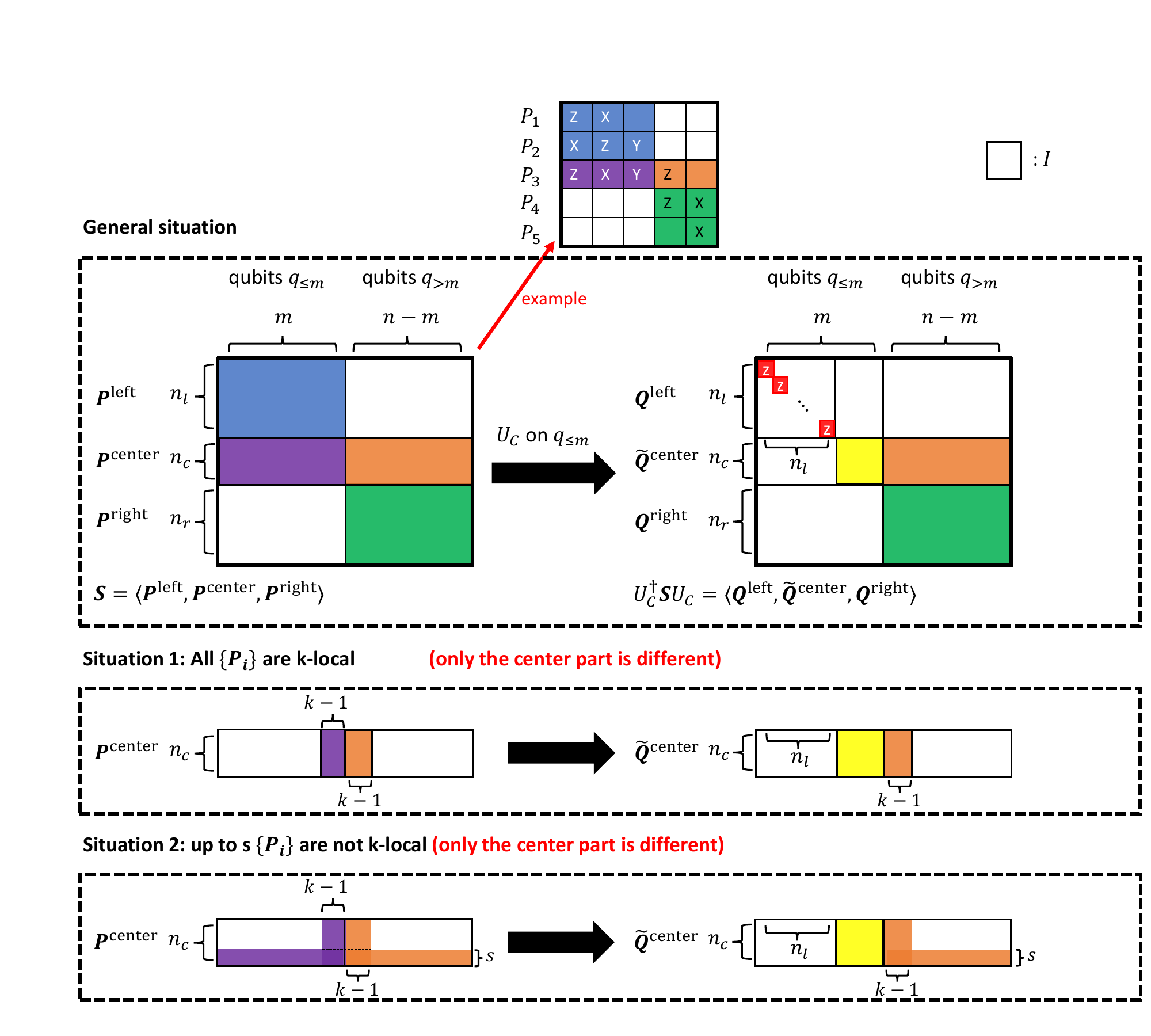}

\caption{Schematic diagram of a single iteration of the algorithm for Theorem 4. Blocks in the same color indicate the same values. Three situations are discussed, including (1) the general situation that all $\{P_{i}\}$ have no locality, (2) all $\{P_{i}\}$ are $k$-local, and (3) up to $s$ elements in $\{P_{i}\}$ are not $k$-local. The performed operation is the same for these three situations (i.e. simply the same algorithm with different types of inputs) except the choice of $m$ in these cases is different. For a given $m$, the qubits are divided into the left part $q_{\protect\leq m}$ and the right part $q_{>m}$, and $\{P_{i}\}$ are divided to $\boldsymbol{P}^{\text{left}}$, $\boldsymbol{P}^{\text{center}}$, and $\boldsymbol{P}^{\text{right}}$ by the part that each $P_{i}$ acts on. A Clifford $U_{C}$ is applied on $q_{\protect\leq m}$ to transform $\boldsymbol{P}^{\text{left}}$ to $\{Z_{1},Z_{2},...,Z_{n_{l}}\}$. One can prove that the other generators of the stabilizer groups are effectively transformed to qubits $q_{>n_{l}}$, i.e. the generators become two decoupled parts. In the general situation that $\{P_{i}\}$ are non-local, we cannot guarantee that $n_{l}>0$. If there are up to $s$ elements in $\{P_{i}\}$ that are not $k$-local, then for any $m\protect\geq3k-3$, one can guarantee that $n_{l}\protect\geq m-2k+2-s$. Thus we can let $m=4(k+s)$, which gives $n_{l}\protect\geq2k+3s+2\sim O(k+s)$. Thus with $O(n/(k+s))$ numbers of $O(k+s)$-qubit Clifford transformations we transform $\langle\boldsymbol{P}\rangle$ to $\langle Z_{1},...,Z_{n}\rangle$.}

\label{fig:preparation}

\end{figure}

Let $q_{I}$ be a shorthand of qubits in the index set $I$ with the same convention of $\tilde{\boldsymbol{P}}_{I}$ in Sec. \ref{sec:mapping}. We first consider a general situation that $\boldsymbol{P}=\{P_{i}\}$ are nonlocal, and we introduce a procedure to transform the stabilizers $\boldsymbol{S}$ and the generators \textbf{$\boldsymbol{P}$}. (see Fig \ref{fig:preparation}, general situation) Let $m$ be some arbitray integer, divide $\boldsymbol{P}=\{P_{1},..,P_{n}\}$ to three parts $\boldsymbol{P}^{\text{left}}=\{P_{1}^{\text{left}},...,P_{n_{l}}^{\text{left}}\}$, $\boldsymbol{P}^{\text{center}}=\{P_{1}^{\text{center}},...,P_{n_{c}}^{\text{center}}\}$, and $\boldsymbol{P}^{\text{right}}=\{P_{1}^{\text{right}},...,P_{n_{r}}^{\text{right}}\}$ such that, each $P_{i}^{\text{left}}$ only acts on qubits $q_{\leq m}$ (the blue block), each $P_{i}^{\text{right}}$ only acts on qubits $q_{>m}$ (the green block), and each $P_{i}^{\text{center}}$ acts on both qubits $q_{\leq m}$ (the purple block) and $q_{>m}$ (the orange block), where $n_{l},n_{c},n_{r}$ are number of elements in $\boldsymbol{P}^{\text{left}},\boldsymbol{P}^{\text{center}},\boldsymbol{P}^{\text{right}}$ respectively. For each $P_{i}^{\text{center}}$, we can uniquely decompose it to the left and right part, i.e. $P_{i}^{\text{center}}=P_{l,i}^{\text{center}}P_{r,i}^{\text{center}}$ (up to a choice of $\pm1$ sign), where $P_{l,i}^{\text{center}}$ (the purple block) only acts on $q_{\leq m}$, and $P_{r,i}^{\text{center}}$ (the orange block) only acts on qubits $q_{>m}$. Next, we construct a Clifford transformation $U_{C}$ on $q_{\leq m}$ such that $U_{C}^{\dagger}P_{i}^{\text{left}}U_{C}=Z_{i}$ (small red blocks) for each $P_{i}^{\text{left}}$. The existence of such $U_{C}$ comes from the fact that elements in $\boldsymbol{P}_{\text{left}}$ commute with each other and are independent. Now we consider the elements in $U_{C}^{\dagger}\boldsymbol{P}U_{C}$, which are generators of the transformed stabilizer group $U_{C}^{\dagger}\boldsymbol{S}U_{C}$. We keep the above conventions but use symbol $Q$ to denote the transformed Pauli operators, i.e. $\boldsymbol{Q}=U_{C}^{\dagger}\boldsymbol{P}U_{C}$, $\boldsymbol{Q}^{\mu}=U_{C}^{\dagger}\boldsymbol{P}^{\mu}U_{C}$ and $Q_{i}^{\mu}=U_{C}^{\dagger}P_{i}^{\mu}U_{C}$ for $\mu=\text{left},\text{center},\text{right}$. Thus we have $Q_{i}^{\text{left}}=Z_{i}$ according to the definition, and $Q_{i}^{\text{right}}=P_{i}^{\text{right}}$ since $U_{C}$ acts on $q_{\leq m}$. For $Q_{i}^{\text{center}}=Q_{l,i}^{\text{center}}Q_{r,i}^{\text{center}}$, we have $Q_{r,i}^{\text{center}}=P_{r,i}^{\text{center}}$. We then focus on $Q_{l,i}^{\text{center}}$. Since elements of a stabilizer group commute with each other, we have
\begin{equation}
\begin{aligned}[][Q_{l,i}^{\text{center}},Z_{j}^{\text{left}}] & =[Q_{l,i}^{\text{center}},Q_{j}^{\text{left}}]\\
 & =[P_{l,i}^{\text{center}},P_{j}^{\text{left}}]\\
 & =[P_{i}^{\text{center}},P_{j}^{\text{left}}]\\
 & =0
\end{aligned}
\end{equation}
for each $i,j$, where the first line uses $Q_{j}^{\text{left}}=Z_{j}^{\text{left}}$, the second line uses the fact that Clifford transformation does not change commutation relation, the third line uses the fact that $P_{r,i}^{\text{center}}$ is on $q_{>m}$ so commutes with $P_{j}^{\text{left}}$, and the last line use the fact that generators of a stabilizer group commute. Thus we can write 
\begin{equation}
Q_{l,i}^{\text{center}}=\Big(\prod_{k=1}^{n_{l}}(Q_{k}^{\text{left}})^{s_{k}}\Big)\tilde{Q}_{l,i}^{\text{center}}
\end{equation}
such that $s_{k}\in\{0,1\}$, and $\tilde{Q}_{l,i}^{\text{center}}$ (the yellow block) acts on qubits $q_{n_{l}+1,m}$ (i.e. between qubit $n_{l}+1$ and $m$). Since $\langle P,PQ\rangle=\langle P,Q\rangle$, for the transformed stabilizer group $\langle\boldsymbol{Q}\rangle=\{\boldsymbol{Q}^{\text{left}},\boldsymbol{Q}^{\text{center}},\boldsymbol{Q}^{\text{right}}\}$, we can now remove the $\prod_{k=1}^{n_{l}}(Q_{k}^{\text{left}})^{s_{k}}$ part in each $Q_{i}^{\text{center}}=\prod_{k=1}^{n_{l}}(Q_{k}^{\text{left}})^{s_{k}}\tilde{Q}_{l,i}^{\text{center}}Q_{r,i}$. Rigorously, let $\tilde{Q}_{i}^{\text{center}}=\tilde{Q}_{l,i}^{\text{center}}P_{r,i}^{\text{center}}$ and $\tilde{\boldsymbol{Q}}^{\text{center}}=\{\tilde{Q}_{i}^{\text{center}}\}$, we have $\langle\boldsymbol{Q}\rangle=\langle\boldsymbol{Q}^{\text{left}},\tilde{\boldsymbol{Q}}^{\text{center}},\boldsymbol{Q}^{\text{right}}\rangle$ as well, where elements in $\tilde{\boldsymbol{Q}}^{\text{center}}$ acts on $q_{>n_{l}}$ and elements in $\boldsymbol{Q}^{\text{right}}$ acts on $q_{>m}$. Thus the first $n_{l}$ sites are decoupled from the other sites in the transformed stabilizer group $\langle\boldsymbol{Q}\rangle$, i.e. $\langle\boldsymbol{Q}\rangle=\langle Z_{1},...,Z_{n_{l}}\rangle\otimes\langle\tilde{\boldsymbol{Q}}^{\text{center}},\boldsymbol{Q}^{\text{right}}\rangle$, where $\langle\tilde{\boldsymbol{Q}}^{\text{center}},\boldsymbol{Q}^{\text{right}}\rangle$ is on qubits $q_{>n_{l}}$. Therefore, we have successfully reduced the $n$-qubit problem to a $n-n_{l}$ qubit problem on $q_{>n_{l}}$. All we do until now is some Clifford operation on $q_{\leq m}$. Since $\tilde{Q}_{l,i}^{\text{center}}$ can be viewed as a truncation of $Q_{l,i}^{\text{center}}$ with the global sign kept, we write it as
\begin{equation}
\tilde{\boldsymbol{Q}}^{\text{center}}=\mathbb{T}_{>n_{l}}^{*}(U_{C}^{\dagger}\boldsymbol{P}^{\text{center}}U_{C}),
\end{equation}
where $\mathbb{T}_{I}^{*}$ indicates the truncation operation to qubits $I$ with signs kept (similar to \thref{def:truncation}).

For general stabilizer groups and generators, there is no guarantee that $n_{l}>0$, so the above procedure may eventually do nothing. Now we consider the situation that $P_{1},...,P_{n}$ are all $k$-local (see Fig. \ref{fig:preparation}, situation 1), and we hope to construct some $m\sim O(k)$, so that when we can iteratively execute the above procedure, we can reduce the problem size by some $n_{l}>0$ in each iteration until $n=0$. When we execute the above procedure for the first time, each $P_{i}^{\text{center}}$ (the purple and orange block) is between qubit $m-k+2$ to $m+k-1$ (totally $2k-2$) due to the $k$-local condition, so $n_{c}\leq2k-2$, where we used the fact that a $n$-qubit stabilizer group has up to $n$ independent generators. For the same reason, we have $n_{l}\leq m$ and $n_{r}\leq n-m$. Since $n_{c}+n_{l}+n_{r}=n$, we should have $n_{l}=n-n_{c}-n_{r}\geq n-(2k-2)-(n-m)=m-2k+2$. Thus we just need to let $m>2k-2$, so that we have $n_{l}>0$ in the first iteration. When it comes to later iterations, the situation becomes slightly different. As explained above, the ``right'' part is unchanged (i.e.$\boldsymbol{Q}^{\text{right}}=\boldsymbol{P}^{\text{right}}$), so each $Q_{i}^{\text{right}}$ is still $k$-local. However, for each $\tilde{Q}_{i}^{\text{center}}=\tilde{Q}_{l,i}^{\text{center}}P_{r,i}^{\text{center}}$, $\tilde{Q}_{l,i}^{\text{center}}$ (the yellow block) is on $q_{n_{l}+1,m}$ and $P_{r,i}^{\text{center}}$ (the orange block) is on $q_{m+1,m+k-1},$so we can only ensure that $Q_{i}^{\text{center}}$ is on $q_{n_{l}+1,m+k-1}$. In the reduced $n-n_{l}$ qubit problem (i.e. the first $n_{l}$ qubits are removed), it is between qubit $1$ and $m+k-n_{l}-1\leq3k-3$, so it is $3k-3$ local. Fortunately, in the above analysis to derive $n_{l}\geq m-2k+2$, we only require each $P\in\boldsymbol{P}^{\text{center}}$ to be $k$-local, but we don't require it for $\boldsymbol{P}^{\text{left}}$ and $\boldsymbol{P}^{\text{right}}$. Therefore we want to additionally ensure that all $\tilde{Q}_{i}^{\text{center}}$ in the previous iteration enter into $\boldsymbol{P}^{\text{left}}$ in the next iteration, so $\boldsymbol{P}^{\text{center}}$ in the next iteration are all taken from $\boldsymbol{Q}^{\text{right}}=\boldsymbol{P}^{\text{right}}$ in the previous iteration, which are $k$-local. Recall that $\boldsymbol{P}^{\text{left}}$ are elements of $\boldsymbol{P}$ on $q_{\leq m}$ only, so we just need to have $m\geq3k-3$. To summarize, with $m\geq3k-3$, we can iteratively apply a $m$-qubit Clifford transformation to reduce the problem size by $n_{l}\geq m-2k+2$. Thus we could, for example, let $m=4k$, so each time we reduce the problem size by $O(k)$. Finally, with $O(n/k)$ number of $O(k)$-qubit Clifford transformations, we can transform the original stabilizer group to $\langle Z_{1},Z_{2},...,Z_{n}\rangle$. Since $O(k)$-qubit Clifford transformation can decompose to $O(k^{2}/\log k)$ single-qubit and double-qubit Clifford transformations \cite{aaronson2004improved}, the total number of single-qubit and double-qubit Clifford transformations we need is $O(n/k)\times O(k^{2}/\log k)=O(nk/\log k)$.

Finally, we consider the situation that there are up to $s$ elements in $\{P_{1},...,P_{n}\}$ that are not $k$-local. (see Fig. \ref{fig:preparation}, situation 2) In fact, we can use the same procedure with $m=4(k+s)$ and no other modifications. To see the reason, in the first iteration, we have all these $s$ elements classified into $\boldsymbol{P}^{\text{center}}$ (the wide purple and orange blocks at the bottom), so we have $n_{c}=|\boldsymbol{P}^{\text{center}}|\leq2k-2+s$ ($2k-2$ local elements and $s$ non-local elements). After the Clifford transformation, for those $k$-local $P_{i}^{\text{center}}$ (the narrow purple block), the corresponding $Q_{i}^{\text{center}}$ (the narrow orange block) are still $3k-3$ local, and for those non-local $P_{i}^{\text{center}}$ (the wide purple block), the corresponding $Q_{i}^{\text{center}}$ (the wide orange bock) are still generally nonlocal (can also accidentally be local). Thus in the next iteration, those local $Q_{i}^{\text{center}}$ enter into $\boldsymbol{P}^{\text{left}}$ (including the nonlocal ones that accidentally become local), and those non-local $Q_{i}^{\text{center}}$ enter into $\boldsymbol{P}^{\text{center}}$. So we still have $n_{c}\leq2k-2+s$ in the next iteration. Thus, with $m\geq3k-3$, we can iteratively apply $m$-qubit Clifford transformation to reduce the problem size by $n_{l}\geq m-n_{c}=m-2k+2-s$. By choosing, for example, $m=4(k+s)$, we achieve the same conclusion with $L\sim O(nk^{\prime}/\log k^{\prime})$, where $k^{\prime}=k+s$. In fact, when $s=0$, this algorithm reduces to the previous situation. A pseudo-code of this algorithm in shown in Algorithm \ref{alg:local_preparation}.

\begin{algorithm}[H]
\caption{Algorithm to find Clifford operations for \thref{local_preparation}}
\label{alg:local_preparation}
\begin{algorithmic}[1]
\State $m = 4(k+s)$
\State Initiate an empty list of single-qubit and double-qubit Clifford operations $\boldsymbol{U}$
\While{$n = |\boldsymbol{P}| > 0$}
    \State Classify each $P\in\boldsymbol{P}$ to $\boldsymbol{P}^{\text{left}}$, $\boldsymbol{P}^{\text{right}}$, $\boldsymbol{P}^{\text{center}}$ if $P$ acts on only $q_{\leq m}$, only $q_{>m}$, or both parts, respectively
    \State $n_{l} = |\boldsymbol{P}^{\text{left}}|$
    \State Construct Clifford transformation $U_{C}$ such that $U_{C}^{\dagger}\boldsymbol{P}^{\text{left}}U_{C} = \{Z_{1}, ..., Z_{n_{l}}\}$
    \State Decompose $U_{C}$ to single-qubit and double-qubit Clifford operations $U_{C} = U_{1} U_{2} \cdots U_{N}$ with the standard method \cite{aaronson2004improved}
    \State Append $\{U_{1}, U_{2}, ..., U_{N}\}$ (the order matters) to the end of $\boldsymbol{U}$
    \State  $\boldsymbol{P}\rightarrow\mathbb{T}_{>n_{l}}^{*}(U_{C}^{\dagger}\boldsymbol{P}^{\text{center}}U_{C})\cup\boldsymbol{P}^{\text{right}}$

\EndWhile
\State \textbf{Output} $\boldsymbol{U}$
\end{algorithmic}
\end{algorithm}

\bibliographystyle{quantum}
\bibliography{main}

\begin{thebibliography}{100}

\bibitem{Kitaev2002a_z}
A.Yu. Kitaev, A.~H. Shen, and M.~N. Vyalyi.
\newblock ``Quantum information''.
\newblock \href{https://dx.doi.org/10.1007/978-0-387-36944-0_13}{Page 203–217}.
\newblock Springer New York. ~(2002).

\bibitem{Aharonov2002_z}
Dorit Aharonov and Tomer Naveh.
\newblock ``Quantum np - a survey''~(2002).
\newblock  \href{http://arxiv.org/abs/quant-ph/0210077}{arXiv:quant-ph/0210077}.

\bibitem{Kempe2006_z}
Julia Kempe, Alexei Kitaev, and Oded Regev.
\newblock ``The complexity of the local hamiltonian problem''.
\newblock \href{https://dx.doi.org/10.1137/s0097539704445226}{SIAM Journal on Computing {\bf 35}, 1070–1097}~(2006).

\bibitem{Schuch2009a_z}
Norbert Schuch and Frank Verstraete.
\newblock ``Computational complexity of interacting electrons and fundamental limitations of density functional theory''.
\newblock \href{https://dx.doi.org/10.1038/nphys1370}{Nature Physics {\bf 5}, 732–735}~(2009).

\bibitem{Huang2020c_z}
Yichen Huang.
\newblock ``Two-dimensional local hamiltonian problem with area laws is qma-complete''.
\newblock \href{https://dx.doi.org/10.1016/j.jcp.2021.110534}{Journal of Computational Physics {\bf 443}, 110534}~(2021).

\bibitem{lee2023evaluating}
Seunghoon Lee, Joonho Lee, Huanchen Zhai, Yu~Tong, Alexander~M. Dalzell, Ashutosh Kumar, Phillip Helms, Johnnie Gray, Zhi-Hao Cui, Wenyuan Liu, Michael Kastoryano, Ryan Babbush, John Preskill, David~R. Reichman, Earl~T. Campbell, Edward~F. Valeev, Lin Lin, and Garnet Kin-Lic Chan.
\newblock ``Evaluating the evidence for exponential quantum advantage in ground-state quantum chemistry''.
\newblock \href{https://dx.doi.org/10.1038/s41467-023-37587-6}{Nature Communications{\bf 14}}~(2023).

\bibitem{Carleo2017_z}
Giuseppe Carleo and Matthias Troyer.
\newblock ``Solving the quantum many-body problem with artificial neural networks''.
\newblock \href{https://dx.doi.org/10.1126/science.aag2302}{Science {\bf 355}, 602–606}~(2017).

\bibitem{Deng2017c_z}
Dong-Ling Deng, Xiaopeng Li, and S.~Das~Sarma.
\newblock ``Quantum entanglement in neural network states''.
\newblock \href{https://dx.doi.org/10.1103/physrevx.7.021021}{Physical Review X{\bf 7}}~(2017).

\bibitem{Glasser2018b_z}
Ivan Glasser, Nicola Pancotti, Moritz August, Ivan~D. Rodriguez, and J.~Ignacio Cirac.
\newblock ``Neural-network quantum states, string-bond states, and chiral topological states''.
\newblock \href{https://dx.doi.org/10.1103/physrevx.8.011006}{Physical Review X{\bf 8}}~(2018).

\bibitem{Zhang2019b_z}
Shi-Xin Zhang, Zhou-Quan Wan, and Hong Yao.
\newblock ``Automatic differentiable monte carlo: Theory and application''.
\newblock \href{https://dx.doi.org/10.1103/physrevresearch.5.033041}{Physical Review Research{\bf 5}}~(2023).

\bibitem{Verstraete2008a_z}
F.~Verstraete, V.~Murg, and J.I. Cirac.
\newblock ``Matrix product states, projected entangled pair states, and variational renormalization group methods for quantum spin systems''.
\newblock \href{https://dx.doi.org/10.1080/14789940801912366}{Advances in Physics {\bf 57}, 143–224}~(2008).

\bibitem{Schollwock2011_z}
Ulrich Schollwöck.
\newblock ``The density-matrix renormalization group in the age of matrix product states''.
\newblock \href{https://dx.doi.org/10.1016/j.aop.2010.09.012}{Annals of Physics {\bf 326}, 96–192}~(2011).

\bibitem{schollwock2005density}
U.~Schollw\"{o}ck.
\newblock ``The density-matrix renormalization group''.
\newblock \href{https://dx.doi.org/10.1103/revmodphys.77.259}{Reviews of Modern Physics {\bf 77}, 259–315}~(2005).

\bibitem{kandala2017hardware}
Abhinav Kandala, Antonio Mezzacapo, Kristan Temme, Maika Takita, Markus Brink, Jerry~M. Chow, and Jay~M. Gambetta.
\newblock ``Hardware-efficient variational quantum eigensolver for small molecules and quantum magnets''.
\newblock \href{https://dx.doi.org/10.1038/nature23879}{Nature {\bf 549}, 242–246}~(2017).

\bibitem{cerezo2022variational}
M.~Cerezo, Kunal Sharma, Andrew Arrasmith, and Patrick~J. Coles.
\newblock ``Variational quantum state eigensolver''.
\newblock \href{https://dx.doi.org/10.1038/s41534-022-00611-6}{npj Quantum Information{\bf 8}}~(2022).

\bibitem{tilly2022variational}
Jules Tilly, Hongxiang Chen, Shuxiang Cao, Dario Picozzi, Kanav Setia, Ying Li, Edward Grant, Leonard Wossnig, Ivan Rungger, George~H. Booth, and Jonathan Tennyson.
\newblock ``The variational quantum eigensolver: A review of methods and best practices''.
\newblock \href{https://dx.doi.org/10.1016/j.physrep.2022.08.003}{Physics Reports {\bf 986}, 1–128}~(2022).

\bibitem{motta2018ab}
Mario Motta and Shiwei Zhang.
\newblock ``Ab initio computations of molecular systems by the auxiliary‐field quantum monte carlo method''.
\newblock \href{https://dx.doi.org/10.1002/wcms.1364}{WIREs Computational Molecular Science{\bf 8}}~(2018).

\bibitem{lee2022twenty}
Joonho Lee, Hung~Q. Pham, and David~R. Reichman.
\newblock ``Twenty years of auxiliary-field quantum monte carlo in quantum chemistry: An overview and assessment on main group chemistry and bond-breaking''.
\newblock \href{https://dx.doi.org/10.1021/acs.jctc.2c00802}{Journal of Chemical Theory and Computation {\bf 18}, 7024–7042}~(2022).

\bibitem{carlson2011auxiliary}
J.~Carlson, Stefano Gandolfi, Kevin~E. Schmidt, and Shiwei Zhang.
\newblock ``Auxiliary-field quantum monte carlo method for strongly paired fermions''.
\newblock \href{https://dx.doi.org/10.1103/physreva.84.061602}{Physical Review A{\bf 84}}~(2011).

\bibitem{schollwock2013matrix}
Ulrich Schollw{\"o}ck.
\newblock ``Matrix product state algorithms: {DMRG}, {TEBD} and relatives''.
\newblock In Springer Series in {Solid-State} Sciences.
\newblock \href{https://dx.doi.org/10.1007/978-3-642-35106-8_3}{Pages 67--98}.
\newblock Springer Berlin Heidelberg, Berlin, Heidelberg~(2013).

\bibitem{orus2008infinite}
R.~Orús and G.~Vidal.
\newblock ``Infinite time-evolving block decimation algorithm beyond unitary evolution''.
\newblock \href{https://dx.doi.org/10.1103/physrevb.78.155117}{Physical Review B{\bf 78}}~(2008).

\bibitem{haegeman2016unifying}
Jutho Haegeman, Christian Lubich, Ivan Oseledets, Bart Vandereycken, and Frank Verstraete.
\newblock ``Unifying time evolution and optimization with matrix product states''.
\newblock \href{https://dx.doi.org/10.1103/physrevb.94.165116}{Physical Review B{\bf 94}}~(2016).

\bibitem{motta2020determining}
Mario Motta, Chong Sun, Adrian T.~K. Tan, Matthew~J. O’Rourke, Erika Ye, Austin~J. Minnich, Fernando G. S.~L. Brandão, and Garnet Kin-Lic Chan.
\newblock ``Determining eigenstates and thermal states on a quantum computer using quantum imaginary time evolution''.
\newblock \href{https://dx.doi.org/10.1038/s41567-019-0704-4}{Nature Physics {\bf 16}, 205–210}~(2019).

\bibitem{mcardle2019variational}
Sam McArdle, Tyson Jones, Suguru Endo, Ying Li, Simon~C. Benjamin, and Xiao Yuan.
\newblock ``Variational ansatz-based quantum simulation of imaginary time evolution''.
\newblock \href{https://dx.doi.org/10.1038/s41534-019-0187-2}{npj Quantum Information{\bf 5}}~(2019).

\bibitem{dargel2011adaptive}
P.~E. Dargel, A.~Honecker, R.~Peters, R.~M. Noack, and T.~Pruschke.
\newblock ``Adaptive lanczos-vector method for dynamic properties within the density matrix renormalization group''.
\newblock \href{https://dx.doi.org/10.1103/physrevb.83.161104}{Physical Review B{\bf 83}}~(2011).

\bibitem{cortes2022quantum}
Cristian~L. Cortes and Stephen~K. Gray.
\newblock ``Quantum krylov subspace algorithms for ground- and excited-state energy estimation''.
\newblock \href{https://dx.doi.org/10.1103/physreva.105.022417}{Physical Review A{\bf 105}}~(2022).

\bibitem{kitaev1995quantum}
A.~Yu. Kitaev.
\newblock ``Quantum measurements and the abelian stabilizer problem''~(1995).
\newblock  \href{http://arxiv.org/abs/quant-ph/9511026}{arXiv:quant-ph/9511026}.

\bibitem{kitaev2002classical}
Daniel Gottesman, A.~Yu. Kitaev, A.~H. Shen, and M.~N. Vyalyi.
\newblock ``Classical and quantum computation''.
\newblock \href{https://dx.doi.org/10.2307/3647986}{The American Mathematical Monthly {\bf 110}, 969}~(2003).

\bibitem{chaikin1995principles}
P.~M. Chaikin and T.~C. Lubensky.
\newblock ``Principles of condensed matter physics''.
\newblock \href{https://dx.doi.org/10.1017/cbo9780511813467}{Cambridge University Press}. ~(1995).

\bibitem{kadanoff2009more}
Leo~P. Kadanoff.
\newblock ``More is the same; phase transitions and mean field theories''.
\newblock \href{https://dx.doi.org/10.1007/s10955-009-9814-1}{Journal of Statistical Physics {\bf 137}, 777–797}~(2009).

\bibitem{negele1982mean}
J.~W. Negele.
\newblock ``The mean-field theory of nuclear structure and dynamics''.
\newblock \href{https://dx.doi.org/10.1103/revmodphys.54.913}{Reviews of Modern Physics {\bf 54}, 913–1015}~(1982).

\bibitem{lykos1963discussion}
P.~Lykos and G.~W. Pratt.
\newblock ``Discussion on the hartree-fock approximation''.
\newblock \href{https://dx.doi.org/10.1103/revmodphys.35.496}{Reviews of Modern Physics {\bf 35}, 496–501}~(1963).

\bibitem{kitaev2006anyons}
Alexei Kitaev.
\newblock ``Anyons in an exactly solved model and beyond''.
\newblock \href{https://dx.doi.org/10.1016/j.aop.2005.10.005}{Annals of Physics {\bf 321}, 2–111}~(2006).

\bibitem{chen2012symmetry}
Xie Chen, Zheng-Cheng Gu, Zheng-Xin Liu, and Xiao-Gang Wen.
\newblock ``Symmetry-protected topological orders in interacting bosonic systems''.
\newblock \href{https://dx.doi.org/10.1126/science.1227224}{Science {\bf 338}, 1604–1606}~(2012).

\bibitem{lyakh2012multireference}
Dmitry~I. Lyakh, Monika Musiał, Victor~F. Lotrich, and Rodney~J. Bartlett.
\newblock ``Multireference nature of chemistry: The coupled-cluster view''.
\newblock \href{https://dx.doi.org/10.1021/cr2001417}{Chemical Reviews {\bf 112}, 182–243}~(2011).

\bibitem{lischka2018multireference}
Hans Lischka, Dana Nachtigallová, Adélia J.~A. Aquino, Péter~G. Szalay, Felix Plasser, Francisco B.~C. Machado, and Mario Barbatti.
\newblock ``Multireference approaches for excited states of molecules''.
\newblock \href{https://dx.doi.org/10.1021/acs.chemrev.8b00244}{Chemical Reviews {\bf 118}, 7293–7361}~(2018).

\bibitem{nielsen2001quantum}
Michael~A. Nielsen and Isaac~L. Chuang.
\newblock ``Quantum computation and quantum information''.
\newblock \href{https://dx.doi.org/10.1017/cbo9780511976667}{Cambridge University Press}. ~(2012).

\bibitem{gottesman1998theory}
Daniel Gottesman.
\newblock ``Theory of fault-tolerant quantum computation''.
\newblock \href{https://dx.doi.org/10.1103/physreva.57.127}{Physical Review A {\bf 57}, 127–137}~(1998).

\bibitem{gottesman1998heisenberg}
Daniel Gottesman.
\newblock ``The heisenberg representation of quantum computers''~(1998).
\newblock  \href{http://arxiv.org/abs/quant-ph/9807006}{arXiv:quant-ph/9807006}.

\bibitem{aaronson2004improved}
Scott Aaronson and Daniel Gottesman.
\newblock ``Improved simulation of stabilizer circuits''.
\newblock \href{https://dx.doi.org/10.1103/physreva.70.052328}{Physical Review A{\bf 70}}~(2004).

\bibitem{campbell2017roads}
Earl~T. Campbell, Barbara~M. Terhal, and Christophe Vuillot.
\newblock ``Roads towards fault-tolerant universal quantum computation''.
\newblock \href{https://dx.doi.org/10.1038/nature23460}{Nature {\bf 549}, 172–179}~(2017).

\bibitem{howard2017application}
Mark Howard and Earl Campbell.
\newblock ``Application of a resource theory for magic states to fault-tolerant quantum computing''.
\newblock \href{https://dx.doi.org/10.1103/physrevlett.118.090501}{Physical Review Letters{\bf 118}}~(2017).

\bibitem{chamberland2019fault}
Christopher Chamberland and Andrew~W. Cross.
\newblock ``Fault-tolerant magic state preparation with flag qubits''.
\newblock \href{https://dx.doi.org/10.22331/q-2019-05-20-143}{Quantum {\bf 3}, 143}~(2019).

\bibitem{Zeng2019_z}
Bei Zeng, Xie Chen, Duan-Lu Zhou, and Xiao-Gang Wen.
\newblock ``Quantum information meets quantum matter: From quantum entanglement to topological phases of many-body systems''.
\newblock \href{https://dx.doi.org/10.1007/978-1-4939-9084-9}{Springer New York}. ~(2019).

\bibitem{perez2006matrix}
D.~Perez-Garcia, F.~Verstraete, M.M. Wolf, and J.I. Cirac.
\newblock ``Matrix product state representations''.
\newblock \href{https://dx.doi.org/10.26421/qic7.5-6-1}{Quantum Information and Computation {\bf 7}, 401–430}~(2007).

\bibitem{fattal2004entanglement}
David Fattal, Toby~S. Cubitt, Yoshihisa Yamamoto, Sergey Bravyi, and Isaac~L. Chuang.
\newblock ``Entanglement in the stabilizer formalism''~(2004).
\newblock  \href{http://arxiv.org/abs/quant-ph/0406168}{arXiv:quant-ph/0406168}.

\bibitem{Webb2016_z}
Zak Webb.
\newblock ``The clifford group forms a unitary 3-design''.
\newblock \href{https://dx.doi.org/10.26421/qic16.15-16-8}{Quantum Information and Computation {\bf 16}, 1379–1400}~(2016).

\bibitem{Huang2020b_z}
Hsin-Yuan Huang, Richard Kueng, and John Preskill.
\newblock ``Predicting many properties of a quantum system from very few measurements''.
\newblock \href{https://dx.doi.org/10.1038/s41567-020-0932-7}{Nature Physics {\bf 16}, 1050–1057}~(2020).

\bibitem{Nahum2017_z}
Adam Nahum, Jonathan Ruhman, Sagar Vijay, and Jeongwan Haah.
\newblock ``Quantum entanglement growth under random unitary dynamics''.
\newblock \href{https://dx.doi.org/10.1103/physrevx.7.031016}{Physical Review X{\bf 7}}~(2017).

\bibitem{VonKeyserlingk2018a_z}
C.~W. von Keyserlingk, Tibor Rakovszky, Frank Pollmann, and S.~L. Sondhi.
\newblock ``Operator hydrodynamics, otocs, and entanglement growth in systems without conservation laws''.
\newblock \href{https://dx.doi.org/10.1103/physrevx.8.021013}{Physical Review X{\bf 8}}~(2018).

\bibitem{Nahum2018a_z}
Adam Nahum, Sagar Vijay, and Jeongwan Haah.
\newblock ``Operator spreading in random unitary circuits''.
\newblock \href{https://dx.doi.org/10.1103/physrevx.8.021014}{Physical Review X{\bf 8}}~(2018).

\bibitem{gottesman1997stabilizer}
Daniel Gottesman.
\newblock ``Stabilizer codes and quantum error correction''~(1997).
\newblock  \href{http://arxiv.org/abs/quant-ph/9705052}{arXiv:quant-ph/9705052}.

\bibitem{fowler2012surface}
Austin~G. Fowler, Matteo Mariantoni, John~M. Martinis, and Andrew~N. Cleland.
\newblock ``Surface codes: Towards practical large-scale quantum computation''.
\newblock \href{https://dx.doi.org/10.1103/physreva.86.032324}{Physical Review A{\bf 86}}~(2012).

\bibitem{Nayak2008_z}
Chetan Nayak, Steven~H. Simon, Ady Stern, Michael Freedman, and Sankar Das~Sarma.
\newblock ``Non-abelian anyons and topological quantum computation''.
\newblock \href{https://dx.doi.org/10.1103/revmodphys.80.1083}{Reviews of Modern Physics {\bf 80}, 1083–1159}~(2008).

\bibitem{bravyi2019simulation}
Sergey Bravyi, Dan Browne, Padraic Calpin, Earl Campbell, David Gosset, and Mark Howard.
\newblock ``Simulation of quantum circuits by low-rank stabilizer decompositions''.
\newblock \href{https://dx.doi.org/10.22331/q-2019-09-02-181}{Quantum {\bf 3}, 181}~(2019).

\bibitem{bravyi2016improved}
Sergey Bravyi and David Gosset.
\newblock ``Improved classical simulation of quantum circuits dominated by clifford gates''.
\newblock \href{https://dx.doi.org/10.1103/physrevlett.116.250501}{Physical Review Letters{\bf 116}}~(2016).

\bibitem{beguvsic2023simulating}
Tomislav Begušić, Kasra Hejazi, and Garnet Kin-Lic Chan.
\newblock ``Simulating quantum circuit expectation values by clifford perturbation theory''.
\newblock \href{https://dx.doi.org/10.1063/5.0269149}{The Journal of Chemical Physics{\bf 162}}~(2025).

\bibitem{cheng2022clifford}
M.~H. Cheng, K.~E. Khosla, C.~N. Self, M.~Lin, B.~X. Li, A.~C. Medina, and M.~S. Kim.
\newblock ``Clifford circuit initialization for variational quantum algorithms''.
\newblock \href{https://dx.doi.org/10.1103/physreva.111.062413}{Physical Review A{\bf 111}}~(2025).

\bibitem{Zhang2021d_z}
Shi‐Xin Zhang, Zhou‐Quan Wan, Chang‐Yu Hsieh, Hong Yao, and Shengyu Zhang.
\newblock ``Variational quantum‐neural hybrid error mitigation''.
\newblock \href{https://dx.doi.org/10.1002/qute.202300147}{Advanced Quantum Technologies{\bf 6}}~(2023).

\bibitem{sun2024toward}
Jiace Sun, Lixue Cheng, and Weitang Li.
\newblock ``Toward chemical accuracy with shallow quantum circuits: A {Clifford}-based {Hamiltonian} engineering approach''.
\newblock J. Chem. Theory Comput. {\bf 20}, 695--707~(2024).
\newblock  url:~\url{https://doi.org/10.1021/acs.jctc.3c00886}.

\bibitem{mishmash2023hierarchical}
Ryan~V. Mishmash, Tanvi~P. Gujarati, Mario Motta, Huanchen Zhai, Garnet Kin-Lic Chan, and Antonio Mezzacapo.
\newblock ``Hierarchical clifford transformations to reduce entanglement in quantum chemistry wave functions''.
\newblock \href{https://dx.doi.org/10.1021/acs.jctc.3c00228}{Journal of Chemical Theory and Computation {\bf 19}, 3194–3208}~(2023).

\bibitem{schleich2023partitioning}
Philipp Schleich, Joseph Boen, Lukasz Cincio, Abhinav Anand, Jakob~S. Kottmann, Sergei Tretiak, Pavel~A. Dub, and Alán Aspuru-Guzik.
\newblock ``Partitioning quantum chemistry simulations with clifford circuits''.
\newblock \href{https://dx.doi.org/10.1021/acs.jctc.3c00335}{Journal of Chemical Theory and Computation {\bf 19}, 4952–4964}~(2023).

\bibitem{CAFQA}
Gokul~Subramanian Ravi, Pranav Gokhale, Yi~Ding, William Kirby, Kaitlin Smith, Jonathan~M. Baker, Peter~J. Love, Henry Hoffmann, Kenneth~R. Brown, and Frederic~T. Chong.
\newblock ``Cafqa: A classical simulation bootstrap for variational quantum algorithms''.
\newblock In Proceedings of the 28th ACM International Conference on Architectural Support for Programming Languages and Operating Systems, Volume 1.
\newblock \href{https://dx.doi.org/10.1145/3567955.3567958}{Page 15–29}.
\newblock ASPLOS ’23. ACM~(2022).

\bibitem{ising_stab_gs}
Bikrant Bhattacharyya and Gokul~Subramanian Ravi.
\newblock ``Optimal clifford initial states for ising hamiltonians''.
\newblock In 2023 IEEE International Conference on Rebooting Computing (ICRC).
\newblock \href{https://dx.doi.org/10.1109/icrc60800.2023.10386971}{Page 1–10}.
\newblock IEEE~(2023).

\bibitem{qian2024augmenting}
Xiangjian Qian, Jiale Huang, and Mingpu Qin.
\newblock ``Augmenting density matrix renormalization group with clifford circuits''.
\newblock \href{https://dx.doi.org/10.1103/physrevlett.133.190402}{Physical Review Letters{\bf 133}}~(2024).

\bibitem{masot2024stabilizer}
Sergi Masot-Llima and Artur Garcia-Saez.
\newblock ``Stabilizer tensor networks: Universal quantum simulator on a basis of stabilizer states''.
\newblock \href{https://dx.doi.org/10.1103/physrevlett.133.230601}{Physical Review Letters{\bf 133}}~(2024).

\bibitem{jeevanesan2024quantum}
Anonymous.
\newblock ``Quantum monte carlo and stabilizer states''.
\newblock \href{https://dx.doi.org/10.1103/8jmy-sj3f}{Physical Review A}~(2025).

\bibitem{mitarai2022quadratic}
Kosuke Mitarai, Yasunari Suzuki, Wataru Mizukami, Yuya~O. Nakagawa, and Keisuke Fujii.
\newblock ``Quadratic clifford expansion for efficient benchmarking and initialization of variational quantum algorithms''.
\newblock \href{https://dx.doi.org/10.1103/physrevresearch.4.033012}{Physical Review Research{\bf 4}}~(2022).

\bibitem{number_stab_groups}
D.~Gross.
\newblock ``Hudson’s theorem for finite-dimensional quantum systems''.
\newblock \href{https://dx.doi.org/10.1063/1.2393152}{Journal of Mathematical Physics{\bf 47}}~(2006).

\bibitem{kirby2021contextual}
William~M. Kirby, Andrew Tranter, and Peter~J. Love.
\newblock ``Contextual subspace variational quantum eigensolver''.
\newblock \href{https://dx.doi.org/10.22331/q-2021-05-14-456}{Quantum {\bf 5}, 456}~(2021).

\bibitem{bloch1929quantenmechanik}
Felix Bloch.
\newblock ``Uber die quantenmechanik der elektronen in kristallgittern''.
\newblock \href{https://dx.doi.org/10.1007/bf01339455}{Zeitschrift für Physik {\bf 52}, 555–600}~(1929).

\bibitem{kratzer2019basics}
Peter Kratzer and Jörg Neugebauer.
\newblock ``The basics of electronic structure theory for periodic systems''.
\newblock \href{https://dx.doi.org/10.3389/fchem.2019.00106}{Frontiers in Chemistry{\bf 7}}~(2019).

\bibitem{martin2020electronic}
Richard~M. Martin.
\newblock ``Electronic structure: Basic theory and practical methods''.
\newblock \href{https://dx.doi.org/10.1017/9781108555586}{Cambridge University Press}. ~(2020).

\bibitem{gharibian2012approximation}
Sevag Gharibian and Julia Kempe.
\newblock ``Approximation algorithms for qma-complete problems''.
\newblock In 2011 IEEE 26th Annual Conference on Computational Complexity.
\newblock \href{https://dx.doi.org/10.1109/ccc.2011.15}{Page 178–188}.
\newblock IEEE~(2011).

\bibitem{anshu2020beyond}
Anurag Anshu, David Gosset, and Karen Morenz.
\newblock ``Beyond product state approximations for a quantum analogue of max cut''~(2020).
\newblock \href{https://doi.org/10.4230/LIPIcs.TQC.2020.7}{10.4230/LIPIcs.TQC.2020.7}.

\bibitem{preskill2018quantum}
John Preskill.
\newblock ``Quantum computing in the nisq era and beyond''.
\newblock \href{https://dx.doi.org/10.22331/q-2018-08-06-79}{Quantum {\bf 2}, 79}~(2018).

\bibitem{shor1996fault}
P.W. Shor.
\newblock ``Fault-tolerant quantum computation''.
\newblock In Proceedings of 37th Conference on Foundations of Computer Science.
\newblock \href{https://dx.doi.org/10.1109/SFCS.1996.548464}{Pages 56--65}.
\newblock ~(1996).

\bibitem{cerezo2021cost}
M.~Cerezo, Akira Sone, Tyler Volkoff, Lukasz Cincio, and Patrick~J. Coles.
\newblock ``Cost function dependent barren plateaus in shallow parametrized quantum circuits''.
\newblock \href{https://dx.doi.org/10.1038/s41467-021-21728-w}{Nature Communications{\bf 12}}~(2021).

\bibitem{mcclean2018barren}
Jarrod~R. McClean, Sergio Boixo, Vadim~N. Smelyanskiy, Ryan Babbush, and Hartmut Neven.
\newblock ``Barren plateaus in quantum neural network training landscapes''.
\newblock \href{https://dx.doi.org/10.1038/s41467-018-07090-4}{Nature Communications{\bf 9}}~(2018).

\bibitem{Zhang2023b_z}
Hao-Kai Zhang, Shuo Liu, and Shi-Xin Zhang.
\newblock ``Absence of barren plateaus in finite local-depth circuits with long-range entanglement''.
\newblock \href{https://dx.doi.org/10.1103/physrevlett.132.150603}{Physical Review Letters{\bf 132}}~(2024).

\bibitem{robin2025stabilizer}
Caroline E.~P. Robin.
\newblock ``Stabilizer-accelerated quantum many-body ground-state estimation''~(2025).
\newblock  \href{http://arxiv.org/abs/2505.02923}{arXiv:2505.02923}.

\bibitem{schon2005sequential}
C.~Schön, E.~Solano, F.~Verstraete, J.~I. Cirac, and M.~M. Wolf.
\newblock ``Sequential generation of entangled multiqubit states''.
\newblock \href{https://dx.doi.org/10.1103/physrevlett.95.110503}{Physical Review Letters{\bf 95}}~(2005).

\bibitem{barenco1995elementary}
Adriano Barenco, Charles~H. Bennett, Richard Cleve, David~P. DiVincenzo, Norman Margolus, Peter Shor, Tycho Sleator, John~A. Smolin, and Harald Weinfurter.
\newblock ``Elementary gates for quantum computation''.
\newblock \href{https://dx.doi.org/10.1103/physreva.52.3457}{Physical Review A {\bf 52}, 3457–3467}~(1995).

\bibitem{malz2024preparation}
Daniel Malz, Georgios Styliaris, Zhi-Yuan Wei, and J.~Ignacio Cirac.
\newblock ``Preparation of matrix product states with log-depth quantum circuits''.
\newblock \href{https://dx.doi.org/10.1103/physrevlett.132.040404}{Physical Review Letters{\bf 132}}~(2024).

\bibitem{gu2024doped}
Andi Gu, Salvatore F.~E. Oliviero, and Lorenzo Leone.
\newblock ``Doped stabilizer states in many-body physics and where to find them''.
\newblock \href{https://dx.doi.org/10.1103/physreva.110.062427}{Physical Review A{\bf 110}}~(2024).

\bibitem{koenig2014efficiently}
Robert Koenig and John~A. Smolin.
\newblock ``How to efficiently select an arbitrary clifford group element''.
\newblock \href{https://dx.doi.org/10.1063/1.4903507}{Journal of Mathematical Physics{\bf 55}}~(2014).

\bibitem{PT_example}
Ruben Verresen, Roderich Moessner, and Frank Pollmann.
\newblock ``One-dimensional symmetry protected topological phases and their transitions''.
\newblock \href{https://dx.doi.org/10.1103/physrevb.96.165124}{Physical Review B{\bf 96}}~(2017).

\bibitem{toric}
Fengcheng Wu, Youjin Deng, and Nikolay Prokof’ev.
\newblock ``Phase diagram of the toric code model in a parallel magnetic field''.
\newblock \href{https://dx.doi.org/10.1103/physrevb.85.195104}{Physical Review B{\bf 85}}~(2012).

\bibitem{shang2023schrodinger}
Zhong-Xia Shang, Ming-Cheng Chen, Xiao Yuan, Chao-Yang Lu, and Jian-Wei Pan.
\newblock ``Schr{\"o}dinger-{Heisenberg} variational quantum algorithms''.
\newblock Phys. Rev. Lett. {\bf 131}, 060406~(2023).
\newblock  url:~\url{https://doi.org/10.1103/PhysRevLett.131.060406}.

\bibitem{ibm_eagle}
Youngseok Kim, Andrew Eddins, Sajant Anand, Ken~Xuan Wei, Ewout van~den Berg, Sami Rosenblatt, Hasan Nayfeh, Yantao Wu, Michael Zaletel, Kristan Temme, and Abhinav Kandala.
\newblock ``Evidence for the utility of quantum computing before fault tolerance''.
\newblock \href{https://dx.doi.org/10.1038/s41586-023-06096-3}{Nature {\bf 618}, 500–505}~(2023).

\bibitem{sycamore}
Frank Arute, Kunal Arya, Ryan Babbush, Dave Bacon, Joseph~C. Bardin, Rami Barends, Rupak Biswas, Sergio Boixo, Fernando G. S.~L. Brandao, David~A. Buell, Brian Burkett, Yu~Chen, Zijun Chen, Ben Chiaro, Roberto Collins, William Courtney, Andrew Dunsworth, Edward Farhi, Brooks Foxen, Austin Fowler, Craig Gidney, Marissa Giustina, Rob Graff, Keith Guerin, Steve Habegger, Matthew~P. Harrigan, Michael~J. Hartmann, Alan Ho, Markus Hoffmann, Trent Huang, Travis~S. Humble, Sergei~V. Isakov, Evan Jeffrey, Zhang Jiang, Dvir Kafri, Kostyantyn Kechedzhi, Julian Kelly, Paul~V. Klimov, Sergey Knysh, Alexander Korotkov, Fedor Kostritsa, David Landhuis, Mike Lindmark, Erik Lucero, Dmitry Lyakh, Salvatore Mandrà, Jarrod~R. McClean, Matthew McEwen, Anthony Megrant, Xiao Mi, Kristel Michielsen, Masoud Mohseni, Josh Mutus, Ofer Naaman, Matthew Neeley, Charles Neill, Murphy~Yuezhen Niu, Eric Ostby, Andre Petukhov, John~C. Platt, Chris Quintana, Eleanor~G. Rieffel, Pedram Roushan, Nicholas~C. Rubin, Daniel Sank, Kevin~J.
  Satzinger, Vadim Smelyanskiy, Kevin~J. Sung, Matthew~D. Trevithick, Amit Vainsencher, Benjamin Villalonga, Theodore White, Z.~Jamie Yao, Ping Yeh, Adam Zalcman, Hartmut Neven, and John~M. Martinis.
\newblock ``Quantum supremacy using a programmable superconducting processor''.
\newblock \href{https://dx.doi.org/10.1038/s41586-019-1666-5}{Nature {\bf 574}, 505–510}~(2019).

\bibitem{wiersema2020exploring}
Roeland Wiersema, Cunlu Zhou, Yvette de~Sereville, Juan~Felipe Carrasquilla, Yong~Baek Kim, and Henry Yuen.
\newblock ``Exploring entanglement and optimization within the hamiltonian variational ansatz''.
\newblock \href{https://dx.doi.org/10.1103/prxquantum.1.020319}{PRX Quantum{\bf 1}}~(2020).

\bibitem{Zhang2022_z}
Shi-Xin Zhang, Jonathan Allcock, Zhou-Quan Wan, Shuo Liu, Jiace Sun, Hao Yu, Xing-Han Yang, Jiezhong Qiu, Zhaofeng Ye, Yu-Qin Chen, Chee-Kong Lee, Yi-Cong Zheng, Shao-Kai Jian, Hong Yao, Chang-Yu Hsieh, and Shengyu Zhang.
\newblock ``Tensorcircuit: a quantum software framework for the nisq era''.
\newblock \href{https://dx.doi.org/10.22331/q-2023-02-02-912}{Quantum {\bf 7}, 912}~(2023).

\bibitem{scipy_LBFGS_B}
Richard~H. Byrd, Peihuang Lu, Jorge Nocedal, and Ciyou Zhu.
\newblock ``A limited memory algorithm for bound constrained optimization''.
\newblock \href{https://dx.doi.org/10.1137/0916069}{SIAM Journal on Scientific Computing {\bf 16}, 1190–1208}~(1995).

\bibitem{2020SciPy-NMeth}
Pauli Virtanen, Ralf Gommers, Travis~E. Oliphant, Matt Haberland, Tyler Reddy, David Cournapeau, Evgeni Burovski, Pearu Peterson, Warren Weckesser, Jonathan Bright, Stéfan~J. van~der Walt, Matthew Brett, Joshua Wilson, K.~Jarrod Millman, Nikolay Mayorov, Andrew R.~J. Nelson, Eric Jones, Robert Kern, Eric Larson, C~J Carey, İlhan Polat, Yu~Feng, Eric~W. Moore, Jake VanderPlas, Denis Laxalde, Josef Perktold, Robert Cimrman, Ian Henriksen, E.~A. Quintero, Charles~R. Harris, Anne~M. Archibald, Antônio~H. Ribeiro, Fabian Pedregosa, Paul van Mulbregt, Aditya Vijaykumar, Alessandro~Pietro Bardelli, Alex Rothberg, Andreas Hilboll, Andreas Kloeckner, Anthony Scopatz, Antony Lee, Ariel Rokem, C.~Nathan Woods, Chad Fulton, Charles Masson, Christian Häggström, Clark Fitzgerald, David~A. Nicholson, David~R. Hagen, Dmitrii~V. Pasechnik, Emanuele Olivetti, Eric Martin, Eric Wieser, Fabrice Silva, Felix Lenders, Florian Wilhelm, G.~Young, Gavin~A. Price, Gert-Ludwig Ingold, Gregory~E. Allen, Gregory~R. Lee, Hervé
  Audren, Irvin Probst, Jörg~P. Dietrich, Jacob Silterra, James~T Webber, Janko Slavič, Joel Nothman, Johannes Buchner, Johannes Kulick, Johannes~L. Schönberger, José~Vinícius de~Miranda~Cardoso, Joscha Reimer, Joseph Harrington, Juan Luis~Cano Rodríguez, Juan Nunez-Iglesias, Justin Kuczynski, Kevin Tritz, Martin Thoma, Matthew Newville, Matthias Kümmerer, Maximilian Bolingbroke, Michael Tartre, Mikhail Pak, Nathaniel~J. Smith, Nikolai Nowaczyk, Nikolay Shebanov, Oleksandr Pavlyk, Per~A. Brodtkorb, Perry Lee, Robert~T. McGibbon, Roman Feldbauer, Sam Lewis, Sam Tygier, Scott Sievert, Sebastiano Vigna, Stefan Peterson, Surhud More, Tadeusz Pudlik, Takuya Oshima, Thomas~J. Pingel, Thomas~P. Robitaille, Thomas Spura, Thouis~R. Jones, Tim Cera, Tim Leslie, Tiziano Zito, Tom Krauss, Utkarsh Upadhyay, Yaroslav~O. Halchenko, and Yoshiki Vázquez-Baeza.
\newblock ``Scipy 1.0: fundamental algorithms for scientific computing in python''.
\newblock \href{https://dx.doi.org/10.1038/s41592-019-0686-2}{Nature Methods {\bf 17}, 261–272}~(2020).

\bibitem{Barahona1982a_z}
F~Barahona.
\newblock ``On the computational complexity of ising spin glass models''.
\newblock \href{https://dx.doi.org/10.1088/0305-4470/15/10/028}{Journal of Physics A: Mathematical and General {\bf 15}, 3241–3253}~(1982).

\bibitem{zhang2019classification}
Shi-Xin Zhang.
\newblock ``Classification on the computational complexity of spin models''~(2019).
\newblock  \href{http://arxiv.org/abs/1911.04122}{arXiv:1911.04122}.

\bibitem{Zhang2020b_z}
Shi-Xin Zhang, Chang-Yu Hsieh, Shengyu Zhang, and Hong Yao.
\newblock ``Differentiable quantum architecture search''.
\newblock \href{https://dx.doi.org/10.1088/2058-9565/ac87cd}{Quantum Science and Technology {\bf 7}, 045023}~(2022).

\bibitem{gu2023zero}
Andi Gu, Hong-Ye Hu, Di~Luo, Taylor~L. Patti, Nicholas~C. Rubin, and Susanne~F. Yelin.
\newblock ``Zero and finite temperature quantum simulations powered by quantum magic''.
\newblock \href{https://dx.doi.org/10.22331/q-2024-07-23-1422}{Quantum {\bf 8}, 1422}~(2024).

\bibitem{munoz2023low}
Manuel~H. Muñoz-Arias, Stefanos Kourtis, and Alexandre Blais.
\newblock ``Low-depth clifford circuits approximately solve maxcut''.
\newblock \href{https://dx.doi.org/10.1103/physrevresearch.6.023294}{Physical Review Research{\bf 6}}~(2024).

\end{thebibliography}
\clearpage

\newcommand{\beginsupplement}{%
        \setcounter{table}{0}
        \renewcommand{\thetable}{S\arabic{table}}%
        \setcounter{figure}{0}
        \renewcommand{\thefigure}{S\arabic{figure}}%
     }



\end{document}